\newif\iffullversion
\fullversiontrue
\iffullversion 
\documentclass[11pt,letterpaper]{article}
\usepackage[paper=letterpaper,margin=1in]{geometry}
\usepackage[margin=1in]{geometry}
\usepackage{amsthm}

\else
\documentclass[envcountsect,envcountsame,orivec,runningheads]{llncs}

\usepackage{amsthm}
\fi

\usepackage{cite}

\usepackage[bookmarks=true,pdfstartview=FitH,colorlinks,linkcolor=darkred,filecolor=darkred,citecolor=darkred,urlcolor=darkred]{hyperref}

\usepackage{multicol}

\usepackage{colortbl}
\usepackage{mdframed}
\usepackage{color,soul}
\usepackage{braket}
\usepackage{paralist}
\usepackage{pdfpages}
\usepackage{enumitem}
\usepackage{float}
\usepackage{booktabs}
\usepackage[override]{cmtt}
\usepackage{color,xcolor}
\usepackage{authblk}
\usepackage{graphicx,color,eso-pic}
\usepackage{color}
\usepackage{xcolor}

\usepackage[
	lambda,
	advantage,
	operators,
	sets,
	adversary,
	landau,
	probability,
	notions,
	logic,
	ff,
	mm,
	primitives,
	events,
	complexity,
	asymptotics,
	keys]
{cryptocode}

\usepackage{tikz}
\usepackage{pgfplots}
\usepgfplotslibrary{fillbetween}
\usetikzlibrary{fit, shapes, positioning, patterns, arrows,shadows}
\usetikzlibrary{decorations.pathreplacing,calligraphy}

\tikzset{
    master/.style={
        execute at end picture={
            \coordinate (lower right) at (current bounding box.south east);
            \coordinate (upper left) at (current bounding box.north west);
        }
    },
    slave/.style={
        execute at end picture={
            \pgfresetboundingbox
            \path (upper left) rectangle (lower right);
        }
    }
}

\sethlcolor{lightgray}

\usepackage{stmaryrd}

\usepackage{amsmath}

\newcommand{\mcal}[1]{\ensuremath{\mathcal {#1}}}

\newcommand{\algC}{{\color{black}\ensuremath{\mcal{C}}}\xspace}

\newcommand{\algR}{{\ensuremath{\mcal{R}}}\xspace}
\newcommand{\algA}{{\ensuremath{\mcal{A}}}\xspace}
\newcommand{\algB}{{\ensuremath{\mcal{B}}}\xspace}

\newcommand{\A}{{\cal A}}

\renewcommand{\P}{{\cal P}}

\renewcommand{\ppt}{\ensuremath{{\text{PPT}}}\xspace}
\renewcommand{\poly}{\ensuremath{{{\sf poly}}}\xspace}

\newcommand{\E}{\mathbf{E}}

\newcommand{\eps}{\epsilon}

\newcommand{\coin}{\ensuremath{{\sf coin}}\xspace}
\newcommand{\digest}{\ensuremath{{\sf digest}}\xspace}
\newcommand{\Gen}{\ensuremath{{\sf Gen}}\xspace}
\newcommand{\Prove}{\ensuremath{{\sf P}}\xspace}
\newcommand{\Ver}{\ensuremath{{\sf V}}\xspace}
\newcommand{\stmt}{\ensuremath{{\sf stmt}}\xspace}

\newcommand{\Real}{\ensuremath{{\rm real}}\xspace}
\newcommand{\Ideal}{\ensuremath{{\rm ideal}}\xspace}

\newcommand{\Hyb}{\ensuremath{{\sf Hyb}}\xspace}

\newcommand{\Comm}{\ensuremath{{\sf Comm}}\xspace}

\newcommand{\AoK}{\ensuremath{{\sf AoK}}\xspace}

\newcommand{\ek}{\ensuremath{{\sf ek}}\xspace}

\definecolor{darkgreen}{rgb}{0,0.5,0}
\definecolor{lightblue}{RGB}{0,176,240}
\definecolor{darkblue}{RGB}{0,112,192}
\definecolor{lightpurple}{RGB}{124, 66, 168}
\definecolor{grey}{RGB}{139, 137, 137}
\definecolor{maroon}{RGB}{178, 34, 34}
\definecolor{green}{RGB}{34, 139, 34}
\definecolor{types}{RGB}{72, 61, 139}
\definecolor{gold}{rgb}{0.8, 0.33, 0.0}
\definecolor{mygreen}{HTML}{588D6A}
\definecolor{myred}{HTML}{C86733}
\definecolor{myblue}{HTML}{5B68FF}
\definecolor{myshadow}{HTML}{E6C5B4}

\definecolor{darkgray}{gray}{0.3}

\newcommand{\skiptext}[1]{}
\newcommand{\msg}{\ensuremath{{{\sf msg}}}\xspace}

\newcommand{\crs}{\ensuremath{{{\sf crs}}}\xspace}

\newcommand{\getr}{\ensuremath{{\overset{\$}{\leftarrow}}}\xspace}

\newcommand{\secu}{\lambda}
\newcommand{\inDom}{\ensuremath{\mathbb{U}}\xspace}

\newcommand{\inDomSecu}{\ensuremath{\mathbb{U}_\secu}\xspace}

\usepackage{boxedminipage}
\usepackage{url}
\usepackage{graphicx}
\usepackage[bf,justification=centering]{caption}
\usepackage[justification=centering]{subcaption}
\usepackage{color}
\usepackage{xspace}
\usepackage{multirow}

\usepackage{amsthm,amstext}
\usepackage{amsmath,amstext,amsfonts,amssymb,latexsym}
\usepackage{bbm}
\usepackage{amsbsy}
\usepackage{stmaryrd}
\usepackage{pifont}

\usepackage[capitalize,noabbrev]{cleveref} 

\usepackage{wrapfig}
\definecolor{darkred}{rgb}{0.5, 0, 0}
\definecolor{darkgreen}{rgb}{0, 0.5, 0}
\definecolor{darkblue}{rgb}{0,0,0.5}

\newcommand\markx[2]{}

\newcommand{\rv}[1]{#1\xspace}

\renewcommand{\path}{\ensuremath{\mathsf{path}}\xspace}

\newcommand{\aux}{{\sf aux}}

\renewcommand{\negl}{{\sf negl}}

\newcommand{\N}{\ensuremath{\mathbb{N}}\xspace}

\newcommand{\ignore}[1]{}

\renewcommand{\part}{\ensuremath{\mcal{S}}\xspace}

\newcounter{task}

\newtheorem{thm}{Theorem}[section]      

\newtheorem{theorem}[thm]{Theorem}

\newtheorem{conjecture}[thm]{Conjecture}
\newtheorem{lemma}[thm]{Lemma}
\newtheorem{claim}[thm]{Claim}
\newtheorem{corollary}[thm]{Corollary}
\newtheorem{fact}[thm]{Fact}
\newtheorem{proposition}[thm]{Proposition}

\theoremstyle{definition}
\newtheorem{definition}[thm]{Definition}
\newtheorem{remark}[thm]{Remark}
\newtheorem{example}[thm]{Example}

\newtheoremstyle{boxes}
{2pt}
{0pt}
{}
{}
{\bfseries}
{}
{\newline}
{\thmname{#1}\thmnumber{ #2}:  
\thmnote{#3}}

\ifdefined\authorNote
\newcommand{\elaine}[1]{{\footnotesize\color{magenta}[Elaine: #1]}}
\newcommand{\hao}[1]{{\footnotesize\color{blue}[Hao: #1]}}
\newcommand{\ke}[1]{{\footnotesize\color{darkgreen}[Ke: #1]}}
\else
\newcommand{\elaine}[1]{}
\newcommand{\hao}[1]{}
\newcommand{\ke}[1]{}
\fi

\newcounter{cnt:challenge}

\newcommand{\bfb}{{\bf b}}
\newcommand{\bfc}{{\bf c}}

\newcommand{\bfp}{{\bf p}}

\newcommand{\bfv}{{\bf v}}
\newcommand{\bfx}{{\bf x}}

\newcommand{\util}{{\sf util}}
\newcommand{\pbutil}{{\sf pbutil}}

\newcommand{\platform}{\mcal{P}}
\newcommand{\seller}{\mcal{S}}

\newcommand{\Fauction}{\mcal{F}_{\sf auction}}

\newcommand{\ok}{\ensuremath{{\sf ok}}\xspace}

\newcommand{\acc}{\ensuremath{{\sf acc}}\xspace}

\newcommand{\rej}{\ensuremath{{\sf rej}}\xspace}

\newcommand{\V}{\ensuremath{\mathcal{V}}}

\title{Foundations of Platform-Assisted Auctions}

\iffullversion

\author[2]{Hao Chung}
\author[3]{Ke Wu}
\author[1]{Elaine Shi\thanks{Author ordering is randomized}}
\affil[1]{CMU, Oblivious Labs Inc., and 0xPARC}
\affil[2]{CMU}
\affil[3]{CMU and UMich}

\date{}

\else
\author{}
\date{}
\institute{}
\fi

\begin{document}

\iffullversion
\begin{titlepage}
\maketitle
\begin{abstract}

Today, many auctions are carried out
with the help of intermediary platforms like Google and eBay.
These platforms 
serve as a 
rendezvous point for the 
buyers and sellers, and charge a fee for its service.
We refer to such auctions as \emph{platform-assisted auctions}.
Traditionally, the auction theory literature mainly focuses on designing auctions that incentivize
the buyers to bid truthfully,
assuming that the platform always faithfully implements the auction.
In practice, however, the platforms have been found to  
manipulate the auctions to earn more profit, resulting in high-profile 
anti-trust lawsuits. 

We propose a new model 
for studying platform-assisted auctions in the permissionless setting,
where anyone can register and participate in the auction.
We explore whether it is possible to design a dream auction in this
new model, such that honest behavior  
is the utility-maximizing 
strategy for each individual buyer, the platform, the seller, 
as well as platform-seller or platform-buyer coalitions.
Through a collection of feasibility and infeasibility results,
we carefully characterize the mathematical landscape of  
platform-assisted auctions. 

Interestingly, our work reveals exciting connections
between cryptography and mechanism design. 
We show how cryptography can lend to the design of an efficient 
platform-assisted auction with dream properties. 
Although a line of works have also used  
multi-party computation (MPC) or the blockchain to remove 
the reliance on a trusted auctioneer, 
our work is distinct in nature in several dimensions.
First, we initiate a systematic exploration of  
the game theoretic implications when the service providers
(e.g., nodes that implement the MPC or blockchain protocol) 
are strategic and can collude with sellers or buyers.  
Second,  
we observe that the full simulation paradigm
used in the standard MPC literature is too stringent 
and leads to high asymptotical costs.  
Specifically, 
because every player has a different private outcome
in an auction protocol, to the best of our knowledge, 
running any generic MPC 
protocol among the players would incur at least $n^2$ total cost 
where $n$ is the number of buyers.
We propose a new notion of simulation called
{\it utility-dominated emulation} that is sufficient
for guaranteeing the game-theoretic properties
needed in an auction. 
Under this new notion of simulation, we show
how to design efficient auction protocols 
with quasilinear efficiency, which gives an $n$-fold improvement over 
any generic approach.

\end{abstract}
\thispagestyle{empty}
\end{titlepage}

\else
\maketitle
\begin{abstract}

\end{abstract}
\fi

\iffullversion
\tableofcontents
\clearpage
\fi

\section{Introduction}
\label{sec:intro}

Traditionally, auctions are conducted in physical spaces
where an auctioneer serves as a trusted
intermediary who oversees the process. 
Today, however,  many auctions 
are run by online platforms such as Google Ad Exchange
or eBay, which not only serve
as an intermediary and a rendezvous point
between sellers and buyers,  
but also provide value-added
services to both sellers and buyers. 
Typically, all players
agree on some {\it fee structure}
such that the platform gets remunerated for its service.

The new status quo of online {\it platform-assisted auctions}
raises new challenges for the auction designer.
Traditionally, the auction theory literature~\cite{myerson,agt} 
focuses on designing auctions that incentivize the 
buyers to bid honestly 
(e.g., the good old second-price auction~\cite{vickrey});
and we take it for granted that the auctioneer
will always implement the auction's rules honestly. 
With platform-assisted auctions,  however, 
this ``fully trusted auctioneer'' assumption is no longer always true.  
Specifically, the platform  
can manipulate the auction to earn more profit.   
For example, 
if we are running a second-price
auction, the platform may inflate the second price
to a buyer to make it pay more; and deflate the second price 
to the seller to withhold 
some profits from the seller. 
Such platform deviations have been observed  
in the real world.
In 2023, the US Department of Justice filed an 
antitrust suit against
Google for manipulating its Ad Exchange auctions 
which harms users and monopolizes the market~\cite{doj,google-dishonest}. 
\ignore{
Partly due to the growing concerns about platform deviations,
some online auction platforms
such as Google Ad Manager are reverting back
to first-price auctions, i.e., they give  
up on achieving {\it buyer incentive compatibility}
in exchange for {\it platform incentive compatibility}. 
This way, following the honest rules 
is the profit-maximizing strategy for the platform, 
which protects the platform from antitrust lawsuits.
}

Further, platform deviations are exacerbated by the fact that 
modern platform-assisted  
auctions are often run in a 
{\it permissionless} setting.  
Specifically, the identities of the buyers
are unknown a-priori, and anyone (including
the platform or the seller) may enter the auction 
with fake identities and inject fake bids
(also known as shill bidding~\cite{shillbid}).
By contrast, 
the vast majority of works
in the traditional auction theory literature
focus on {\it permissioned} auctions 
where the identities of the buyers are known a-priori. 
The permissionless nature of platform-assisted
auctions broadens the strategy space  
and is a new challenge we need to tackle in designing
modern platform-assisted auctions. 

\paragraph{Can we have a dream platform-assisted auction?}
In this paper, we want to rethink the design of platform-assisted 
auctions in light 
of these new requirements.  
We consider a scenario where 
a seller wants to sell a finite number of $k$ identical
items to buyers, assisted by a platform. 
We ask what is a {\it dream} platform-assisted auction, 
and whether we can design a platform-assisted auction with
such dream properties. 
Our work systematically explores what kind of fundamental tradeoffs
we face when designing platform-assisted auctions, and
illustrates how platform-assisted auctions can benefit from
cryptographic tools.

We argue that a dream platform-assisted auction 
should satisfy the following properties:

\begin{itemize}[leftmargin=5mm,itemsep=1pt]
\item 
{\it Incentive compatibility for $X \in \{\text{buyer, platform, seller}\}$.}
We want that 
following the prescribed protocol  
honestly 
maximizes the utility for any individual buyer
(or the platform, seller respectively). 
Henceforth, we
use the shorthands {\it bIC, pIC, and sIC} to refer
to buyer, platform, and seller incentive compatibility, respectively. 
\item 
{\it Strategy proof 
when the platform colludes with $X \in \{\textit{buyer,
seller}\}$}.
To make sure that 
the platform cannot become a monopoly, 
we not only need platform incentive compatibility (pIC),
we also want to make sure that strategic
deviations do not help the platform
even when it 
may collude with some of the buyers or the seller. 
We say that a platform-assisted
auction satisfies 
{\it platform-seller incentive compatibility (psIC)}
or 
{\it $c$-platform-buyer incentive compatibility ($c$-pbIC)}
iff honest behavior  
is profit-maximizing for the platform-seller coalition,
or any platform-buyer coalition that consists of 
at least one and at most $c$ buyers, 
respectively.
\end{itemize}
\elaine{what about platform-seller-buyer coalition}

Although there are some other collusion-resilience properties
one can formulate, e.g., resilience against buyer-seller coalitions,
our formulation prioritizes properties that 
mathematically capture notions of {\it anti-trust} against the {\it platform}.
Therefore, we believe that defending against coalitions centered 
around the platform is our top concern. 

\subsection*{Our Results and Contributions}

We reveal the mathematical structure
of platform-assisted  
auctions through a combination of novel feasibility
and infeasibility results. 
Our results also reveal new connections between 
mechanism design and cryptography. 
On one hand, we show how cryptography can lend to the design
of platform-assisted auctions; and on the 
other hand,
we show that since 
mechanism design asks only for game-theoretic properties, 
we can avoid using 
generic cryptographic primitives such as multi-party computation,
and instead opt for 
weaker notions of ``simulation'' that allow
us to achieve asymptotical performance improvement over
known generic approaches.  

We now give an overview of our results and contributions.

\subsection{New Model for Platform-Assisted Auctions}

We formulate a new model and lay the groundwork 
for formally reasoning about platform-assisted auctions. 
We assume that each buyer and the seller communicates
with the platform through a pairwise private channel. 
There is no buyer-buyer or buyer-seller communication. 
Besides exchanging pairwise messages between the platform
and the buyers/seller, 
all players are additionally allowed
to post messages to a broadcast channel\footnote{Note that a broadcast
channel with availability guarantees
cannot be directly realized from a star topology where
all players talk only to the platform.}, and everyone
can read messages posted to the broadcast channel. 
However, we would like our auctions to minimize
the usage of broadcast since  
it is an expensive operation (especially in a permissionless environment). 
Depending on which consensus protocol is used to realize
the broadcast, every broadcast message incurs at least quadratic
and possibly higher communication. 
Moreover, 
if the broadcast is instantiated with a blockchain, 
posting messages to the blockchain typically incurs transaction fees.

We stress that in certain auctions, 
some information about the auction (e.g., its outcomes)
naturally becomes publicly visible --- in such cases,
we may assume that a 
broadcast
channel exists for free 
for posting selected information about the auction. 
For example, consider an
ad auction for the Super Bowl --- in this case,
the winners of the auction are publicly announced 
since anyone can view the ads 
on natural television.

Henceforth, we often refer to public information 
posted on the blockchain (or revealed to the public through a social channel) 
the {\it public outcome} of the auction. 

\paragraph{Comparison with earlier models.}
Our model strictly generalizes earlier models  
such as the credible auction model~\cite{credibleauction} and 
blockchain-backed auctions~\cite{strain,riggs,crediblebroadcast,nftauction00,nftauction01}. 
\elaine{TODO: more cites, Chitra et al.?} 
Not only so, we avoid some important limitations of earlier models
as mentioned below.  
The elegant credible auction framework~\cite{credibleauction} 
also considers a setting
where buyers communicate with the platform through pairwise private channels. 
However, their model differs from ours in two important ways.
First, their model does not allow 
posting messages to a broadcast channel.
This makes some of their impossibility results 
overly pessimistic. In particular, as mentioned above, in some auctions, 
a social broadcast channel naturally presents itself, and makes  
(part of) the auction's outcome publicly visible. 
Their model and their impossibility results 
do not necessarily extend to such auctions. 
Second, they assume that the platform and the seller 
are the same entity. 
In other words, an auction that is considered incentive compatible 
in their framework
may still permit strategies 
where the platform cheats the seller to make additional profit ---
such deviations have been documented in real-life 
anti-trust lawsuits~\cite{google-dishonest}.

Various works~\cite{strain,riggs,crediblebroadcast,nftauction00,nftauction01} 
have considered blockchain-backed auctions
where all protocol messages are posted to a blockchain --- such 
auctions can be used for selling both digital goods~\cite{nftauction00,nftauction01}
as well as physical items (e.g., Taylor Swift concert tickets~\cite{taylorswift}).
Blockchain-backed auctions suffer from high cost both in terms
of bandwidth consumption (when the blockchain is actually realized
with a consensus protocol) and transaction fees. 

By separating the seller and the platform's roles, 
our modeling choices allows us to explicitly capture
platform strategies that involve cheating the seller, and 
we can ask natural questions such as what kind of fee structure
is possible between the platform and the seller. 
By basing our model on pairwise channels  
but additionally permitting broadcast messages, 
we can ask natural questions such as how to 
minimize the use of broadcast in the auction design.

\elaine{computationally sound notion of ``safe''}

\subsection{Inefficient Information-Theoretic Feasibility}
We show that the ascending auction
with reserve and fixed platform fees 
(detailed in \Cref{sec:ascending})
satisfies almost all desired
properties,
and meanwhile achieves approximate revenue optimality,
as stated in the following theorem.

\begin{theorem}[Informal: information-theoretic feasibility] 
The ascending auction with an appropriate reserve 
and fixed platform fees 
satisfies bIC, pIC, and 1-pbIC. Moreover,
if buyers' true values are drawn from a natural distribution, 
the mechanism additionally 
satisfies Bayesian sIC and Bayesian psIC, and is almost revenue-optimal. 
\label{thm:intro-ascend}
\end{theorem}

Despite these desirable properties, the ascending auction 
of \Cref{sec:ascending}
suffers from a few important limitations. 
\begin{enumerate}[leftmargin=7mm,itemsep=1pt]
\item 
The platform can only get a fixed fee
that is independent of the auction's revenue, and 
it would be nice to support other fee structures if possible.
\item 
The protocol 
requires that the platform post messages 
to a broadcast channel, which we ideally would like to avoid if possible.
\item 
The protocol satisfies only 1-pbIC, i.e., 
it provides incentive compatibility when the platform colludes
with at most one buyer, but not more. 
\item 
The protocol satisfies psIC and sIC only in the Bayesian sense, 
whereas 
we would ideally like to achieve the stronger
notion of ex post incentive compatibility if possible. 
Roughly speaking, {\it Bayesian} incentive compatibility 
assumes that the strategic players have some a-priori belief
about honest  buyers' 
true values, but they are allowed to adaptively 
adjust their actions as 
additional information potentially gets 
revealed during the protocol. 
By contrast, {\it ex post} incentive compatibility 
requires that the strategic players are still incentivized
to behave honestly even with full knowledge of honest buyers' true values.  
\item 
Finally, the protocol's round complexity is as large as 
the number of possible values in the (discretized) value domain. 
In practice, to have sufficient precision in encoding buyers' values,
we often choose the value domain to be exponentially large, thus
leading to  
an exponentially large round complexity. 
\end{enumerate}

This raises the natural question, {\it can we avoid the above limitations}?
We show that the first four limitations
are in fact inherent in some sense, 
whereas the last limitation can be avoided
through the use of cryptography.  
Below we elaborate on these results.

\subsection{Impossibility Results}
\label{sec:intro-imp}
We prove a collection of impossibility results 
that together uncover the necessary mathematical structure 
of platform-assisted auctions.

\paragraph{Necessity of fixed platform fees.}
We prove that any auction that simultaneously satisfies
bIC and 1-pbIC must 
impose a certain fee structure, i.e., the platform
must get fixed fees that are independent of the auction's revenue:

\begin{theorem}[Informal: bIC + 1-pbIC $\Longrightarrow$ fixed fee structure]
Any (possibly multi-round) auction that is bIC and 1-pbIC 
(in the ex post setting) must  
pay a fixed fee to the platform that is independent
of the auction's revenue
(even assuming the existence of a broadcast channel). 
The same fee structure restriction also holds for any mechanism that 
simulataneously satisfies 
bIC, pIC, and 1-pbIC in the Bayesian setting. 
Further, for computationally bounded players
the same fee structure restriction holds ignoring negligibly small differences. 
\label{thm:intro-0fee}
\end{theorem}

\Cref{thm:intro-0fee}
also explains why in the ascending auction of \Cref{thm:intro-ascend}, the platform
must get a fixed fee.

\paragraph{Necessity of broadcast.}
We prove an impossibility result 
which shows the significance of the broadcast channel: 

\begin{theorem}[Informal: bIC + 1-pbIC $\Longrightarrow$ broadcast necessary]
Suppose we do not allow  
posting to a broadcast channel; and further, either there is no
public-key infrastructure or the seller is offline
and does not send any messages
during the auction. 
Then, no (possibly multi-round) auction 
can simultaneously satisfy bIC and 1-pbIC in the ex post setting. 
Moreover, no auction can simultaneously satisfy Bayesian 
notions of bIC, pIC, and 1-pbIC. 
The impossibility holds even 
when the players are computationally bounded. 
\label{thm:intro-pubout}
\end{theorem}

This impossibility also implies that posting to the broadcast channel
is necessary in the ascending auction in \Cref{thm:intro-ascend}.

\paragraph{Impossibility of 2-pbIC.}
We prove that 
it is not possible to simultaneously satisfy
bIC and 2-pbIC, which explains
why the ascending auction of \Cref{thm:intro-ascend}
satisfies only 1-pbIC:

\begin{theorem}[Informal: bIC and 2-pbIC $\Longrightarrow$ impossible]
No auction can simultaneously 
satisfy bIC and 2-pbIC. Further, this impossibility
holds even for multi-round auctions, even for Bayesian notions
of bIC and 2-pbIC, even  
when allowing a broadcast channel, and even 
for computationally bounded players.  
\label{thm:intro-2pbic}
\end{theorem}

\paragraph{Impossibility of ex post psIC.}

We prove the following theorem 
which shows that simultaneously asking for (Bayesian) 
bIC and ex post psIC 
will severely constrain the design space
such that the resulting auction must have small revenue. 
This partly explains why the ascending auction of \Cref{thm:intro-ascend}
achieves only Bayesian psIC and Bayesian sIC. 

\begin{theorem}[Informal: (Bayesian) bIC + ex post psIC $\Longrightarrow$ impossible]
Any (possibly multi-round) auction 
that simultaneously satisfies information-theoretic 
(Bayesian) bIC and
ex post psIC must be revenue dominated by a posted price auction. 
Moreover, this restriction holds even when allowing broadcast.  
\label{thm:intro-expostpsic}
\end{theorem}

In the above, ``revenue dominated by posted price'' means 
that the auction's revenue is dominated by some posted price auction
for every value vector. 

\Cref{thm:intro-expostpsic}  
holds in the information theoretic setting, for ex post psIC, 
and for even multi-round protocols. 
The proof of this theorem can be modified to show that 
if we restrict ourselves to one-round protocols, then
the same impossibility holds even for Bayesian psIC and even for
computationally bounded players. 

\begin{corollary}[Informal: extension of \Cref{thm:intro-expostpsic}]
Any $1$-round auction that simultaneously satisfies
bIC and psIC must be revenue dominated by a posted price auction; moreover, 
this restriction holds even 
for Bayesian notions of bIC and psIC, even for computationally bounded
players, and even when allowing broadcast.  
\label{cor:intro-expostpsic}
\end{corollary}

Interestingly, \Cref{cor:intro-expostpsic}
can be viewed as a strengthening of an impossibility
shown by Akbarpour and Li~\cite{credibleauction}. 
Translating their result using our terminology, they effectively 
proved that in the information theoretic setting,
any 1-round auction that simultaneously 
satisfies even Bayesian notions of bIC and psIC\footnote{Since they consider
the platform and seller as a single entity, our psIC notion is aligned
with their ``credible'' notion.}
cannot achieve revenue optimality. 
We strengthen their impossibility in multiple dimensions:
\begin{itemize}[leftmargin=6mm,itemsep=1pt]
\item 
We explicitly characterize how small the revenue is by comparing the revenue with some
posted price auction, whereas they only prove that any one-round, revenue-maximizing 
(Bayesian) bIC auction
cannot additionally satisfy Bayesian psIC. 
\item 
Our impossibility holds even for computationally bounded players whereas
their modeling and proofs are restricted to the information theoretic setting.
\end{itemize}

\subsection{Cryptography Meets Platform-Assisted Auctions}

We ask whether we can have an {\it efficient} auction 
that achieves the same incentive compatibility properties
of the ascending auction of \Cref{thm:intro-ascend}. 
We show that the answer 
is affirmative if we can employ the help of cryptography. 

As mentioned, 
one challenge that comes with platform-assisted
auctions is that 
there is no party 
entrusted to honestly implement the auction's rules.
We can mitigate this challenge and restrict 
the strategy space by using a multi-party computation (MPC) protocol
to realize a trusted party. 
A common misconception 
is that simply employing an MPC protocol would trivialize
the problem of designing a platform-assisted auction, 
and bring us back to the classical land of auction design.  
This is not true partly because 
modern platform-assisted auctions 
are permissionless by nature where strategic players can inject fake bids; 
and partly because we ask for additional collusion resilience properties 
--- both of these challenges are typically not considered in the classical
auction theory literature. 
Because of these new challenges, all the impossibility results
of \Cref{sec:intro-imp} would still apply even allowing the use of MPC.

\paragraph{Efficiency limitations of generic MPC.}
As mentioned, we want to use MPC to help restrict the strategy space. 
However, using generic MPC techniques would incur a high cost. 
In particular, in an auction, all the buyers would obtain
a different private outcome. 
To the best of our knowledge, for this setting, 
{\it it is not known how to achieve
generic MPC with subquadratic communication and computation cost}. 
For example, even with  
Threshold Fully Homomorphic Encryption 
(TFHE)~\cite{tfhe00,constrnd-mpc-fair,tfhe01}, 
we would need all players to participate in a joint decryption protocol
to help each buyer decrypt its outcome, 
resulting in at least quadratic total communication.   
Other state-of-the-art approaches~\cite{goyal-linearmpc} 
incur $O(n |C|)$
total communication
where $n$ is the number of players and $|C|$
denotes the circuit size --- but for any $n$-ary
function that must read every input, the circuit size  
is at least linear 
in $n$, making the total communication
at least quadratic again.
If we allow the use of indistinguishability obfuscation (iO)~\cite{io,newio01},
it is possible to adapt Hubacek and Wichs' techniques~\cite{comm-sfe-long} 
to get a protocol with quasilinear communication; but the 
computation would still be quadratic. 

\paragraph{Our approach: utility-dominated emulation.}
We observe that in mechanism design, since we are asking for only game-theoretic
properties, we do not need the underlying 
MPC protocol to offer full simulation-based security. 
Instead, we propose a weaker notion of simulation 
called {\it utility-dominated-emulation}
which suffices for mechanism design. 
Informally, 
let $\mcal{C}$ denote 
some strategic player or coalition; 
the utility-dominated emulation 
notion asks that 
1) for any real-world strategy 
$S$, there exists an ideal-world strategy $S'$ such that  
$\mcal{C}$'s utility 
in the ideal world  
dominates its utility in the real world; and 2)
under an honest execution, 
all players' utilities are identically distributed in the real
and ideal worlds. 
In this way, if honest behavior maximizes
the coalition $\mcal{C}$'s utility in the ideal world,
honest behavior should also maximize its utility in the real world.

The utility-dominated emulation notion
allows mechanism designers 
to focus on the game theoretic
aspects of auction design without having to worry
about the concrete cryptographic instantiation. 
Specifically, the mechanism designer will be 
designing an auction in an ideal world where there is a trusted party
that enforces the correct execution of the auction's rules. 
Because of the existence of this trusted party, the ideal auction need
not employ any cryptography. 
We provide a light-weight cryptographic compiler
such that given an ideal-world auction that satisfies
the desired game theoretic properties, 
we can compile 
it to a real-world protocol that replaces
the trusted party  
with actual cryptography, while retaining
the same game-theoretic guarantees.

Instead of using generic MPC to instantiate the trusted party,
the utility-dominated emulation paradigm allows us 
to devise a new compiler that achieves 
{\it quasi-linear} total communication and computation
cost, which represents an $n$-fold 
efficiency improvement 
relative to known generic MPC protocols (or any protocol
where all players post messages to the broadcast channel). 
Our result can be summarized 
with the following theorem: 

\begin{theorem}[Cryptographic auction with $O(1)$ rounds and quasilinear efficiency]
Assume that the strong repeated squaring assumption, 
the Decisional Diffie Hellman (DDH) and 
Decisional Composite Residuosity (DCR) assumptions (in suitable
groups) all hold against quasi-polynomial-time adversaries.
Then, assuming polynomially bounded players, 
there exists a platform-assisted auction  
that satisfies bIC, pIC, and 1-pbIC; moreover, 
assuming natural value distributions, 
the protocol additionally satisfies 
Bayesian psIC and Bayesian sIC. 
The protocol enjoys the following efficiency guarantees:
\begin{itemize}[leftmargin=6mm,itemsep=1pt]
\item the round complexity is $O(1)$; 
\item the number of bits posted to the blockchain is  
$\widetilde{O}_\lambda(1)$;
\item 
every buyer and seller's computation and 
communication are $\widetilde{O}_\lambda(1)$; and 
\item 
the platform's 
computation and communication is $\widetilde{O}_{\lambda}(n)$ 
where $n$ is the number of players. 
\end{itemize}
\end{theorem}
In the above, we use the notation $\widetilde{O}_\lambda(\cdot)$
to hide 
factors that depend on the security parameter $\lambda$
as well as polylogarithmic factors.

In our cryptographic protocol, the strong repeated squaring
is needed for ensuring an additional desirable property
called robustness, i.e., as long as the platform is honest,
all players will accept with all but negligible probability.  
\elaine{is it prob 1 or all but negl}
If robustness is not required, we can 
get the same result 
from standard assumptions including 
the existence of Non-Interactive Zero-Knowledge (NIZK) and 
collision-resistant hashing
secure against quasi-polynomially sized adversaries.
\elaine{TODO: write this in the appendix}

\subsection{Additional Definitional Contributions}
The vast majority of the mechanism design literature 
does not consider computationally bounded agents,
and their modeling choices are often incompatible
with computationally bounded players. 
We are faced with various definitional subtleties when trying to translate
classical game theoretical concepts
to a computationally bounded setting.  
We explicitly choose not to adopt  
classical modeling techniques and terminology that are incompatible
with computationally bounded players, such 
as extensive-form game, 
information sets, and static protocols.  
Instead, we model an auction as a cryptographic protocol
between players modeled as interactive Turing Machines, and we 
consider the round complexity of the protocol rather than
the number of information sets.

One notable subtlety that arises in our modeling 
is how to define a suitable notion of credible~\cite{credibleauction}
(or psIC in our terminology) that is compatible with computationally bounded
players as explained below.  

\paragraph{A computationally sound notion of ``credible''.} 
The notion of ``credible'' (or psIC in our terminology) is proposed in 
the elegant work of Akbarpour and Li~\cite{credibleauction}. 
Translated to our terminology, 
an auction is said to be 
credible if the platform-seller coalition 
cannot benefit from any {\it safe} deviation. 
Specifically, a strategy is considered safe
iff the following holds with probability $1$: 
\elaine{is it prob 1?}
for every buyer, its view in the protocol has a plausible explanation 
there exist some inputs and random coins of all other
players 
that are consistent with the buyer's view. 
However, the explanations for different buyers are allowed to be different. 

Unfortunately, Akbarpour and Li's formulation of credible  
is incompatible with the cryptographic setting where players
are computationally bounded. 
As a definitional contribution, 
we generalize the notion of credibility (or psIC) 
to the computationally 
bounded  
setting, by 
having every player explicitly output accept or reject
in the protocol's syntax definition. 
An execution trace
is considered safe if all honest players output accept.  
In \Cref{sec:defn-it-ic}, we discuss
how our psIC notion compares
with the ``credible'' notion in more detail
(see paragraph entitled ``Comparison of our psIC notion and credible''). 


\subsection{Additional Related Work}

\elaine{riggs paper cites a bunch of related work,
cite game theory meets mechanism design work, 
cite TFM work}

\paragraph{Cryptography meets auction design.}
Several works have suggested to use 
multi-party computation or related cryptographic techniques to remove
the reliance on 
a trusted auctioneer~\cite{privauction,survey-privauction,danishbeet,riggs}. 
Notably, MPC was deployed in real life in a Danish 
suger beets auction~\cite{danishbeet}.
The excellent survey of Alvarez and 
Nojoumin~\cite{survey-privauction} also provides a more comprehensive
review of this body of work.
Besides cryptography, some works
have also suggested the use of secure processors 
to remove the reliance on a trusted auctioneer~\cite{hawk}. 

We stress that even if the auction itself is based on  
a fully decentralized MPC protocol (or secure processors), 
it does not obviate the need
for a middleman platform
that provides value-added 
services such as product discovery, recommender system, and serves
as a rendezvous point 
between the sellers and the buyers. 
To the best of our knowledge, the 
aforementioned line of work focuses more on cryptographic protocol design, 
and does not  
explore the game theoretic implications when the platform 
can behave strategically to increase its revenue, 
such as injecting fake bids, or 
forming coalitions with either the seller or buyers. 

In our work, we consider a single untrusted platform. 
In comparison, some earlier works 
consider a scenario where the platform is realized
with two or three non-colluding parties (henceforth
called {\it service providers})~\cite{privauction,danishbeet}.
Even in this case, some of our impossibility
results, including \Cref{thm:intro-0fee} and \Cref{thm:intro-2pbic} 
still hold (assuming that the fees paid to the platform
are divided among the parties that jointly realize the platform). 
Here, $c$-pbIC means that any service provider colluding with at most $c$
buyers should be incentivized to behave honestly. 
Note that 
not relying on two or more non-colluding service providers
also makes our approach 
much easier to deploy in practice. 


\paragraph{Credible auction with cryptography or blockchain.}
Recent works
by 
Ferreira et al.~\cite{commit-credible-auction} and
Chitra et al.~\cite{crediblebroadcast}
showed that by employing either 
cryptographic commitments 
or a blockchain, we can circumvent the trilemma 
shown by Akbarpour and Li~\cite{credibleauction}, and get
a constant-round auction that is bIC and Bayesian psIC.
Just like the original credible auction
framework, these works treat the platform and the seller
as the same entity, and 
thus their definitions do not protect against strategies
where the platform cheats the seller. 
Further, both works suffer from 
{\it at least quadratic} total communication. 
Both Ferreira et al.~\cite{commit-credible-auction} and 
Chitra et al.~\cite{crediblebroadcast}
rely on collateral 
to prevent the auctioneer from injecting fake committed bids
and refusing to open later. 
It might be possible to replace their commitments with timed commitments   
and get rid of the collateral in their work, 
but formalizing the resulting scheme is outside the scope of our paper. 



\paragraph{Transaction fee mechanism design.}
Besides platform-assisted auctions, 
transaction fee mechanisms 
(TFMs)~\cite{roughgardeneip1559,roughgardeneip1559-ec,foundation-tfm}
are another example of decentralized mechanism design
where miners (or consensus nodes), who 
partly implement the role of the auctioneer, 
can behave strategically  and possibly form coalitions
with users. 
Transaction fee mechanisms in the plain model~\cite{foundation-tfm} 
(without cryptography)
can be viewed
as a restricted form of a platform-assisted auction, where  
1) the protocol must be a direct revelation mechanism;
2) there is no seller and the revenue to the seller is   
effectively burned on the blockchain --- as a result, the notions
of sIC and psIC are not relevant in the TFM context.  
Further, our pIC notion 
corresponds to miner incentive compatibility, our bIC notion corresponds 
to user incentive compatibility, and our $c$-pbIC notion corresponds
to $c$-side-contract-proofness (SCP) from the TFM literature. 
Because our model can be seen as a generalization of both TFMs and 
the credible auction framework, 
an extra contribution we make is to {\it unify and elucidate the connections
between 
these two previously separate lines of work}.  
We also borrow some proof techniques from the TFM literature
to prove some of our impossibility results, 
e.g., \Cref{thm:intro-0fee} and \Cref{thm:intro-2pbic}. 

Subsequent works~\cite{crypto-tfm,reasonable-tfm} 
on TFM also showed how MPC can help  
circumvent some of the impossibilities in the plain model. 
In comparison, these works~\cite{crypto-tfm,reasonable-tfm} 
simply run a generic MPC protocol
among the miners or MPC infrastructure providers, whereas our work
introduces a new notion called utility-dominated emulation
that allows us to achieve asymptotical efficiency improvements. 
While the existing works~\cite{crypto-tfm,reasonable-tfm} 
rely on the full notion of simulation
to abstract away the cryptography such that the mechanism
designer can work in an idealized model without cryptography,
our utility-dominated emulation provides the same benefits.   
Moreover, the existing works~\cite{crypto-tfm,reasonable-tfm}
need multiple MPC service providers 
among whom a threshold number must be honest, whereas
our approach  
needs only a single untrusted platform to act as the service provider. 






\section{Model}

\elaine{note: bchain or bcast cannot be realized in 
star/hub pairwise channel model}

\subsection{Platform-Assisted Auction}
\label{sec:platform-assisted-auction}
We consider a platform-assisted auction for selling
$k$ identical items. Throughout the paper, we assume that $k$ is finite.  
We assume that every buyer $i$ has a non-negative 
true value
denoted $v_i \in \mathbb{R}_{\geq 0}$ for winning an item. 
Each buyer has unit demand, i.e., winning more than one item
brings the same value as winning exactly one item.

A (possibly multi-round and randomized) platform-assisted auction 
parametrized by $k$ is a {\it protocol}
between the {\it platform}, a set of {\it buyers}, and a {\it seller}:
\begin{enumerate}[leftmargin=6mm]
\item 
{\bf Inputs}:
An honest buyer $i$ always uses its true value $v_i$ as input to the protocol.
The platform does not have any input. 
\item 
{\bf Communication}:
\elaine{mention in our practical instantiations of the crypto protocols, the seller
need not be involved}
During the protocol, 
each player $\in \{\text{buyer, seller}\}$ 
communicates with the platform
through a {\it pairwise private channel}. 
There is no buyer-buyer communication
or buyer-seller communication. 
Additionally, the 
buyers, seller, and platform
may also
post messages to a {\it broadcast channel}, henceforth also
called a {\it blockchain}.  
Messages posted to the blockchain are visible to every player.
\hao{Do we call it blockchain, or just broadcast channel?}
\item 
{\bf Private outcomes}:
We may assume that at the end of the main protocol, 
the platform
sends a single message to every buyer and seller to inform the player of its private outcome.
Specifically, the platform sends either $\bot$ indicating it wants to 
reject the execution, or it sends a private outcome of the following format: 
\elaine{we do not count this msg in round complexity}
\begin{itemize}
\item 
{\it Buyer $i$}:  its private outcome is of the form $(x_i, p_i)$ 
where $x_i$ denotes the number of items allocated to 
buyer $i$, and $p_i$ denotes its payment. 
\item 
{\it Seller}: its private outcome is 
of the form 
$(t, \mu_\mcal{S})$ 
where $t \leq k$ denotes the total number of items sold, 
and $\mu_{\mcal{S}}$ denotes the seller's revenue.  
\end{itemize}
\item 
{\bf Acceptance decision}: 
Finally, 
the seller, the platform, and 
every buyer will 
decide whether to {\it accept} or {\it reject} the execution. 
\ignore{
At the end of the protocol, 
each player outputs the following depending on its role:
\begin{itemize}
\item 
{\it Buyer $i$}: outputs a pair $(b_i, x_i, p_i)$ 
where $b_i \in \{0, 1\}$ indicates
whether the buyer {\it accepts} or {\it rejects} the result, 
$x_i$ denotes the number of items allocated to 
buyer $i$, and $p_i$ denotes its payment. 
\item 
{\it Seller}:
outputs a tuple 
$(b_{\mcal{S}}, t, \mu_\mcal{S})$ 
where $b_{\mcal{S}} \in \{0, 1\}$ indicates
whether the buyer {\it accepts} or {\it rejects} the result, 
$t \leq k$ denotes the total number of items sold, 
and $\mu_{\mcal{S}}$ denotes the seller's revenue.  
\item 
{\it Platform}:
outputs 
a bit $b_{\mcal{P}}\in \{0, 1\}$ indicating
whether to accept or reject, and a vector 
$\{(\widetilde{b}_i, \widetilde{x}_i, \widetilde{p}_i)\}_i || 
(\widetilde{b}_{\mcal{S}}, \widetilde{t}, \widetilde{\mu}_\mcal{S})$ 
that denotes its belief of all buyers' and the seller's outcomes. 
\end{itemize}
}
\end{enumerate}

We require the following natural guarantees: 
\begin{enumerate}[leftmargin=7mm]
 \item 
 {\it Correctness:}
If the platform is honest, all honest players agree on whether to accept
or reject with probability 1. 
Further, if all players are honest, then  
every one accepts with probability 1. 
\item 
{\it Individual rationality:}
If any honest buyer $i$ accepts, it must output some
 private outcome $(x_i, p_i)$ 
where 
$x_i \in \{0, 1\}$, 
and $x_i \cdot  v_i - p_i \geq 0$, i.e., its utility is non-negative.
\item 
{\it Consistency on number sold:}
As long as the platform is honest, it must be that 
$t = \sum_i x_i$, that is, the number of items sold as viewed
by the seller is the same as the number  
of items acquired by all buyers.  
If the platform is strategic (or part of a strategic coalition), then  
it must be that $t \geq \sum_i x_i$, i.e., 
a strategic platform must still 
obtain enough items from the seller to distribute to the buyers.
\item 
{\it Budget feasibility:}
If the platform is honest, it must be that 
$\mu_{\mcal{S}} \leq \sum_i p_i$.
The difference $\sum_i p_i - \mu_{\mcal{S}}$ represents
the platform's revenue.
\end{enumerate}

\ignore{
In a strategic execution, it is possible for a buyer
to get more than one item
(e.g., the buyer can take on multiple identities 
and pretend to be multiple buyers).
Moreover, we require $t \geq \sum_i x_i$ (instead
of strict equality), i.e., even a strategic  
platform must obtain enough items from the seller
to distribute to all winning buyers.  
Further, when some players are strategic, it is not guaranteed
that all honest players will accept
the result.  
Henceforth, we say that an execution trace
is {\bf safe} 
if all honest players accept. 
}


\begin{definition}[View of buyer and seller]
Fix an execution trace, \ke{trace not defined before} 
the \emph{view} of an honest buyer
or seller includes 
its input, all random coins it has consumed and all protocol
messages it has received (including messages posted to the blockchain). 
\end{definition}

\begin{remark}[On the use of the blockchain]
For our model to be general, we allow
the buyers/platform to post messages real-time 
to the blockchain. 
However, in practice, it would be desirable 
to minimize posting to the blockchain 
to save cost. 
For all of our feasibility results 
that require the use of the blockchain, we only
need the platform to 
post a small message to the blockchain 
at the very end of the protocol, i.e., we do not need
the capability of posting to 
blockchain real-time.
On the other hand,  
all of our infeasibility results
hold even  
when the platform and buyers are allowed to post messages 
to the blockchain real-time during the protocol. 
This makes both of our feasibility and infeasibility results
stronger. 
\elaine{double check}
\end{remark}

\paragraph{Weak symmetry.} 
In real life, the auction may look at additional auxiliary 
information (e.g., 
buyer's cryptographic identities, time-of-arrival of the bids)
besides the amount of the bids in making decisions.
Henceforth, all of this additionally information is referred
to as the {\it identity} of the buyer or bid. 
In this paper, we focus on auctions that satisfy
a weak notion of symmetry, i.e., the auction
makes use of such identity information only in tie-breaking. 
For example, if there are multiple buyers bidding at the same amount,
the auction may give preference to the bids that arrive earlier,
or whose identities are associated with higher reputation.  
The 
auction's outcome should not make use of identity information
in any other way. 
More formally, 
we define {\it weak symmetry} as follows, same 
as earlier works~\cite{oca-revisit}. 

\begin{definition}[Weak symmetry]
An auction is called weakly symmetric if it can
always be equivalently described in the following manner: 
given a value vector $\bfb$ where each bid
may carry some extra identity 
information (e.g., cryptographic keys or timestamp), the honest 
(possibly multi-round) auction protocol is a realization
of  
the following functionality. First, it 
sorts the vector $\bfb$ by the values. 
During the sorting step, if multiple values 
are the same, then arbitrary tie-breaking rules may be applied, 
and the tie-breaking can depend on 
the identity information and can be randomized. 
After this sorting step, 
the auction's algorithms 
depend only on the amount of the values and their relative
position in the sorted vector.
\label{defn:weaksym}
\end{definition}

\ignore{
Given any integer $n$, suppose there are $n$ buyers with true values $(v_1,\dots,v_n)$ participating in the auction,
and let $((X_1,P_1),\dots,(X_n,P_n))$ be a vector of random variables where $(X_i,P_i)$ represents buyer $i$'s private outcome.
An auction is said to be \emph{weakly symmetric} 
iff for any $i \neq j$ and $v_i = v_j$,
the random variables $(X_i,P_i)$ and $(X_j,P_j)$ are identically distributed
assuming everyone executes the protocol honestly.
}
\ignore{
Conceptually, a symmetric auction is one that
makes use of only buyers' bids to decide
the outcome and the auction does not use any metadata
(e.g., the pseudonym of the buyer, arrival time, etc.)
to make decisions. 
\elaine{the way you describe second price auction may not be symmetric} 
}

\subsection{Strategy Space and Utility}
We consider a {\it permissionless} model where a player (buyer, seller,
or platform) can create
one or more fake identities to participate in the auction. 
For our computationally IC auctions, we assume that 
the identity space is super-polynomial in the security
parameter $\lambda$, such that a probabilistic polynomial-time (\ppt)
adversary can make arbitrarily 
many fake identities as long as the total number is
bounded by its running time. 
For the information theoretic setting, 
we may assume that the identity space is infinitely large.

\paragraph{Buyer strategy space.}
A buyer can use an arbitrary input to the protocol rather than its true value;
it can also deviate from the honest protocol arbitrarily. 
In particular, it can take future actions in an adaptive manner
based on its view in the protocol so far. 
A buyer can also take on multiple identities 
and pretend to be multiple buyers.

\paragraph{Seller strategy space.}
A strategic seller can enter the auction pretending 
to be one or more buyers, and using arbitrary values as inputs. 
It can send arbitrary 
messages 
to the platform or the blockchain, 
and it can adapt its future actions based
on its view in the protocol so far. 
 
\paragraph{Platform and coalition strategy space.}
A strategic platform may deviate arbitrarily from the prescribed protocol,
including creating and participating under fake identities; moreover, 
it can adapt its actions based on its view in the protocol so far. 
For a platform-seller or platform-buyer coalition, the strategy
space is the union of the possible strategies 
of all players in the coalition.  

\paragraph{Safe and unsafe execution traces.}
In a strategic execution, 
it is not guaranteed that all honest players will decide to accept.
We say that an execution trace is {\it safe} if all honest players accept,
otherwise it is said to be {\it unsafe}.

\paragraph{Utility.}
If the platform is honest, then we can define utilities as follows.  
If the platform rejects, then every one's utility is 0.
If the platform accepts, then 
\begin{itemize}[leftmargin=5mm,itemsep=1pt]
\item 
a buyer $i$'s utility is 
$v_i - p_i$ if $x_i > 0$, otherwise its utility is $0$; 
\item the seller's
utility is 
its total revenue $\mu_{\mcal{S}}$;
and 
\item 
the platform's utility
is $\sum_i p_i - \mu_{\mcal{S}}$. In other words, 
the platform is remunerated 
the difference between 
the total payment and what the seller gets. 
\end{itemize}

If the platform is strategic 
(or part of a strategic coalition),  
then the platform (or the coalition) has utility $0$ 
on any unsafe execution trace.  \elaine{TODO: define}
Otherwise, if the execution trace is safe,  
we can compute the strategic platform's  (or coalition's) utility as follows:
\begin{itemize}[leftmargin=5mm,itemsep=1pt]
\item For a platform-seller coalition, 
its utility is the 
sum of all honest buyers' payments;
\item 
For a strategic platform colluding with some buyers $\mcal{B}$, 
we can compute its utility  
as follows. The number of items obtained by 
the coalition 
is $x^* = t-\sum_{i \in \mcal{H}} x_i$
where $\mcal{H}$ is the set of honest buyers. 
Let $\mcal{B}^*$ be the 
top $\min(x^*, |\mcal{B}|)$  
buyers in $\mcal{B}$
with the highest true values. 
Then, 
the coalition's utility is 
$\sum_{i \in \mcal{H}}  p_i - \mu_{\mcal{S}} + 
\sum_{j \in \mcal{B}^*} v_j$. 
In other words, the coalition's utility is computed
assuming it allocates its items to the 
colluding buyers with highest true values, and any leftover item has no value.
\item 
For a strategic platform alone, we can compute its utility in the same way as above
but assuming $\mcal{B} = \emptyset$, i.e., its utility
is simply $\sum_{i \in \mcal{H}}p_i - \mu_{\mcal{S}}$.
\end{itemize}

Note that 
the motivation for forcing the platform
or its coalition's utility 
to be zero upon an unsafe trace is aligned with  
the existing literature on 
credible auctions~\cite{credibleauction,commit-credible-auction,crediblebroadcast}.
Specifically, we consider a platform that values its reputation,
and thus would not want to risk  
taking strategies 
that can impair its reputation. 
In fact, it may also  
make sense to set the platform or its coalition's utility
to be $-\infty$ upon an 
unsafe trace --- see the paragraph 
entitled ``Comparison of our psIC notion and credible''
in \Cref{sec:defn-it-ic} for more detailed discussions.
\elaine{TODO: give some philosophical interpretation here.}


\ignore{
If the execution involves a strategic player or a strategic platform-buyer
or platform-seller coalition, then  
the execution trace may not always be safe. 
For a safe execution trace,
a strategic player's  
utility is 
calculated in the same way as above, and a strategic coalition's utility
is the sum of all colluding members' utilities.  
On the other hand, for an unsafe execution trace, we 
define the strategic player or coalition's utility 
to be $0$. 
}


\elaine{TODO: it makes sense to define  
{\bf robustness} as follows: as long as platform is honest, all honest 
players will output accept.
{\bf no false incrimination}: if a strategic user i falsely rejects, an honest platform
can prove user i's outcome to a third party.  
}

\section{Preliminaries}

\elaine{TODO: organize}
A function $f: \N \rightarrow [0,1]$ is called \emph{negligible},
if for every positive polynomial $p(\cdot)$ and all sufficiently large $n$, 
it holds that $f(n) < 1/p(n)$.
\ignore{
Given some distribution $X$ on (a subset of) real values, we say that $X$ is 
{\it non-negative}
if $\Pr_{x \getr X}[x \geq 0] = 1$.
Similarly, a probability ensemble 
$\{X_\lambda\}_\lambda$ indexed by $\lambda$ is said to be non-negative,
iff 
$X_\lambda$ is non-negative for every $\lambda$. 
}
\elaine{deleted the non-negative part}

A possibly randomized 
machine is said to be probabilistic polynomial time (\ppt), iff
it completes in total work upper-bounded
by some polynomial function in its input length.  
Let $D(\cdot)$ be a polynomially bounded function. 
We say that a (non-uniform) \ppt machine is 
{\it depth-$D$-bounded}, iff on 
any input $a$ of length $|a|$, 
its depth is at most $D(|a|)$. 

\elaine{NOTE: with 
depth bounded, post proc lemma needs to calculate depth}  

\begin{definition}[Computational indistinguishability against depth-bounded adversaries]
Let $D(\lambda)$ be a polynomially
bounded functions in $\lambda$. 
Consider probability ensembles $\{\rv{X}_{\secu}\}_\secu$ and $\{\rv{Y}_{\secu}\}_\secu$
indexed by the security parameter $\secu \in \N$.
We say that 
$\{\rv{X}_{\secu}\}_\secu$ and $\{\rv{Y}_{\secu}\}_\secu$ are 
\emph{computationally indistinguishable 
against 
depth-$D$ adversaries},  
iff for every 
depth-$D$ bounded non-uniform \ppt machine,  
there exists a negligible function $\negl$ such 
that for any $\secu \in \N$, \[
\left|\Pr_{x \getr X_\lambda}[\mcal{A}(1^\secu, x) = 1] - \Pr_{y \getr Y_\lambda}[\mcal{A}(1^\secu, y) = 1]\right| \leq \negl(\secu).
\]

As a special case, if there is no depth-$D$ restriction
on the adversary, we simply say that  
$\{\rv{X}_{\secu}\}_\secu$ and $\{\rv{Y}_{\secu}\}_\secu$ are 
\emph{computationally indistinguishable}.
\end{definition}

\elaine{TODO: add the technical lemma needed for
comp ind for log depth ==> expectation close} 

\begin{fact}
Suppose that 
$\{X_\lambda\}_\lambda$
and $\{Y_\lambda\}_\lambda$
are probability ensembles taking value over $[0, 1]$. 
Suppose that 
$\{X_\lambda\}_\lambda$
and $\{Y_\lambda\}_\lambda$
are computationally indistinguishable 
for depth-$C \log(\cdot)$-bounded adversaries
where $C$ is a suitably large universal constant. 
Then, there is a negligible function $\negl(\cdot)$ such that 
for every $\lambda \in \mathbb{N}$, 
$|\E[X_\lambda] - 
\E[Y_\lambda]| \leq \negl(\lambda)$. 
\label{fct:indtoexp}
\end{fact}
\begin{proof}
We first prove the following useful claim. 
\begin{claim}
Suppose 
$\{X_\lambda\}_\lambda$
and $\{Y_\lambda\}_\lambda$
are computationally indistinguishable
for depth-$C \log(\cdot)$ adversaries. Then, it must be
that there is a negligible function $\negl(\cdot)$
such that for every $\lambda$,  
for any $v \in [0, 1]$, $|{\sf CDF}_{X_\lambda}(v) - {\sf CDF}_{Y_\lambda}(v)| 
\leq \negl(\cdot)$
where ${\sf CDF}_X(\cdot)$ 
denotes the cumulative distribution function
of some random variable $X$. 
\end{claim}
\begin{proof}
Suppose the claim is not true, i.e., 
there is some $\poly(\cdot)$ such that 
for infinitely many $\lambda$'s, 
there is some $v_\lambda \in [0, 1]$
such that
$|{\sf CDF}_{X_\lambda}(v_\lambda) - {\sf CDF}_{Y_\lambda}(v_\lambda)|
> 1/\poly(\lambda)$.  
Then, we can construct a polynomially sized distinguisher $\mcal{B}$ of depth 
at most $C \log(\ell(\lambda))$ 
where $\ell(\lambda)$ is the bit-length  
of $X_\lambda$ or $Y_\lambda$, that can effectively distinguish
a random sample of $X_\lambda$ or $Y_\lambda$
with $1/\poly(\lambda)$ 
probability for infinitely many $\lambda$'s. 
Basically, $\algB(1^\lambda)$ 
receives a value that is sampled
from either $X_\lambda$ or $Y_\lambda$,  
if the value is at most $v_\lambda$, it outputs $1$; else output $0$. 
The depth of $\mcal{B}$ is upper bounded by the depth of a comparator,
which is logarithmic in the input length. 
\end{proof}

We now 
continue with the proof of \Cref{fct:indtoexp}. 
It suffices to show that 
$\E[X_\lambda]  \leq 
\E[Y_\lambda] + \negl(\lambda)$ since the other direction is symmetric. 
Consider the random variable $X'_\lambda$
whose cumulative distribution function is 
defined as follows where $\negl(\cdot)$ is the negligible function
in the above claim:
\[
{\sf CDF}_{X'_\lambda}(v) = 
\max\left({\sf CDF}_{X_\lambda}(v)  + \negl(\lambda), 1\right)
\]
Intuitively, 
${X'_\lambda}$ is obtained by 
forcing the largest negligible fraction of 
$X_\lambda$ to $0$. 
Because of the above claim, it holds that 
$X'_\lambda$
is stochastically dominated 
by $Y_\lambda$. 
Therefore, 
we have that 
\[\E[X_\lambda] \leq 
\E[X'_\lambda] + \negl(\lambda)
\leq 
\E[Y_\lambda] + \negl(\lambda)
\]
\end{proof}

\section{Incentive Compatibility Definitions}

We now define our incentive compatibility (IC) notions.
Throughout the paper, if the auction makes
use of identity information to break ties
(see also the definition of weak symmetry 
in \Cref{sec:platform-assisted-auction}), we require
that our IC notions hold for any choice of identities. 
\elaine{i added this b/c the impossibility proofs assume this}

\subsection{Computational Incentive Compatibility}
\label{sec:IC-computational}

To define computational incentive compatibility, we assume
that the execution is parameterized with a security parameter $1^\lambda$
and every player receives $1^\lambda$ as an additional input.  
The players are modeled as Turing Machines. 

\paragraph{Normalized and discretized value domain.}
Because players are computationally bounded, we cannot
have infinite 
precision to represent real-numbered values.
We assume that every player's true value
is drawn 
from some finite domain $\inDomSecu := \{0, 1\}^{\ell(\lambda)}$ 
that is a subset of the non-negative reals, 
such that the true
value can always be represented by 
$\ell(\lambda) \leq \poly(\lambda)$ bits. 
Since the value domain is finite, 
there is a maximum possible value. 
Therefore, {\it by rescaling, we 
may assume without loss of generality that the value domain 
$\inDomSecu$ encodes
a finite subset 
of $[0, 1]$}, written as $\inDomSecu \subset [0, 1]$. 
{\it Throughout the paper, we assume
that $0 \in \inDomSecu$}
for any $\lambda$. 
\elaine{TODO: comment on how to remove this assumption for LB and UB}

We use $\inDomSecu^k$ to denote a length-$k$ vector where 
each coordinate is from 
$\inDomSecu$, 
and 
we use the notation 
$\inDomSecu^*$
to denote an arbitrary-length vector where 
each coordinate is from 
$\inDomSecu$.


\begin{definition}[Computational (ex post) incentive compatibility]
\label{def:comp-IC}
Given an auction parametrized by $\lambda$ over
the value domain $\inDomSecu\subset [0, 1]$, 
we say that the auction satisfies (ex post) 
incentive compatibility w.r.t. a 
strategic player or coalition denoted $\mcal{C}$, iff
for any polynomial function $n(\cdot)$, 
for any 
probabilistic polynomial time (\ppt)\footnote{Throughout this paper,
we allow the strategic player or coalition's algorithm to be a non-uniform machine.} strategy $S$ adopted by $\mcal{C}$,
there exists a negligible function $\negl(\cdot)$ such that 
for any $\secu \in \N$,
for any $n_H \leq n(\lambda)$
for any true value 
vector ${\bf v}_{-\mcal{C}} \in  \inDomSecu^{n_H}$ corresponding to honest buyers, 
and for any true value vector 
${\bf v}_{\mcal{C}} \in  \inDomSecu^{n_C}$ of the buyers in $\mcal{C}$
where $n_C$ is the number of buyers in $\mcal{C}$, 
$$
\E\left[{\sf util}^S_{\mcal{C}}(1^\lambda, {\bf v}_{-\mcal{C}}, {\bf v}_{\mcal{C}}) \right]
\leq 
\E\left[
{\sf util}^H_{\mcal{C}}(1^\lambda, {\bf v}_{-\mcal{C}}, {\bf v}_{\mcal{C}}) 
\right]
+ \negl(\lambda)
$$
where $H$ denotes the honest strategy, and moreover,
the notation ${\sf util}^S_{\mcal{C}}(1^\lambda, {\bf v}_{-\mcal{C}}, {\bf v}_{\mcal{C}})$
represents the random variable  
corresponding to coalition $\mcal{C}$'s utility 
in the  following randomized experiment:
\begin{enumerate}[leftmargin=6mm,itemsep=1pt]
\item 
execute the protocol with security parameter $1^\lambda$, where 
$\mcal{C}$ adopts the strategy $S$, and all remaining players act honestly
with the honest buyers taking true values ${\bf v}_{-\mcal{C}}$ as input;  
\item 
output the utility of $\mcal{C}$
assuming that the buyers
in $\mcal{C}$ have the true value vector\footnote{Note
that the true value vector ${\bf v}_{C}$ of the buyers in $\mcal{C}$ is only
used to calculate $\mcal{C}$'s utility, the coalition
can adopt an arbitrary strategy which implies that it can submit
an arbitrary bid vector.} 
 ${\bf v}_{C}$. 
\end{enumerate}
\end{definition}

\begin{definition}[Computational Bayesian incentive compatibility]
\label{def:comp-bayesian-IC}
Given an auction parametrized by $\lambda$ over
the value domain $\inDomSecu\subset [0, 1]$, 
we say that the auction satisfies 
Bayesian 
incentive compatibility w.r.t. a 
strategic player or coalition denoted $\mcal{C}$ and the distribution $\mcal{D}_\lambda$, iff
for any polynomial $n(\cdot)$, 
for any $\ppt$ strategy $S$ adopted by $\mcal{C}$,  
there exists a negligible function $\negl(\cdot)$, 
such that for any $\lambda$, 
for any $n_H \leq n(\lambda)$, 
any true value vector $\bfv_{\mcal{C}} \in \inDomSecu^{n_C}$
corresponding to the buyers in $\mcal{C}$ where $n_C$ denotes
the number of buyers in $\mcal{C}$, 
$$
\E\left[{\sf util}^S_{\mcal{C}}(1^\lambda, n_H, {\bf v}_{\mcal{C}})\right] 
\leq 
\E\left[{\sf util}^H_{\mcal{C}}(1^\lambda, n_H, {\bf v}_{\mcal{C}})\right] 
+ \negl(\lambda)
$$
where $H$ denotes the honest strategy, and moreover,
the notation 
${\sf util}^S_{\mcal{C}}(1^\lambda, n_H, {\bf v}_{\mcal{C}})$
represents the random variable  
corresponding to coalition $\mcal{C}$'s utility 
under the following execution: 
\begin{enumerate}[leftmargin=6mm,itemsep=1pt]
\item 
sample a vector ${\bf v}_{-\mcal{C}}$ of length $n_H$ from $\mcal{D}_\lambda^{n_H}$;
\item 
execute the protocol with security parameter $1^\lambda$, where 
$\mcal{C}$ adopts the strategy $S$, and all remaining players act honestly
with the honest buyers taking true values 
${\bf v}_{-\mcal{C}}$;  
\item 
output the utility of $\mcal{C}$
assuming that the buyers
in $\mcal{C}$ have the true value vector ${\bf v}_{C}$. 
\end{enumerate}
\end{definition}

\ignore{
Note that in the above 
\Cref{def:comp-IC,def:comp-bayesian-IC}, we implicitly
assume that the length of ${\bf v}_{-\mcal{C}}$
and $n_H$
are bounded by some fixed polynomial function in $\lambda$. 
However, the auction may not know this polynomial in advance. 
}

Given 
\Cref{def:comp-IC,def:comp-bayesian-IC},
we can define the following notions:

\begin{itemize}[leftmargin=6mm]
\item 
{\it Computational (Bayesian) buyer incentive compatibility (bIC)}: 
\Cref{def:comp-IC} (or \Cref{def:comp-bayesian-IC}) must hold when 
$\mcal{C}$ contains only an individual buyer
assuming the platform follows the protocol honestly; 
\item 
{\it Computational (Bayesian) seller incentive compatibility (sIC)}: 
\Cref{def:comp-IC} (or \Cref{def:comp-bayesian-IC}) must hold when 
$\mcal{C}$ contains only the seller
assuming the platform follows the protocol honestly;
\item 
{\it Computational (Bayesian) platform incentive compatibility (pIC)}: 
\Cref{def:comp-IC} (or \Cref{def:comp-bayesian-IC}) must hold when 
$\mcal{C}$ contains only the platform; 
\item 
{\it Computational (Bayesian) platform-seller incentive compatibility (psIC)}:
\Cref{def:comp-IC} (or \Cref{def:comp-bayesian-IC}) must hold when 
$\mcal{C}$ contains the platform and the seller; 
\item 
{\it Computational (Bayesian) $c$-platform-buyer incentive compatibility ($c$-pbIC)}:
\Cref{def:comp-IC} (or \Cref{def:comp-bayesian-IC}) must hold when 
$\mcal{C}$ contains the platform and 
an arbitrary non-empty set of at most $c$ buyers. 
\end{itemize}

\elaine{TODO: discuss why not consider other coalitions}


\subsection{Strong Computational Incentive Compatibility}

We also define a stronger variant of our computational IC properties.
Intuitively, these stronger notions 
capture 
the idea that the only way to gain over the honest strategy is to 
break the cryptography. 
In other words, 
if the strategic player or coalition is restricted
to strategies that do 
not involve breaking cryptography, then the $\negl(\lambda)$
term in 
\Cref{def:comp-IC,def:comp-bayesian-IC}
would be forced to $0$.
More specifically, we require that any strategy 
that does nothing more than using 
non-truthful values, injecting fake bids, and dropping out
should not do better than the honest strategy. 

\paragraph{Input replacement strategies.}
We say that a coalition (or an individual buyer) $\mcal{C}$ adopts an 
{\it (extended) input replacement} strategy 
if it uses as input 
an arbitrary bid vector ${\bf b}_{\mcal{C}}$
which need not be the same as the true values ${\bf v}_\mcal{C}$
of the buyers in $\mcal{C}$. 
Otherwise, $\mcal{C}$ follows the honest protocol.
With extended input 
replacement, the length
of the strategic bid vector ${\bf b}_{\mcal{C}}$ need not 
be the same as ${\bf v}_\mcal{C}$, 
i.e., $\mcal{C}$ may inject fake bids or have some colluding buyers drop 
out. By contrast, an input replacement 
strategy is more restrictive and requires that the two lengths be the same.  
\elaine{there's some implicit assumption about the identities here}

\begin{definition}[Strong computational (Bayesian) incentive compatibility]
We say that an auction satisfies strong computational (Bayesian)
incentive compatibility 
w.r.t. the strategic player or coalition $\mcal{C}$, 
iff 
\begin{itemize}[leftmargin=6mm,itemsep=1pt]
\item it satisfies 
\Cref{def:comp-IC} (or \Cref{def:comp-bayesian-IC}), 
and 
\item 
if $\mcal{C}$ is  
restricted to \ppt extended input replacement strategies, 
then 
\Cref{def:comp-IC} (or \Cref{def:comp-bayesian-IC})
is satisfied even when 
the $\negl(\cdot)$ function is forced to $0$. 
\end{itemize}
\label{defn:strongcompic}
\end{definition}

Based on \Cref{defn:strongcompic}, we can define
strong computational (Bayesian) 
bIC, sIC, pIC, pbIC, and psIC, respectively.

\begin{remark}[Computational IC by utility-dominated emulation 
implies strong computational IC.]
Later in \Cref{sec:defnsim}, we will define
a notion of utility-dominated emulation 
which can be viewed as an alternative way to define
computational IC (see \Cref{rmk:alternate-comp-ic}). 
This utility-dominated emulation 
framework also gives a new paradigm for designing 
and reasoning about the IC properties of cryptographic auctions. 
It is not hard to see that  
any auction that is a utility-dominated 
emulation of an incentive compatible ideal auction 
satisfies our strong computational IC notion.
\end{remark}

\subsection{Information-Theoretic Incentive Compatibility}
\label{sec:defn-it-ic}

\elaine{TODO: rewrite this section, 
see whether we still need to index it with lambda.
and whether we should unroll the ensemble notation 
here 
}

Given an auction parametrized by $\lambda$ over
the the family of value domains $\inDomSecu$, 
we say that the auction 
satisfies 
{\it information-theoretic} (Bayesian) incentive compatibility
against some coalition $\mcal{C}$,
iff 
\Cref{def:comp-IC} (or \Cref{def:comp-bayesian-IC}) 
holds
but with the following modifications:
\begin{itemize}[itemsep=1pt]
\item 
it holds for not just \ppt but even unbounded strategies; 
\item 
the $\negl(\cdot)$ function in 
\Cref{def:comp-IC} (or \Cref{def:comp-bayesian-IC}) 
is forced to $0$; 
and 
\item the restriction that 
the length of ${\bf v}_{-\mcal{C}}$
and $n_H$
are polynomially bounded is removed. 
\end{itemize}

Equivalently, the definitions can also be rephrased as the following
(expanding the ensemble notations for clarity): 

\begin{definition}[Information theoretic incentive compatibility]
Given an auction parametrized by $\lambda$ over
the family of value domains $\inDomSecu$, 
we say that the auction 
satisfies 
{\it information-theoretic incentive compatibility}
against some coalition $\mcal{C}$,
iff for any $\lambda$, 
for any honest value vector $\bfb_{-\mcal{C}} \in \inDomSecu^*$, 
for any value vector $\bfv_{\mcal{C}} \in \inDomSecu^{|\mcal{C}|}$ 
of the coalition $\mcal{C}$, 
for any strategy $S$ of the coalition $\mcal{C}$,  
$\E[{\sf util}^S_{\mcal{C}}(1^\lambda, \bfb_{-\mcal{C}}, 
{\bf v}_{\mcal{C}})] \leq \E[{\sf util}^H_{\mcal{C}}(1^\lambda, \bfb_{-\mcal{C}}, 
{\bf v}_{\mcal{C}})]$
where the notations 
${\sf util}^H(1^\lambda, \bfb_{-\mcal{C}}, \bfv_{\mcal{C}})$ 
and 
${\sf util}^S(1^\lambda, \bfb_{-\mcal{C}}, \bfv_{\mcal{C}})$ 
are defined
in the same way as in \Cref{def:comp-IC}.
\label{defn:it-ic}
\end{definition}

\begin{definition}[Information theoretic Bayesian incentive compatibility]
Given an auction parametrized by $\lambda$ over
the family of value domains $\inDomSecu$, 
we say that the auction 
satisfies 
{\it information-theoretic Bayesian incentive compatibility}
against some coalition $\mcal{C}$,
iff for any $\lambda$, 
for any 
$n_H$ that corresponds to the number of honest buyers, 
for any value vector $\bfv_{\mcal{C}} \in \inDomSecu^{|\mcal{C}|}$ 
of the coalition $\mcal{C}$, 
for any strategy $S$ of the coalition $\mcal{C}$,  
$\E[{\sf util}^S_{\mcal{C}}(1^\lambda, n_H, 
{\bf v}_{\mcal{C}})] \leq \E[{\sf util}^H_{\mcal{C}}(1^\lambda, n_H, 
{\bf v}_{\mcal{C}})]$
where the notations 
${\sf util}^H(1^\lambda, n_H, \bfv_{\mcal{C}})$ 
and ${\sf util}^S(1^\lambda, n_H, \bfv_{\mcal{C}})$ are defined
in the same way as in \Cref{def:comp-bayesian-IC}.
\label{defn:it-bayesian-ic}
\end{definition}

\begin{remark}[Note about non-cryptographic auctions over reals]
In classical mechanism design (without cryptography), 
we typically consider a single auction
whose value domain is over non-negative reals,  
and the auction is 
not parametrized by $\lambda$. 
This can be viewed as a special case of 
a family of auctions parametrized by $\lambda$, 
where
 $\inDomSecu = \mathbb{R}_{\geq 0}$
for any $\lambda$, 
and further, the auction protocol does not depend on $\lambda$ 
--- note that for the information theoretic setting,
we can remove the constraint that $\inDomSecu$ must be a finite
domain where each value is 
encoded by polynomial in $\lambda$ bits. 
In this sense, 
\Cref{defn:it-ic,defn:it-bayesian-ic} 
can be interpreted for classical auctions over a real-valued domain as well. 
\label{rmk:classical}
\end{remark}

\elaine{we used the terminology IC under extended input replacement,
it means strict IC, but the strategy needs to be poly time.
we should define IC under XXX strategies.}

\elaine{TODO: if we set coalition util = -inf, IT psIC is equivalent to
credible from a feasibility perspective, do it in appendix}

\paragraph{Comparison of our psIC notion and credible.}
Our information theoretic psIC notion can be viewed
as a desirable strengthening of the credible
notion proposed by Akbarpour and Li~\cite{credibleauction}.
Specifically, their notion requires that honest behavior
be utility-maximizing
among only strategies that must be safe with probability $1$.
However, a coalition may consider adopting a strategy
that risks generating an unsafe trace with extremely small
probability, and the vast majority of times, the strategy results
in a safe execution that benefits the coalition. 
In our formulation, we simply force the coalition's utility to be $0$
on unsafe traces when computing the coalition's expected utility. 
In other words, Akbarpour and Li's formulation can also be equivalently viewed
as forcing the coalition's
utility to be $-\infty$ upon an unsafe trace in our framework. 

Moreover, we define a safe 
trace as one where all honest players accept, whereas  
Akbarpour and Li
requires a safe trace to be one where there exists a possibly different
innocent explanation to every honest player. 
In the information theoretic setting, 
for any protocol that is credible  by their notion, we can simply
augment the protocol by having the platform providing an explanation
to every honest player, and the player accepts
if the explanation is valid.  
On the other hand, in a computationally bounded setting, searching
for an innocent explanation 
may not be feasible within polynomial time. 
Therefore, our new modeling approach
is necessary for the definitions to be compatible
with computationally bounded agents.

\subsection{Preliminary: Myerson's Lemma}
We will rely on the famous Myerson's Lemma for our impossibility
proofs. 
Below we state a version of Myerson's Lemma
that holds regardless whether the value domain is continuous
or discrete. 

\elaine{comment on the input replacement strategy} 

Below, we use the notation 
$\bfx(\bfb) \in [0, 1]^{|\bfb|}$
to denote 
the probabilities that each buyer gets an item 
under the value vector $\bfb$ 
(assuming that every one acts honestly).
Similarly, we use the 
random variables 
$\bfp(\bfb) \in \mathbb{R}_{\geq 0}^{|\bfb|}$
and $\mu_\seller (\bfb) \in \mathbb{R}_{\geq 0}$
to denote everyone's expected payment  as well as the seller's revenue 
under $\bfb$.  
We use $x_i(\cdot)$ and 
$p_i(\cdot)$ to denote the 
$i$-th coordinate of $\bfx$ and $\bfp$, that is, 
the $i$-th buyer's probability of getting an item and  
its expected payment. 
Note that $\bfx$, $\bfp$, and $\mu_\seller$ are well-defined
even for multi-round auctions.

\begin{lemma}[Myerson's Lemma]
Suppose an auction satisfies bIC under input replacement strategies, 
then the following must hold:
for any buyer $i$, any 
value vector ${\bf b}_{-i}$ corresponding to all other buyers, 
\begin{itemize}[leftmargin=5mm]
\item 
{\bf Monotone allocation rule:}
suppose $b'_i > b_i$, 
it must be that 
$x_i({\bf b}_{-i}, b'_i) \geq 
x_i({\bf b}_{-i}, b_i)$. 
\item 
{\bf Payment sandwich:}
suppose $b'_i > b_i$, 
then 
$$
b_i \cdot  \left(x_i({\bf b}_{-i}, b'_i) - x_i({\bf b}_{-i}, b_i)\right)
\leq 
p_i({\bf b}_{-i}, b'_i) 
- p_i({\bf b}_{-i}, b_i) 
\leq 
b'_i \cdot  \left(x_i({\bf b}_{-i}, b'_i) - x_i({\bf b}_{-i}, b_i)\right). 
$$ 
\end{itemize} 
\label{lem:myerson}
\end{lemma}

The above statement of the Myerson's Lemma  
(\Cref{lem:myerson}) 
holds even when the value domain is discrete --- in this case, 
the variables ${\bf b}_{-i}$, $b_i$, and $b'_i$
in the statement \Cref{lem:myerson}
are assumed to take values from the discrete value domain.

In the standard literature, Myerson's Lemma 
is sometimes stated using a unique payment rule in lieu of  
the payment sandwich of \Cref{lem:myerson}.  
Note that the unique payment rule does not hold
in the case of a discrete value domain.

\elaine{TODO: do the bayesian case here.}

The Myerson's Lemma can be extended to the Bayesian setting stated as follows.

\begin{lemma}[Myerson's Lemma in Bayesian setting]
Suppose an auction satisfies Bayesian bIC 
under input replacement strategies
w.r.t. some distribution $\mcal{D}$, 
then for any $n \in \N$, and for any buyer $i$, the following must hold:
\begin{itemize}[leftmargin=5mm]
\item 
{\bf Monotone allocation rule:}
suppose $b'_i > b_i$, 
it must be that 
$\displaystyle \mathop{\E}_{{\bf b}_{-i} \getr \mcal{D}^n}[x_i({\bf b}_{-i}, b'_i)] \geq 
\displaystyle \mathop{\E}_{{\bf b}_{-i} \getr \mcal{D}^n}[x_i({\bf b}_{-i}, b_i)]$.
\item 
{\bf Payment sandwich:}
suppose $b'_i > b_i$, 
then it must be
\begin{multline*}
	\underset{{\bf b}_{-i} \getr \mcal{D}^n}{\E}
	\left[b_i \cdot  \left(x_i({\bf b}_{-i}, b'_i) - x_i({\bf b}_{-i}, b_i)\right)\right]
	\leq 
	\underset{{\bf b}_{-i} \getr \mcal{D}^n}{\E}
	\left[p_i({\bf b}_{-i}, b'_i) 
	- p_i({\bf b}_{-i}, b_i)\right] \\
	\leq 
	\underset{{\bf b}_{-i} \getr \mcal{D}^n}{\E}
	\left[b'_i \cdot  \left(x_i({\bf b}_{-i}, b'_i) - x_i({\bf b}_{-i}, b_i)\right)\right]. 
\end{multline*}
\end{itemize} 
\label{lem:myerson-bayesian}
\end{lemma}

\elaine{move to the lb section:
Given a computational (Bayesian) bIC mechanism, 
we say that it is {\it natural}
iff for any strategic buyer $i$, 
 any true value $v_i$ of the buyer, 
no input-replacement strategy can allow $i$ to achieve better expected  
utility than behaving honestly, 
where the expectation is taken over the randomness
of the mechanism (or both randomness of the mechanism and all other players' 
true values). 
In other words, although a computationally bIC 
allows a negligible slack in the incentive compatibility (IC)
notion for general strategies, we require that 
when restricted to 
input replacement strategies only, 
this slack must be 0. 
For example, any auction that is  
a utility-dominated emulation 
of an ideal (Bayesian) bIC auction 
must satisfy this naturalness
requirement --- 
see \Cref{defn:sim} of \Cref{sec:defnsim}
for the notion of utility-dominated emulation. 
In fact, as we point out in \Cref{rmk:alternate-comp-ic},
utility-dominated emulation 
of an ideal IC auction 
is an alternative way of defining computational IC.
}

\section{Inefficient Information-Theoretic Feasibility}
\label{sec:ascending}

In this section, we show that the ascending auction
where the platform gets a fixed fee  
satisfies bIC, pIC, 1-pbIC, as well as Bayesian sIC  
and Bayesian psIC.   \elaine{TODO: double check}
Like Akbarpour and Li~\cite{credibleauction}, 
rather than introducing a model for continuous-time auctions, 
for simplicity, 
we shall assume a value domain $\inDom$ that is discretized
and consists of the  
values $0 = \theta_1 < \theta_2 < \ldots \theta_T$.  
Since we are in the information theoretic setting, 
we do not need to 
index the auction with the security
parameter $\lambda$ (see \Cref{rmk:classical}).

We assume that honest buyers' true values
are sampled independently from the distribution 
$\mcal{D}$. 
We assume all buyers' have the same distribution
$\mcal{D}$ 
for convenience, but our results readily extend
to the setting where 
all buyers' values are sampled independently but from
different distributions. 

Given $\mcal{D}$ over $\inDom$ with the 
cumulative distribution function $F$
and 
probability density function $f$, we 
define the \emph{virtual value} 
of $\theta_i$ as 
$\phi(\theta_i) := \theta_i - \frac{1 - F(\theta_i)}{f(\theta_i)}\left(
\theta_{i+1} - \theta_i\right)$, and 
let $\phi(\theta_T) = \theta_T$. \elaine{double check}
We say that the distribution $\mcal{D}$
is {\it regular}, iff
$\phi(\cdot)$
is a strictly increasing function, i.e., $v > v' \Longrightarrow
\phi(v) > \phi(v')$. 

\begin{mdframed}
\begin{center}
{\bf Ascending auction with reserve and fixed platform fee}
\end{center}

\vspace{-8pt}
\paragraph{Input:} each buyer $i$ has a true value $v_i$.

\vspace{-5pt}
\paragraph{Auction protocol:}

\begin{itemize}[leftmargin=6mm,itemsep=1pt]
\item 
{\it Register:} 
Let $\tau_0$ be the smallest positive integer
such that $\phi(\theta_{\tau_0}) \geq 0$. 
Every buyer $i$ whose value is at least 
$\theta_{\tau_0}$
sends ${\sf register}$ to the platform. 
Let $R$ be the set of buyers that have registered. 
\item 
{\it Auction:} 
{Each round $\tau  = \tau_0, \tau_0 + 1, \ldots, T$}:
\begin{itemize}[leftmargin=6mm,itemsep=1pt]
\item 
If a buyer $i$ has not received $({\sf stop}, \_)$ from the platform, 
it checks whether $v_i > \theta_\tau$. If so, 
send ${\sf ok}$ to the platform; else 
send $\bot$ to the platform, 
and stop sending messages in future rounds. 
\item 
Let $\mcal{I}$ be the set of buyers that either sent
$\bot$ in round $\tau$ or failed to respond in round $\tau$;
the platform lets $R \leftarrow R \backslash \mcal{I}$. 
\item 
If $|R| \leq k$, the platform 
posts $\tau$ to the blockchain and 
sends the following messages to the buyers and the seller: 
\begin{enumerate}[itemsep=1pt,label=(\alph*)]
\item
It sends $({\sf stop}, 1)$ to all buyers in $R$. 
\item 
If $\tau \neq \tau_0$, 
the platform randomly chooses $k - |R|$ buyers
from the set $\mcal{I}$ and sends $({\sf stop}, 1)$
to them, and it  
$({\sf stop}, k)$ to the seller. 
Else, it sends $({\sf stop}, |R|)$ to the seller. 
\item 
To all remaining buyers, the platform sends $({\sf stop}, 0)$. 
\end{enumerate}
\end{itemize}
\hspace{1mm}
At the end of round $T$, if $|R| > k$, 
the platform posts $\bot$ to the blockchain. 
\item
{\it Acceptance:}
Any buyer or the seller accepts the auction 
iff one of the following holds:
\begin{enumerate}[leftmargin=6mm,itemsep=1pt]
\item 
It received $({\sf stop}, \_)$ from the platform in some round $\tau$,
and the platform has posted 
the same $\tau$ to the blockchain. 
\item 
It has not received $({\sf stop}, \_)$ from the platform, 
and the platform has posted $\bot$ to the blockchain. 
\end{enumerate}
\item 
{\it Outcomes:}
The platform always accepts the auction
and it gets no revenue (or a fixed fee that is independent
of the auction). 
If a buyer or the seller
accepts the auction, it can determine its 
private outcome as follows\footnote{Or equivalently, the platform
sends the private outcomes to each player and each player
rejects if the received private outcome does
not agree with what it has computed on its own.}.
\begin{itemize}[leftmargin=6mm,itemsep=1pt]
\item 
If the platform has posted $\bot$ to the blockchain, then no item is sold.
\item 
Otherwise, let $\tau$ be the number posted to the blockchain. 
A buyer that has 
received $({\sf stop}, 1)$
gets an item and pays $\theta_\tau$, 
otherwise it does not get an item.  
Let $({\sf stop}, \widetilde{k})$ 
be the message the seller received from the platform earlier, 
the seller's private outcome  
is $(\widetilde{k}, \widetilde{k} \cdot \theta_\tau)$. 
\end{itemize}
\end{itemize}
\end{mdframed}

\begin{theorem}[Ascending auction with reserve and fixed platform fee]
The above ascending auction 
with reserve and fixed platform fee 
over a discretized value domain $\inDom$ 
satisfies the following properties:
\begin{itemize}[leftmargin=6mm,itemsep=1pt]
\item 
it satisfies (information theoretic) bIC, pIC, and 1-pbIC regardless of the 
value distribution $\mcal{D}$;
\item 
suppose buyers' true values are sampled
independently from a regular distribution $\mcal{D}$,  
then, the auction additionally satisfies
(information theoretic) Bayesian sIC and Bayesian psIC; 
\item 
suppose the buyers' true values are sampled independently 
from a regular distribution $\mcal{D}$, 
then, 
the auction is 
\elaine{double check expr}
$k \cdot {\sf tick}(\inDom)$-approximately revenue maximizing
where 
${\sf tick}(\inDom) := \max_{2 \leq j \leq T} (\theta_j - \theta_{j - 1})$. 
\end{itemize}
\label{thm:ascend}
\end{theorem}
In the above, $\epsilon$-approximately revenue maximizing
means that assuming that buyers' values are sampled  
independently from $\mcal{D}$, the expected revenue 
is at most $\epsilon$ away from the bIC auction 
that maximizes the expected revenue. 

The proof of \Cref{thm:ascend}
is deferred to  \Cref{sec:ascend-proof}.

\section{Impossibility Results}

\subsection{bIC + 1-pbIC $\Longrightarrow$ Almost Zero Platform Revenue}
We first show  
that any auction that satisfies bIC and pbIC (even in the Bayesian sense)
must suffer from almost no platform revenue.  
In fact, in the information theoretic setting, 
the platform revenue must be $0$. 
For the computational setting, there is a small slack
related to the discretization of  
the value domain. 
Typically, 
the value domain should have super-polynomially many ticks
to have sufficient precision in the encoding scheme, 
in this case, the resulting slack is negligibly small. 

\paragraph{Interpretation.}
Recall that in our formal 
model in \Cref{sec:platform-assisted-auction}, we require that 
the total revenue of the seller and the platform 
do not exceed the total payment. 
When this requirement is relaxed, the $0$ platform revenue restriction 
should actually be interpreted as 
having {\it fixed platform revenue}
that is independent of the auction's revenue. In practice,
it is possible to charge a fixed entry fee to the seller and possibly 
each (winning) buyer on top of the revenue earned from the auction.  

The proof
techniques are inspired
by the recent transaction fee mechanism 
literature~\cite{foundation-tfm,crypto-tfm}, and we extend
the techniques to work for computationally IC and multi-round mechanisms.


\paragraph{Additional notations.}
For some countable domain $\inDom \subset \mathbb{R}$, 
we say that $a, b \in \inDom$ are adjacent, denoted
$a \overset{\text{adj}}{\sim} b$ iff $\inDom \cap (a, b) = \emptyset$. 
We define 
\[
{\sf tick}(\inDom) := 
\max_{a, b \in \inDom, \ 
a \overset{\text{adj}}{\sim} b
}
|a - b| 
\]
Henceforth we will use 
the notation 
$\bfx(\bfb) \in [0, 1]^{|\bfb|}$
to denote 
the probabilities that each buyer gets an item 
under the value vector $\bfb$ 
(assuming that every one acts honestly).
Similarly, we use the 
random variables 
$\bfp(\bfb) \in \mathbb{R}_{\geq 0}^{|\bfb|}$
and $\mu_\seller (\bfb) \in \mathbb{R}_{\geq 0}$
to denote everyone's expected payment  as well as the seller's revenue 
under $\bfb$.  
We use $x_i(\cdot)$ and 
$p_i(\cdot)$ to denote the 
$i$-th coordinate of $\bfx$ and $\bfp$, that is, 
the $i$-th buyer's probability of getting an item and  
its expected payment.

\elaine{TODO: tick should be redefined when 0 is not part of domain}

\ignore{j
In the context of cryptographic protocols, 
it is natural to let the value domain $\inDomSecu$
to be parametrized by the security parameter $\lambda$. 
Intuitively, this also allows 
players to encode their values with more bits and higher precision 
with a larger choice of $\lambda$.

\begin{example}[Natural encoding]
For example,
we can define $\inDomSecu := \{0, 1\}^\lambda$
and use it 
to encode a real-valued domain $[0, M]$ as follows: 
for 
$i \in \{0, 1, \ldots, 2^\lambda-1\}$, the $i$-th string
encodes the value $i \cdot M/(2^\lambda -1)$. 
In this case, we say that $\inDomSecu$
is the natural encoding of $[0, M]$. 
\label{eg:naturaldomain}
\end{example}
}

\begin{lemma}
\label{lem:revenue-tick}
Let $\inDom \subset \mathbb{R}^{\geq 0}$ 
be a countable subset of non-negative real numbers.  
Let $x : \inDom \rightarrow [0, 1]$ 
be a monotonically non-decreasing function.
and let 
$p: \inDom \rightarrow \mathbb{R}^{\geq 0}$ and  
$\mu: \inDom \rightarrow \mathbb{R}^{\geq 0}$
be two other functions.
Suppose that 
for any $b, b' \in \inDom$ where  
$b < b'$, 
the following sandwich inequality holds:  
$b \cdot (x(b') - x(b)) \leq  p(b') - p(b) \leq b' \cdot (x(b')- x(b))$. 
Further, suppose that for any 
$b', b \in \inDom$, it holds that
$\mu(b') - \mu(b) \leq p(b')-p(b) - (b \cdot (x(b') - x(b)))$. 
%
Then, it holds that 
for any $b, b'$, $\mu(b') - \mu(b) \leq {\sf tick}(\inDom) \cdot |x(b') - x(b)|$.  

Further, if the above conditions hold for a real-valued domain 
$\inDom = \mathbb{R}^{\geq 0}$, 
then,  
for any $b, b'$, $\mu(b') = \mu(b)$. 
\end{lemma}
\begin{proof}
We prove it for the case when $b' > b$
since the proof for the other direction when $b' < b$ is similar. 
Let $b_0 := b$ and $b_k := b'$, 
and suppose $b_0 < b_1 < b_2 < \ldots < b_{k-1} < b_k$, 
where for every $i \in [k-1]$, $b_i \in \inDom$, and 
$b_1, \ldots, b_{k-1}$ are all the values
in $\inDom$ 
between $b_0$ and $b_k$. 
For convenience, let $\mu_i = \mu(b_i)$, $p_i = p(b_i)$, and 
$x_i = x(b_i)$ for $i \in \{0, 1, \ldots, k\}$. 

\begin{align*}
\mu_k - \mu_0 
& = \mu_k - \mu_{k-1} + \mu_{k-1} - \mu_{k-2} + \ldots + \mu_1 - \mu_0 \\
 & \leq
\sum_{i \in [k]}\left(
(p_{i} - p_{i-1}) - b_{i-1}(x_i - x_{i-1})
\right)  \\
& \leq
\sum_{i \in [k]}\left(
b_i(x_i - x_{i-1}) - b_{i-1}(x_i - x_{i-1})
\right)\\
& =
\sum_{i \in [k]}
(b_i - b_{i-1})(x_i - x_{i-1})
\\
& \leq
\sum_{i \in [k]}
{\sf tick}(\inDom) \cdot (x_i - x_{i-1}) \leq {\sf tick}(\inDom) \cdot (x_k - x_0).
\end{align*}

The above concludes the proof
when $\inDom$ is a countable subset of the non-negative reals.
For the continuous case when $\inDom$ is the non-negative reals, 
the conclusion that $\mu(b') = \mu(b)$ follows directly by taking 
${\sf tick}(\inDom)$ to be infinitesimally small. 
\end{proof}


\begin{theorem}[bIC + pIC + 1-pbIC $\Longrightarrow$ Almost zero platform revenue]
Let $\Pi_\lambda$ be an auction for selling $k$ items 
satisfies strong computational Bayesian bIC,
strong computational Bayesian pIC, 
and strong computational Bayesian $1$-pbIC w.r.t~$\mcal{D}_\lambda$,
where 
$\mcal{D}_\lambda$ is a distribution over   
some finite domain $\inDom_\lambda \subset [0,1]$.
Then, for any $\lambda$, 
the platform's expected revenue under a random value vector
drawn from $\mcal{D}_\lambda^n$ is at most 
${\sf tick}(\inDom_\lambda) \cdot k \cdot (\ln n + O(1))$.

Further, if $\Pi_\lambda$ satisfies 
strong computational bIC and strong computational $1$-pbIC 
\elaine{we don't need the pic here}
both in the ex post setting, 
then,  for any $\lambda$,
for any input vector $\bfb \in \inDomSecu^n$,
the platform's expected revenue is at most 
${\sf tick}(\inDom_\lambda) \cdot k \cdot (\ln n + O(1))$.
\label{thm:0platformrev}
\end{theorem}

As mentioned, by rescaling, we can without loss of generality
assume that $\inDomSecu$
is an encoding of a subset of $[0, 1]$. 
If we use $\lambda$ bits  
to encode evenly spaced points in $[0, 1]$, then
${\sf tick}(\inDomSecu) = \Theta(1/2^\lambda)$. 
Therefore, intuitively, the above theorem says
that as long as the encoded subspace of $[0, 1]$ 
is reasonably dense without leaving large gaps in between, 
then the platform's revenue should be negligibly small. 
Note also that for the Bayesian case, 
the above theorem additionally needs 
Bayesian pIC, whereas for the ex post version, the theorem
holds without requiring pIC. 
\elaine{should we add the counter example to appendix?}

\ignore{
\elaine{TODO: change this text}
For example, if $\inDomSecu$ is the natural
encoding (\Cref{eg:naturaldomain}) of 
the normalized range $[0, 1]$, then 
the above theorem implies that the platform's revenue
must be negligibly small since $n$ and $k$  must be polynomially bounded
in $\lambda$. 
We assume that the 
IC notions are satisfied in a strong sense 
in the above impossibility (see \Cref{defn:strongcompic}). 
This means that extended 
input replacement strategies which do not involve breaking
the cryptography strictly cannot do better than the honest strategy
(i.e., without the negligible slack).
One interpretation is that the 
${\sf tick}(\inDom) \cdot k \cdot (\ln n + O(1))$
amount 
of platform revenue 
is an artifact of the discretization of the real domain 
rather than a slack in the IC notions.
}

\begin{proof}[Proof of \Cref{thm:0platformrev}.]
Throughout the proof, we fix an arbitrary $\lambda$ and henceforth write 
$\mcal{D}:= \mcal{D}_\lambda$
and $\inDom := \inDom_\lambda$
for convenience. 

\elaine{ke says notation not defined}

\elaine{i changed bfxi to xi, may not be consistent}

\paragraph{Bayesian setting.}
We first prove the Bayesian case. 
Let $\mu(\bfb)$ denote the platform's expected revenue under the input $\bfb$ 
in an honest execution.
Given any buyer $i$, let 
$$\bar{x}_i(\cdot) := 
\displaystyle\mathop{\E}_{\bfb_{-i} \getr \mcal{D}^{n-1}}[x_i(\cdot, \bfb_{-i})],
\quad 
\quad 
\bar{p}_i(\cdot) := 
\displaystyle\mathop{\E}_{\bfb_{-i} \getr \mcal{D}^{n-1}}[p_i(\cdot, \bfb_{-i})], 
\quad 
\quad 
\bar{\mu} (\cdot) :=
\displaystyle\mathop{\E}_{\bfb_{-i} \getr \mcal{D}^{n-1}}[\mu(\cdot, 
\bfb_{-i})]$$
Due to strong computational bIC 
and Myerson's Lemma (\cref{lem:myerson-bayesian}), 
for any buyer $i$, and for any $b' > b$,
it must be $b \cdot (\bar{x}_i(b') - \bar{x}_i(b)) \leq  \bar{p}_i(b') - \bar{p}_i(b) \leq b' \cdot (\bar{x}_i(b')- \bar{x}_i(b))$, and 
$\bar{x}_i$ is a monotonically increasing. 

Suppose the platform colludes with a buyer with true value $b$. 
Due to strong computational Bayesian 1-pbIC, the 
platform's increase in expected revenue 
when the buyer bids $b'$ instead of $b$  
is upper bounded by the buyer's loss in utility. 
Otherwise, the buyer should strategically bid $b'$ instead of its true value
$b$ to increase the coalition's joint utility.  
In other words, 
\[
\bar{\mu}(b') - \bar{\mu}(b)
\leq 
(b \cdot \bar{x}_i(b) - \bar{p}_i(b))-
  (b \cdot \bar{x}_i(b') - \bar{p}_i(b')) 
 = 
\bar{p}_i(b') - \bar{p}_i(b) - b \cdot (\bar{x}_i(b') - \bar{x}_i(b))
\]

By \cref{lem:revenue-tick}, 
we have $\bar{\mu}(b') - \bar{\mu}(0) \leq {\sf tick}(\inDom) \cdot \bar{x}_i(b')$ for any $b'$. 

Next, by strong computational Bayesian pIC, injecting a bid $0$ should not increase the platform's revenue. Thus, 
we have 
$$\displaystyle\mathop{\E}_{\bfb_{-i} \getr \mcal{D}^{n-1}}[\mu(0, \bfb_{-i})] \leq \displaystyle\mathop{\E}_{\bfb_{-i} \getr \mcal{D}^{n-1}}[\mu(\bfb_{-i})]$$

Consequently, for any $n$, we have 
\begin{align*}
	\displaystyle\mathop{\E}_{\bfb \getr \mcal{D}^{n}}[\mu(\bfb)] 
& = \displaystyle\mathop{\E}_{b_i \getr \mcal{D}} \displaystyle\mathop{\E}_{\bfb_{-i} \getr \mcal{D}^{n-1}}[\mu(b_i, \bfb_{-i})] \\
& \leq \displaystyle\mathop{\E}_{b_i \getr \mcal{D}} 
\left(\displaystyle\mathop{\E}_{\bfb_{-i} \getr \mcal{D}^{n-1}}\left[\mu(0, \bfb_{-i})\right] + {\sf tick}(\inDom)
\cdot 
\displaystyle\mathop{\E}_{\bfb_{-i} \getr \mcal{D}^{n-1}}
\left[{x}_i(b_i, \bfb_{-i})
\right]
\right)
\\
& = 
\displaystyle\mathop{\E}_{\bfb_{-i} \getr \mcal{D}^{n-1}}\left[\mu(0, \bfb_{-i})\right] + 
{\sf tick}(\inDom)
\cdot 
\displaystyle\mathop{\E}_{\bfb \getr \mcal{D}^{n}}
\left[{x}_i(\bfb) 
\right]
\qquad \qquad \qquad (\star)
\ignore{
\\
	& \leq \displaystyle\mathop{\E}_{\bfb_{-i} \getr \mcal{D}^{n-1}}[\mu(0, \bfb_{-i})] + {\sf tick}(\inDom)\\
	& \leq \displaystyle\mathop{\E}_{\bfb_{-i} \getr \mcal{D}^{n-1}}[\mu(\bfb_{-i})] + {\sf tick}(\inDom)\\
	& \leq \displaystyle\mathop{\E}_{\bfb_{-i} \getr \mcal{D}^{n-2}}[\mu(0, \bfb_{-i})] + 2 \cdot {\sf tick}(\inDom)\\
	& \leq \cdots \leq \displaystyle\mathop{\E}_{b \getr \mcal{D}}[\mu(b)] + (n-1) \cdot {\sf tick}(\inDom)\\
	& \leq n \cdot {\sf tick}(\inDom).
}
\end{align*}
The above holds for any buyer $i$. 
We will pick an $i$ such that 
$\displaystyle\mathop{\E}_{\bfb \getr \mcal{D}^{n}}
\left[{x}_i(\bfb) \right]$ 
is minimized --- it is not hard to see
that in this case, 
$\displaystyle\mathop{\E}_{\bfb \getr \mcal{D}^{n}}
\left[{x}_i(\bfb) \right] \leq k/n$. 
Therefore, we have
\begin{align*}
(\star) & \leq 
\displaystyle\mathop{\E}_{\bfb_{-i} \getr \mcal{D}^{n-1}}\left[\mu(0, \bfb_{-i})\right] + 
{\sf tick}(\inDom) \cdot k/n \\
& \leq  
\displaystyle\mathop{\E}_{\bfb_{-i} \getr \mcal{D}^{n-1}}\left[\mu(\bfb_{-i})\right] + 
{\sf tick}(\inDom) \cdot k/n \\
\end{align*}
Applying the above inductively on $\displaystyle\mathop{\E}_{\bfb_{-i} \getr \mcal{D}^{n-1}}\left[\mu(\bfb_{-i})\right]$, 
and using the base condition that $\mu(0) = 0$, 
we have that  
\[
\displaystyle\mathop{\E}_{\bfb \getr \mcal{D}^{n}}\left[\mu(\bfb)\right]
\leq 
{\sf tick}(\inDom)
\cdot \left(\frac{k}{n} + \frac{k}{n-1}
+ \ldots + \frac{k}{1}  
\right) 
\leq 
{\sf tick}(\inDom)
\cdot k \cdot (\ln n + O(1))
\]
This concludes the proof of the Bayesian case.

\paragraph{Ex post setting.}
We now consider the ex post case. 
Redefine 
$\bar{x}_i(\bfb)$, 
$\bar{p}_i(\bfb)$, 
$\bar{\mu}_i(\bfb)$, 
to be $i$'s probability of getting an item, its expected payment,
and the platform's expected revenue 
under the value vector $\bfb$.
Using the same argument as before, we can prove  
that 
for an arbitrary $\bfb_{-i}$, any $b'$,
$\bar{\mu}(\bfb_{-i}, b') - \bar{\mu}(\bfb_{-i}, 0) \leq {\sf tick}(\inDom) 
\cdot \bar{x}_i(\bfb_{-i}, b')$. 

Now, given some value vector $\bfb = (b_1, \ldots, b_n)$, we have 
\begin{align*}
\bar{\mu}(\bfb) 
& \leq 
\bar{\mu}(\bfb_{-i}, 0) + {\sf tick}(\inDom) \cdot k/n\\
& \leq 
\bar{\mu}(\bfb_{-\{i, j\}}, 0, 0) 
+ {\sf tick}(\inDom) \cdot k/n
+ {\sf tick}(\inDom) \cdot k/(n-1)
\\
& \ldots 
\end{align*}
where the first line chooses an $i$  
such that $\bar{x}_i(\bfb)$
is minimized; and the second line chooses 
a $j \neq i$ such that $\bar{x}_j(\bfb_{-i}, 0)$ is minimized. 
It must be that 
$\bar{x}_i(\bfb) \leq k/n$;
and 
$\bar{x}_j(\bfb_{-i}, 0) \leq k/(n-1)$.

Carrying out the above derivation inductively, we will eventually
arrive at a world where everyone bids $0$,
and in this case, the platform's revenue is at most $0$, and thus
the conclusion holds. 
Note that in the ex post case, we do not need the strong computational pIC 
condition because we no longer need to remove the $0$ bids.  

\ignore{
First, it is not hard to see that if we assume
not just strong computational bIC and 1-pbIC, but also 
strong computational pIC (all in the ex post setting), 
then the conclusion in the theorem holds.
To see this, 
observe that our Bayesian proof can easily be extended
to the case where honest buyers' values are sampled
independently but not necessarily identically.  
The ex post case can be viewed as a special case of this where
each honest buyer's value distribution is just
a fixed number. 
Now, since strong computational bIC, pIC, and 1-pbIC guarantee 
incentive compatibility
under an arbitrary honest bid vector, the conclusion also holds
for an arbitrary honest bid vector.

To complete the proof, it remains to argue that the same conclusion 
still holds even when we remove the pIC requirement. 
In the above proof, we used strong computational Bayesian pIC to remove a $0$ bid from
the vector. 
However, for the ex post case, we need 
not remove the $0$ bid because the incentive compatibility guarantees
now hold for an arbitrary bid vector.  
So we can easily simplify 
the proof without using 
strong computational Bayesian pIC,  
and instead just carry the $0$ bids around.  
At the end of the inductive proof, we will be left
with a world that consists of solely 0 bids which becomes our new base condition.
It is not hard to see that under this base condition, the platform's revenue must
be $0$, and which can be readily used to replace
the base condition in the above Bayesian proof. 
}
\end{proof}


The proofs in this section also directly imply
the following corollary that says the platform's revenue
must be $0$ if the value domain is continuous over the non-negative reals.

\begin{corollary}
Let $\Pi$ be an auction for selling $k$ items 
over a continuous, real-valued domain. 
Suppose that $\Pi$
satisfies Bayesian bIC,
Bayesian pIC, 
and Bayesian $1$-pbIC w.r.t. the distribution $\mcal{D}$. 
Then, the platform's expected revenue under a random value vector
drawn from $\mcal{D}^n$ is $0$. 
Further, if $\Pi$ satisfies 
bIC and $1$-pbIC both in the ex post setting, 
then,  
the platform's revenue is $0$ for any value vector.
\label{cor:0platformrev}
\end{corollary}



\subsection{bIC + 1-pbIC $\Longrightarrow$ Public Outcome Necessary}
\label{sec:puboutcome}

\paragraph{Proof roadmap.}
We next show that any auction with non-trivial utility that satisfies
bIC and 1-pbIC (even in the Bayesian sense) 
must make use of a broadcast channel. 
Informally speaking, the intuition behind our proof
is as follows. We will use the ex post, information theoretic
setting to explain the intuition, but our actual proof
later generalizes it to even Bayesian and computational settings.   
Suppose there is a value vector $\bfb = (b_1, \ldots, b_n)$  
such that user $b_i$ gets $u > 0$ utility. 
Imagine that 
the world consists of not only $\bfb = (b_1, \ldots, b_n)$ 
but also many buyers
indexed $n+1, n+2, \ldots, n+K$ 
with the true value $b_i$. 
Since the number of buyers is much greater than the  
number of items $k$, one of the buyers with value  
$b_i$, 
denoted $U$, 
must get smaller than $u$ expected utility 
(as long as $K$ is sufficiently large).
Now, if $U$ is colluding with the platform,
the coalition can 
adopt the following strategy.  
To the buyers 
$[n]\backslash \{i\}$ as well as the seller, 
the coalition pretends that the world consists
of $\bfb$
where the $i$-coordinate is replaced with $U$; 
to every other buyer $j \in \{i, n+1, \ldots, n+K\} \backslash \{U\}$, 
the coalition pretends
that the world consists of an extremely large number of  
buyers all with the same value, such that with high probability, 
none of these buyers get an item. Note that to make this strategy
work, the coalition needs to simulate 
many fake buyers and the seller to the 
buyers $j \in \{i, n+1, \ldots, n+K\} \backslash \{U\}$. 
Because bIC and pbIC implies that the platform's revenue
must be $0$ (\Cref{cor:0platformrev}), 
\elaine{TODO: write this corollary}
this strategy 
benefits the coalition positively by ensuring that $U$ 
can get utility $u$. This violates 1-pbIC.

It is easy to make the above argument work
if we assume that 1) the mechanism satisfied {\it strong symmetry}, i.e.,
two buyers with the same true value have the same outcome distribution;
and 2) the value vector $\bfb$ consists of distinct values.  
A major technical challenge in our subsequent formal proofs 
is how to 
make the argument work 
relying only on {\it weak symmetry} (see \Cref{defn:weaksym}
of \Cref{sec:platform-assisted-auction}), and without 
assuming the distinctness of $\bfb$. 

Note that having a broadcast channel can make this impossibility goes away,
because players can use the broadcast channel to corroborate 
their views in the protocol, such that the platform-buyer coalition
is no longer able to present a different world to different honest players.

\paragraph{Ex post setting.}
We say that an auction $\Pi_\lambda$
over the value domain $\inDomSecu \subset [0, 1]$  
is non-trivial, iff  
there exists some polynomially bounded function $\poly(\cdot)$, 
such that 
for infinitely many $\lambda$'s, 
there exists some value vector $\bfb \in \inDomSecu^*$
whose length is polynomially bounded in $\lambda$, 
and the expected total buyer utility 
under $\bfb$ is at least $1/\poly(\lambda)$.  



\ignore{
An auction is said to be \emph{trivial} if 
there exists a negligible function $\negl(\cdot)$ such that 
for any $\lambda$, for any $n$ (that is polynomially bounded in $\lambda$), 
for any value vector $\bfv = (v_1, \ldots, v_n) \in \inDom^n_\secu$,  
for any $i$, 
\[
\Pr[\util_i(1^\lambda, \bfv) \leq v_i \cdot \negl(\lambda)] \geq 1 - \negl(\lambda),
\]
\elaine{todo: change the constants}
where $\util_i(1^\lambda, \bfv)$ denote buyer $i$'s expected utility 
under an honest execution using the value vector $\bfv$ as input. 
For example, if 
the value domain is normalized to $[0, 1]$, it means
all buyers get negligible utility with all but with negligible probability.
}

\begin{theorem}[bIC + pbIC $\Longrightarrow$ public outcome necessary]
Let $\Pi_\lambda$ be a non-trivial auction over the value domain
$\inDomSecu \subset [0, 1]$. 
Suppose that there is no public-key infrastructure (PKI)
and the players are not allowed to post any public outcome
to the blockchain.  Further, suppose that 
there exists some negligible function $\negl(\lambda)$ such
that for every $\lambda$, 
${\sf tick}(\inDomSecu)\leq \negl(\lambda)$. 
Then, 
$\Pi_\lambda$ cannot simultaneously satisfy 
strong computational bIC and strong computational $1$-pbIC.
\label{thm:publicoutcome}
\elaine{TODO: double check}
\end{theorem}

In the above, 
the condition 
${\sf tick}(\inDomSecu) \leq \negl(\lambda)$ is very mild. 
For example, if $\inDomSecu$ 
uses $\lambda$ bits to encode evenly spaced
points between $[0, 1]$, then 
${\sf tick}(\inDomSecu) = \Theta(1/2^\lambda)$. 
Intuitively, this condition wants the encoding
of the normalized value domain $[0, 1]$ to be sufficiently
dense, without leaving large gaps in between.

\begin{proof}(of \Cref{thm:publicoutcome}.)
Due to non-triviality, 
there exists some polynomial $q(\cdot)$, 
such that 
for infinitely many $\lambda$'s, 
there exists
some value vector $\bfb = (b_1, \ldots, b_n) \in \inDomSecu^n$ 
where $n$ is polynomially bounded in $\lambda$, 
and buyer identities ${\it id}_1, \ldots, {\it id}_n$ 
(henceforth called Scenario 1), 
and under $\bfb$ and these identities, 
the $i$-th buyer
has expected utility at least $1/q(\lambda)$. 

Recall that the number of items is $k$. Therefore,
if the world consists of 
$K = k^2 \cdot q^4(\lambda)$ 
bids 
all at $b_i$, then  for at least one of these $K$
buyers (henceforth denoted buyer $U$), 
its expected probability of getting an item
is at most $k/K$. 
The choice of this buyer $U$ may depend on auxiliary identity information
--- 
in practice, this the identity can contain information like cryptographic keys,
time-of-arrival, and any other information. 
It must be that at least half of the identities 
in the identity space (denoted the set $\mcal{J}$) satisfy the following: 
for any $j \in \mcal{J}$, 
if $j$ participates in an auction with $K-1$ other bidders
with random distinct identities, and everyone bids $b_i$,
then $j$'s probability of getting an item is at most $2k/K$. 
\elaine{TODO: this might need a proof?}

Let $n'$ be the number of bids in $\bfb = (b_1, \ldots, b_n)$ 
that are equal to $b_i$. Henceforth we use the term
$b_i$-bids to refer to the bids in $\bfb$ that are equal
to $b_i$, and the term   
non-$b_i$-bids to refer to the bids in $\bfb$ that are not equal
to $b_i$. 
Pick $n'-1$ distinct identities $\mcal{J}' \subset \mcal{J}$
and a distinct $j^* \in \mcal{J}$ that satisfy the following condition:
when the world consists 
of $(b_1, \ldots, b_n)$ where all non-$b_i$-bids 
take the corresponding identities of Scenario 1, 
and the $b_i$-bids take on 
the identities
$\mcal{J}' \cup \{j^*\}$, then $j^*$ has  
at least $1/q(\lambda)$
expected utility.
This is possible due to weak symmetry. 
\elaine{is this defined?}

Now, imagine the following scenario.  
\begin{enumerate}[leftmargin=6mm]
\item 
There are $n$ buyers 
whose true values are  
$b_1, \ldots, b_n$ respectively.
All the non-$b_i$-bids take the corresponding identities as in Scenario 1,
and the $b_i$-bids take the identities $\mcal{J}' \cup \{V\}$
where $V$ is an arbitrary fixed identity in $\mcal{J}$ distinct
from all the other identities that have been chosen.  
Henceforth $V$ is also called the victim.

\item 
Additionally, 
there are $K' = k q^2(\lambda)$ buyers
whose identities are chosen arbitrarily from $\mcal{J}$ as long
as they are distinct from all other consumed identities; and 
all of these $K'$ buyers have the true value $b_i$. 
Among these $K'$ buyers, at least one of them denoted $U$ must get an item
with probability at most $k/K'$.  
\end{enumerate}
Suppose $U$ and the platform are colluding. 
Then, the coalition, which has the advice string $j^*$ hardwired
in its algorithm, can  
adopt the following polynomial-time strategy:  
\begin{itemize}
\item 
The platform not only acts on behalf of itself,
but also simulates a buyer with identity $j^*$ using input value $b_i$.
In this way, the platform 
interacts with the first $n$ buyers except $V$ as well as the seller. 
The messages from 
$V$ and the remaining $K'-1$ honest buyers bidding $b_i$ are 
not passed into this session. 
If the simulated buyer $j^*$ gets an item, this item is actually allocated  
to the colluding user $U$. 
\item 
With each of $V$ and the remaining $K'-1$ players bidding $b_i$, 
the platform forks $K'$ sessions, and each session does 
the following: 
the platform creates $K-1$ random fake buyer identities 
and all of them use the input $b_i$. The platform
simulates (the honest behavior of) these fake buyers and the seller 
in its head, and 
interacts with the buyer. 
\end{itemize}

The execution trace is safe as long as 
\begin{enumerate}[leftmargin=7mm,itemsep=1pt]
\item 
in the sessions with $V$ and the $K'-1$ buyers bidding $b_i$, the
randomly sampled identities do not collide with honest buyers' identities
(see \Cref{rmk:stealid}) --- this happens with negligible probability
as long as the identity space is super-polynomial in $\lambda$; 
\elaine{TODO: add the assumption that id space super poly}
and 
\item $V$ and the $K'-1$ buyers bidding $b_i$ 
all fail to get an item in their respective sessions. 
\end{enumerate}
Therefore, the probability of a safe trace 
is at least 
\[
p_{\rm safe}
\geq 1 - (2k/K) \cdot K' -\negl(\lambda) \geq 1 - 3/q^2(\lambda)
\]

Under this strategy, $U$'s expected gain in utility 
is at least 
$$
\frac{1}{q(\lambda)} \cdot p_{\rm safe} - b_i \cdot \frac{k}{K'}
\geq 
\frac{1}{q(\lambda)} 
\cdot (1-3/q^2(\lambda)) - 
\frac{1}{q^2(\lambda)}
\geq 
\frac{1}{2q(\lambda)} 
$$
The platform's expected loss in 
utility is at most
${\sf tick}(\inDom) \cdot k \cdot (\ln (n + K') + O(1))$
due to \Cref{thm:0platformrev}. 
Due to the assumption on 
${\sf tick}(\inDom)$, we have that 
$1/{2q(\lambda)} - 
{\sf tick}(\inDom) 
\cdot k \cdot (\ln (n + K') + O(1))
\geq 1/\poly(\lambda)
$ 
for sufficiently large $\lambda$, 
i.e., 
for infinitely many $\lambda$'s, 
the coalition can gain at least $1/\poly(\lambda)$  
in expected utility 
through this  
strategy, which violates computational 1-pbIC. 
\elaine{do we need strong}
\end{proof}

\begin{remark}
We need the no PKI 
assumption in the above proof, 
since the strategic coalition needs to simulate 
the seller when interacting with some honest buyers
--- this requires that the coalition steals
the identity of the seller. 
We stress that the proof is constructed such that the coalition
need not 
steal any honest buyer's identity.
This means that the impossibility also holds
if the ``no PKI'' assumption is replaced
with the requirement that the seller does not
actively send messages in the auction.
The latter
is also a natural assumptions since the seller's intention 
is to outsource the auction to the platform. 
\elaine{refer to practical instantiation where seller need not send msg}
\label{rmk:stealid}
\end{remark}



\paragraph{Bayesian setting.}
We modify the non-trivial definition above
for the Bayesian setting. 
We say that an auction $\Pi_\lambda$ over the value domain
$\inDomSecu \subset [0, 1]$ is non-trivial
w.r.t. 
the distribution 
$\mcal{D}_\lambda$ over $\inDomSecu$, 
iff there exists some 
polynomial function $1/\poly(\lambda)$, 
such that for 
infinitely many $\lambda$'s, 
there exists $n(\lambda)$ buyer identities denoted $\mcal{I}$ where
$n(\cdot)$ is polynomially bounded in $\lambda$, 
and the expected utility of all buyers in $\mcal{I}$ is at least
$1/\poly(\lambda)$ 
when 
every one draws its  true value independently
from $\mcal{D}_\lambda$.

\begin{corollary}
Suppose some auction $\Pi_\lambda$ over the value domain
$\inDomSecu \subset [0, 1]$ is non-trivial
w.r.t. some polynomial-time samplable 
distribution $\mcal{D}_\lambda$ over $\inDomSecu$, and further, 
there is some negligible function $\negl(\cdot)$ such that for every
$\lambda$, ${\sf tick}(\inDomSecu) \leq \negl(\lambda)$.  
Suppose that there is no public-key infrastructure (PKI)
and the players are not allowed to post any public outcome
to the blockchain.  
Then, $\Pi_\lambda$ 
cannot simultaneously satisfy 
strong computational Bayesian bIC,  
strong computational Bayesian pIC,  
strong computational Bayesian pbIC w.r.t. 
$\mcal{D}_\lambda$ over $\inDomSecu$. 
\elaine{and the input replacement notions}
\end{corollary}
\begin{proof}
The proof is similar to that of \Cref{thm:publicoutcome}. 
Let $\mcal{J}$ be a set of superpolynomially many identities
such that for $i \in \mcal{J}$, when we sample $K$ 
distinct identities at random including $i$, $i$'s probability
of getting an item is at most $2k/K$ assuming everyone samples
their true value from $\mcal{D}_\lambda$.

Now, sample $n(\lambda) -1$ identities $\mcal{J}' \subset \mcal{J}$
and some fixed $j^*$ such that when 
$\mcal{J}' \cup \{j^*\}$ participate in an auction together
using input values independently sampled from $\mcal{D}_\lambda$, 
$j^*$ has expected utility at least
$1/(\poly(\lambda)\cdot n(\lambda))$. 

Imagine a world consisting of the identities
$\mcal{J}'$,  
some victim $V$
with an arbitrary distinct identity in $\mcal{J}$, 
and moreover $K'$ additional buyers  
with arbitrary distinct identities from $\mcal{J}$.
Among the last  $K'$ buyers, some buyer (henceforth called $U$)
must 
have at most $k/K'$ probability of getting an item if everyone
draws their value at random from $\mcal{D}_\lambda$.  

Now, imagine that the platform colludes with $U$, and they adopt
the following polynomial-time strategy:
\begin{itemize}[leftmargin=6mm,itemsep=1pt]
\item 
The platform simulates itself and a player $j^*$, and interacts
with $\mcal{J}'$; all real and simulated buyers draw  
their 
true value independently from 
$\mcal{D}_\lambda$; 
\item 
With each of the victim $V$, as well as the $K'-1$ remaining buyers,
the platform forks 
$K'$ sessions 
and each session does 
the following: 
the platform creates $K-1$ random fake buyer identities 
and all of them use a random input from $\mcal{D}_\lambda$. The platform
simulates (the honest behavior of) these fake buyers and the seller 
in its head, and interacts with the buyer. 
\end{itemize}
The remaining calculations are similar to that of \Cref{thm:publicoutcome}.
For the Bayesian version, we additionally need
strong computational Bayesian pIC because
\Cref{thm:0platformrev} which we rely on needs it for the Bayesian setting. 
\end{proof}

\subsection{bIC + 2-pbIC $\Longrightarrow$ Impossibility}

\paragraph{Proof roadmap.}
We want to show that it is not possible to ask for bIC and 2-pbIC
at the same time.  
Our proof is inspired by the techniques of 
Shi et al.~\cite{crypto-tfm}.
We describe the informal intuition
assuming the information theoretic and ex post settings, and assuming
strong symmetry. 
Our actual proof later generalizes it to computational  
and Bayesian settings, and replaces
the strong symmetry assumption with weak symmetry (see \Cref{defn:weaksym}).

Due to \Cref{cor:0platformrev}, 2-pbIC (which implies 1-pbIC) plus
bIC implies that the platform's revenue is always $0$. 
This fact, plus 2-pbIC, implies 2-bIC, that is, 
honest behavior maximizes
the utility for any coalition of $2$ buyers. 
Due to the elegant work of Goldberg and Hartline~\cite{goldberghartline}, 
we know that any 2-bIC 
auction must be utility-equivalent to posted price, that is, what buyer $j$'s
true value is has no effect on buyer $i$'s 
utility (assuming that every one
acts honestly). 
Therefore, we can lower buyer $j$'s true 
value to $0$ without affect buyer $i$'s utility. 
We then argue that because of 2-pbIC, if we drop a buyer $j$ whose 
true value is $0$, user $i$'s utility must be unaffected too
(\Cref{lem:2pbic-imported}). 
If we repeat this argument and lower every buyer's true value
to $0$, and then make the buyer drop, we conclude
that buyer $i$'s utility is the same in a world
with only buyer $i$, and in 
a world with other buyers with arbitrary values. 
In particular, in a crowded world with many buyers all having the
same value as $i$, 
due to strong symmetry, buyer $i$'s utility goes to $0$. 
This means that any buyer's utility must be $0$ in any scenario. 

We now present the formal proofs where we remove the strong symmetry
assumption and generalize the argument to computational and Bayesian settings.

\begin{lemma}
Suppose an auction parametrized by the parameter $\lambda$ satisfies 
strong computational Bayesian $2$-pbIC w.r.t~$\mcal{D}_\lambda$,
where 
$\mcal{D}_\lambda$ is a distribution over   
some finite domain $\inDom_\lambda \subset \mathbb{R}_{\geq 0}$.
Then, for any $\lambda$, 
for any $n$ that is polynomially bounded in $\lambda$,
for any set $\mcal{I}$ of $n+2$ identities, 
for any distinct $i, j \in \mcal{I}$
it holds that 
\begin{equation}
\displaystyle\mathop{\E}_{\bfb_{-j}\getr \mcal{D}_\lambda^{n+1}}
\left[\pbutil_i(0_j, \bfb_{-j})
 \right]
       = \displaystyle\mathop{\E}_{\bfb_{-j}\getr \mcal{D}_\lambda^{n+1}}
\left[\pbutil_i(\bfb_{-j})
 \right]
 \end{equation}
where $\pbutil_i(\bfb'_{\mcal{J}})$
 denotes the expected joint utility of the platform and buyer $i$
when bidders with identities from $\mcal{J}$
take on the value vector $\bfb'_{\mcal{J}}$.
\label{lem:2pbic-imported}
\end{lemma}
In the above lemma, when we write $v_i$ or $0_j$, the subscripts
indicates the identity 
of the buyer. We use $-j$ to denote $\mcal{I}\backslash \{j\}$,
and the notation $\bfb_{-j}$ 
means the bids in this vector take on the identities in the 
set $\mcal{I}\backslash \{j\}$.
\begin{proof}(of \Cref{lem:2pbic-imported})
The proof is similar to the proof of Lemma 5.3 in Appendix C.3 
of Shi et al.\cite{crypto-tfm}.
Below we provide a simpler version of the proof adapted to our context.

\paragraph{The $\geq$ direction.} 
It is easy to show that the left-hand 
side is greater than or equal to the right-hand side. 
If the right-hand side is strictly larger, 
the following holds for some specific
$v_i$:
$\displaystyle\mathop{\E}_{\bfb_{-\{i,j\}}\getr \mcal{D}_\lambda^{n}}
\left[\pbutil_i(v_i, 0_j, \bfb_{-\{i,j\}})
 \right]
      < \displaystyle\mathop{\E}_{\bfb_{-\{i,j\}}\getr \mcal{D}_\lambda^{n}}
\left[\pbutil_i(v_i, \bfb_{-\{i,j\}})
 \right]$. This means that if buyer $i$ with true value $v_i$
and buyer $j$ with true value $j$ both 
collude with the platform, the coalition  
is better off having $j$ drop off, which violates 2-pbIC.  

\paragraph{The $\leq$ direction.} 
Below, we focus on proving that 
the left-hand side is smaller than or equal
to the right-hand side.
Suppose the lemma is not true, i.e., there exists
a $\lambda$, 
a set of $n + 2$ identities $\mcal{I}$ 
such that for some $i, j \in \mcal{I}$,  
the following holds --- henceforth, for convenience
let $\mcal{D} := \mcal{D}_\lambda$: 
\begin{equation}
\displaystyle\mathop{\E}_{\bfb_{-j} \getr \mcal{D}^{n+1}}
\left[\pbutil_i(0_j, \bfb_{-j})
\right]
> 
\displaystyle\mathop{\E}_{\bfb_{-j}\getr \mcal{D}^{n+1}}
\left[\pbutil_i(\bfb_{-j})
\right].
\label{eqn:inject0}
\end{equation}
If the coalition can hijack the identity $j$, then 
we can easily show that the above 
\Cref{eqn:inject0} is not possible as follows. 
Given \Cref{eqn:inject0}, it must be that there exists
some specific $v_i$
such that 
\begin{equation}
\displaystyle\mathop{\E}_{\bfb_{-\{i,j\}}\getr \mcal{D}^{n}}
\left[\pbutil_i(v_i, 0_j, \bfb_{-\{i,j\}})
\right]
> 
	\displaystyle\mathop{\E}_{\bfb_{-\{i,j\}}\getr \mcal{D}^{n}}
\left[\pbutil_i(v_i, \bfb_{-\{i,j\}})
\right].
\label{eqn:inject0-fixvi}
\end{equation}
This  means that if the buyer $i$ 
with value $v_i$ colludes with the platform, and the
rest of the world consists of 
random bids from $\mcal{D}$ under the identities 
$- \{i, j\}$, 
then the coalition should simulate a fake identity $j$
whose true value is $0$, 
which improves their joint utility. This violates strong computational $2$-pbIC. 

In the remainder of the proof, 
we will show that the lemma holds even when 
the strategic coalition cannot hijack the fixed identity $j$. 
Suppose that 
some identity $m$ is under the control of the strategic platform. 
It must be that 
\begin{align*}
\displaystyle\mathop{\E}_{\bfb_{-j}\getr \mcal{D}^{n+1}}
\left[\pbutil_i(0_m, \bfb_{-\{i,j\}})
\right]
 &\leq
\displaystyle\mathop{\E}_{\bfb_{-j}\getr \mcal{D}^{n+1}}
\left[\pbutil_i(\bfb_{-j})
\right]
< 
\displaystyle\mathop{\E}_{\bfb_{-j}\getr \mcal{D}^{n+1}}
\left[\pbutil_i(0_j, \bfb_{-j})
\right]
\end{align*}
where the second inequality 
is due to \Cref{eqn:inject0}, and the first inequality can be 
shown 
using the same argument as above: if not, we can reach a contradiction
by showing that the platform-buyer $i$ coalition can improve their
utility by injecting a fake bid $0_m$.

Due to weak symmetry, 
the platform's expected revenue does not change when one $0$ bid changes
its identity from $j$ to $m$. Thus, 
we have that 
$
\displaystyle\mathop{\E}_{\bfb_{-j}\getr \mcal{D}^{n+1}}
\left[{\sf util}_i(0_m, \bfb_{-j})
\right]
< 
\displaystyle\mathop{\E}_{\bfb_{-j}\getr \mcal{D}^{n+1}}
\left[{\sf util}_i(0_j, \bfb_{-j})
\right]$
where ${\sf util}_i(\bfb'_\mcal{J})$
denotes buyer $i$'s expected utility when 
the buyers with identities $\mcal{J}$ adopt 
the value vector $\bfb'_{\mcal{J}}$.
Due to weak symmetry, we claim 
that there exists some $i' \in \mcal{I}\backslash \{i, j\}$, 
such that 
$
\displaystyle\mathop{\E}_{\bfb_{-j}\getr \mcal{D}^{n+1}}
\left[\util_{i'}(0_m, \bfb_{-j})
\right]
> 
\displaystyle\mathop{\E}_{\bfb_{-j}\getr \mcal{D}^{n+1}}
\left[\util_{i'}(0_j, \bfb_{-j})
\right]$. 
This is because by weak symmetry, 
the expected total buyer utility should be the same when one $0$-bid changes
its identity from $j$ to $m$.  
\ignore{
Therefore, we have 
$$
\displaystyle\mathop{\E}_{\bfb_{-j})\getr \mcal{D}^{n+1}}
\left[\pbutil_{i'}(0_m, \bfb_{-j})
\right]
> \displaystyle\mathop{\E}_{\bfb_{-j})\getr \mcal{D}^{n+1}}
\left[\pbutil_{i'}(0_j, \bfb_{-j})
\right]$$
}
This implies that there exists some specific choice of $v_{i'}$, such that 
$$\displaystyle\mathop{\E}_{\bfb_{-\{i',j\}}\getr \mcal{D}^{n}}
\left[\pbutil_{i'}(v_{i'}, 0_m, \bfb_{-\{i',j\}})
\right]
	> \displaystyle\mathop{\E}_{\bfb_{-\{i',j\}}\getr \mcal{D}^{n}}
\left[\pbutil_{i'}(v_{i'}, 0_j, \bfb_{-\{i',j\}})
\right]$$
Now, if the buyer $i'$
has true value $v_{i'}$, the buyer $j$
has true value $0$, and both buyers
collude with the platform,
and the rest of the world consists
of identities $\mcal{I}\backslash\{i', j\}$
and their true values  
are randomly drawn from $\mcal{D}$, 
then the coalition is better off if $j$ simply drops offline,
and the coalition simulates a fake identity $m$ with true value $0$, 
which
violates strong computational 2-pbIC.
\end{proof}

\begin{lemma}
\label{lem:2pbic-revenue-gap}
Let $\Pi_\lambda$ be an auction over the value domain $\inDomSecu$
which satisfies strong computational Bayesian bIC 
and strong computational Bayesian $2$-pbIC w.r.t~$\mcal{D}_\lambda$,
which is a distribution over  $\inDomSecu$.
Then, for any $\lambda \in \N$, 
for any set $\mcal{I}$ of $n$ identities
where
$n$ is polynomially bounded in $\lambda$, 
and for any buyer $i \in \mcal{I}$, 
it must be that \[
\left|
\mathop{\E}_{\bfb \getr \mcal{D}_\lambda^{n}}[\util_i(\bfb)]
-
\mathop{\E}_{v_i \getr \mcal{D}_\lambda}[\util_i(v_i)]
\right|
\leq 2 \cdot {\sf tick}(\inDomSecu) \cdot k \cdot (\ln n + O(1)).
\]
\label{lem:2pbic-key}
\end{lemma}
\begin{proof}
Throughout the proof, we fix an arbitrary $\lambda$, $n$, and 
an arbitrary buyer $i$.  
For simplicity, we write 
$\mcal{D} = \mcal{D}_\lambda$, $\mcal{\inDom} = \inDomSecu$,
and $-j = \mcal{I}\backslash \{j\}$ 
below in this proof.
Given any buyer $j$, and any bid $b_j$,
we define
\begin{align*}
	\bar{x}_{j}(b_j) &\coloneqq 
\underset{\bfb_{-j} \getr \mcal{D}^{n-1}}{\E}[{x}_j(b_j, \bfb_{-j})]\\
	\bar{p}_{j}(b_j) &\coloneqq \underset{\bfb_{-j} \getr \mcal{D}^{n-1}}{\E}[{p}_j(b_j, \bfb_{-j})]\\
	\bar{\mu}(b_j)   &\coloneqq \underset{\bfb_{-j} \getr \mcal{D}^{n-1}}{\E}[\mu(b_j, \bfb_{-j})]
\end{align*}

\begin{claim}
For any $i$, $j$ and $b_j$, 
$$
\left|\mathop{\E}_{\bfb_{-j}
\getr \mcal{D}^{n-1}}[\util_i(b_j, \bfb_{-j})]
-
\mathop{\E}_{\bfb_{-j}
\getr \mcal{D}^{n-1}}[\util_i(0_j, \bfb_{-j})]
\right|
\leq 2 \cdot {\sf tick}(\inDom) \cdot \bar{x}_j(b_j)$$
\label{clm:2pbic-helper}
\end{claim}
\begin{proof}
\elaine{0platform rev requires pic, use proof non black box}
Using the same proof as that of \Cref{thm:0platformrev}, 
due to strong computational bIC and 1-pbIC (which is implied by 2-pbIC), 
we have that for any $b_j$, 
\begin{equation}
|\bar{\mu}(b_j) - \bar{\mu}(0_j)| \leq 
{\sf tick}(\inDom) \cdot \bar{x}_i(b_j)
\label{eqn:prevdiff}
\end{equation}
Similarly, using the 
technique as the proof of \Cref{thm:0platformrev}, 
when the platform, $i$, and $j$ form a coalition,
the increase in the joint utility of the platform and $i$ should
not exceed the loss in $j$'s utility 
when $j$ bids strategically. 
Therefore, we can show that 
for any $b_i, b_j$, 
\[
\begin{array}{ll}
\left|
\displaystyle\mathop{\E}_{\bfb_{-\{i, j\}} \getr \mcal{D}^{n-2}} 
\left[
\pbutil_i(b_i, 0_j, \bfb_{-\{i, j\}}) 
\right]
- 
\displaystyle\mathop{\E}_{\bfb_{-\{i, j\}} \getr \mcal{D}^{n-2}} 
\left[
\pbutil_i(b_i, b_j, \bfb_{-\{i, j\}})
\right]
\right| 
\\[8pt]
\qquad \qquad \qquad \qquad \qquad \qquad \qquad \qquad \qquad \qquad
\leq 
{\sf tick}(\inDom) \cdot 
\displaystyle\mathop{\E}_{\bfb_{-\{i, j\}} \getr \mcal{D}^{n-2}} 
\left[x_j(b_i, b_j, \bfb_{-\{i, j\}})
\right]
\end{array}
\]
where $\pbutil_i(\bfb)$ denotes the joint utility of the platform and  
buyer $i$ under $\bfb$.
Thus,  taking expectation over $b_i$,  
we have 
\begin{equation}
\left|
\mathop{\E}_{\bfb_{-j} \getr \mcal{D}^{n-1}} 
\left[
\pbutil_i(0_j, \bfb_{-\{i, j\}}) 
\right]
- 
\mathop{\E}_{\bfb_{-j} \getr \mcal{D}^{n-1}} 
\left[
\pbutil_i(b_j, \bfb_{-\{i, j\}})
\right]
\right| \leq 
{\sf tick}(\inDom) \cdot \bar{x}_i(b_j)
\label{eqn:pbutildiff}
\end{equation}

The claim follows by combining \Cref{eqn:prevdiff,eqn:pbutildiff}. 

\ignore{
Suppose that the claim is not true
for the sake of contradiction, that is,
$
\left|
\mathop{\E}_{\bfb_{-j}
\getr \mcal{D}^{n-1}}[\util_i(b_j, \bfb_{-j})]
-
\mathop{\E}_{\bfb_{-j}
\getr \mcal{D}^{n-1}}[\util_i(0_j, \bfb_{-j})]\right|
> 2 \cdot {\sf tick}(\inDom) \cdot \bar{x}_j(b_j, \bfb_{-j})$. 
Without loss of generality, 
suppose 
$\mathop{\E}_{\bfb_{-j}
\getr \mcal{D}^{n-1}}[\util_i(b_j, \bfb_{-j})]
> \mathop{\E}_{\bfb_{-j}
\getr \mcal{D}^{n-1}}[\util_i(0_j, \bfb_{-j})]$ since the other way
has a symmetric proof. 
In this case, if buyer  $j$ has true value $0$
and colludes with the platform, the coalition should have
buyer $j$ use the untruthful input $b_j$ instead, and the coalition
is better off under this input replacement strategy.
}
\end{proof}

We  now continue with the proof of 
\Cref{lem:2pbic-revenue-gap}.
By ~\Cref{lem:2pbic-imported} and \Cref{clm:2pbic-helper}, 
we have 
that  for all $b_j$, 
$$
\left|\mathop{\E}_{\bfb_{-j}
\getr \mcal{D}^{n-1}}[\util_i(b_j, \bfb_{-j})]
-\mathop{\E}_{\bfb_{-j}
\getr \mcal{D}^{n-1}}[\util_i(\bfb_{-j})]\right| \leq 
2 \cdot {\sf tick}(\inDom) \cdot \bar{x}_j(b_j) 
$$
Therefore, 
$$
\left|\mathop{\E}_{\bfb
\getr \mcal{D}^{n}}[\util_i(\bfb)]
-\mathop{\E}_{\bfb_{-j}
\getr \mcal{D}^{n-1}}[\util_i(\bfb_{-j})]\right|
\leq 2 \cdot {\sf tick}(\inDom) \cdot 
\mathop{\E}_{\bfb
\getr \mcal{D}^{n}}[{x}_j(\bfb)]
$$
If we pick $j \neq i$ to be the one such that 
$\mathop{\E}_{\bfb
\getr \mcal{D}^{n}}[{x}_j(\bfb)]$ is minimized, we
have that 
$$
\left|\mathop{\E}_{\bfb
\getr \mcal{D}^{n}}[\util_i(\bfb)]
- \mathop{\E}_{\bfb_{-j}
\getr \mcal{D}^{n-1}}[\util_i(\bfb_{-j})] \right| 
\leq 
2 \cdot {\sf tick}(\inDom) \cdot  
k/(n-1)
$$
Now, using the same reasoning inductively on 
$\mathop{\E}_{\bfb_{-j}
\getr \mcal{D}^{n-1}}[\util_i(\bfb_{-j})]$, we get 
that 
\begin{align*}
\left|
\mathop{\E}_{\bfb
\getr \mcal{D}^{n}}[\util_i(\bfb)]
- \mathop{\E}_{v_i
\getr \mcal{D}}[\util_i(v_i)] \right| 
& \leq  
2 \cdot {\sf tick}(\inDom) \cdot  
\left(\frac{k}{n-1} + \frac{k}{n-1} + \ldots +  \frac{k}{1} \right)
\\
& \leq 
2 \cdot {\sf tick}(\inDom) \cdot k 
\cdot (\ln n + O(1))
\end{align*}
\end{proof}

\begin{theorem}[2pbIC + bIC $\Longrightarrow$ impossible]
\label{thm:2pbic-impossible}
Let $\Pi_\lambda$ be an 
auction for selling $k$ items over the value
domain 
$\inDomSecu \subset [0, 1]$.  
Suppose there exists a negligible function $\negl$ such that
${\sf tick}(\inDomSecu) \leq \negl(\lambda)$ for 
all sufficiently large $\lambda$. 
We have that 
for any distribution $\mcal{D}_\lambda$ over $\inDomSecu$,
if $\Pi_\lambda$ is non-trivial w.r.t. $\mcal{D}_\lambda$,  
then $\Pi_\lambda$
cannot simultaneously satisfy strong computational Bayesian bIC 
and strong computational 
Bayesian $2$-pbIC simultaneously w.r.t.
$\mcal{D}_\lambda$ over $\inDomSecu$. 
\label{thm:2pbic}
\end{theorem}
In the above, the non-trivial condition
is defined in the same way as 
in \Cref{sec:puboutcome}. 

\begin{proof}(of \Cref{thm:2pbic})
By weak symmetry, 
$\mathop{\E}_{v_i \getr \mcal{D}_\lambda}[\util_i(v_i)]$ 
is the same 
no matter what the identity of the lone buyer takes. 
Due to \Cref{lem:2pbic-key}, 
the fact that the auction is non-trivial, 
and the fact that ${\sf tick}(\inDomSecu)$
is negligibly small, 
it must be that infinitely many $\lambda$'s, 
$\mathop{\E}_{v_i \getr \mcal{D}_\lambda}[\util_i(v_i)] \geq 1/\poly(\lambda)$. 
Now, consider a crowded world with 
$K$ buyers where $K$ is sufficiently large. 
Due to \Cref{lem:2pbic-key}, 
everyone's expected utility is at least $1/\poly(\lambda) - \negl(\lambda) 
\geq 1/\poly_1(\lambda)$. 
This means that every one has more than $1/\poly_2(\lambda)$
probability of getting an item.  
By setting $K$ to be greater than $k \cdot \poly_2(\lambda)$,
we get that the expected number of items allocated is strictly greater
than $k$, which gives the contradiction. 
\end{proof}


\begin{corollary}
\label{cor:2pbic-impossible}
Let $\Pi_\lambda$ be a non-trivial auction for selling $k$ items over the value
domain 
$\inDomSecu \subset [0, 1]$.  
Suppose there exists a negligible function $\negl$ such that
${\sf tick}(\inDomSecu) \leq \negl(\lambda)$ for all sufficiently 
large $\lambda$. 
Then, 
$\Pi_\lambda$
cannot simultaneously satisfy strong computational bIC 
and strong computational 
$2$-pbIC simultaneously (in the ex post setting).
\end{corollary}

\begin{proof}
Observe that in the proof of \Cref{thm:2pbic}, 
we in fact only need that all honest buyers' values are sampled 
{\it independently}, 
but not necessarily from the same distribution. 
The ex post case is directly implied 
when honest buyers' bids 
are sampled independently but not necessarily from the same distribution,
since when we fix a value vector of the honest buyers, it is
the same as sampling the value vector from a deterministic distribution. 
\end{proof}

\subsection{bIC + psIC + 1-Message $\Longrightarrow$ Small Revenue}
\label{sec:psic}

We say that an auction is 1-message, 
iff the only communication in the protocol is for each buyer to send
a single message to the platform (not 
counting the 
round where the platform informs each player of their private outcome). 
A 1-message auction need
not be a direct revelation mechanism
where the buyer reveals its true value, 
since the message can contain arbitrary information and randomness
that may depend on the buyer's true value.

\paragraph{Proof roadmap.}
We want to show that any 1-message auction
that is bIC and psIC must be revenue dominated by posted price. 
We now describe the 
informal proof roadmap. For simplicity, we assume
the information theoretic and ex post settings for describing
the intuition, but our  
formal proofs later generalize the argument to computational and Bayesian 
settings. 
First, relying on bIC and psIC, we prove that if the world
has only one buyer, then the auction is revenue dominated by posted price. 
Then, we argue that 
under the value vector $\bfb = (b_1, \ldots, b_n)$, 
every buyer $i$'s expected pay must not exceed
its expected pay if the world 
has only buyer $i$ and no other buyers. 
Suppose this is not true, that is, under $\bfb$, there is some buyer $i$
who in expectation pays more than when the world has only buyer $i$. 
Then, 
the strategic platform-seller coalition   
can adopt the following strategy if the world consists of only buyer $i$.
The coalition
can simulate $n-1$ fake buyers with values $\bfb_{-i}$,
and pretend to $i$ that the world is $\bfb$. 
This way, the coalition will get more revenue in expectation
than honest behavior, which violates psIC. 
We now present the formal proofs. 

\subsubsection{The Ex Post Setting}
We begin by proving a limit on revenue for the ex post setting.  

\begin{definition}[Revenue dominated by posted price]
Let $\Pi_\lambda$ be an auction parametrized by $\lambda$ 
over the value domain $\inDomSecu$. 
We say that $\Pi_\lambda$ 
is $\delta(\cdot)$-revenue dominated by posted
price 
iff 
 every $\lambda$, 
there exist some $r_\lambda$ such that 
for any value vector ${\bf b} \in \inDomSecu^*$, 
$\Pi_\lambda$'s 
expected total payment under ${\bf b}$ 
is at most $\delta(\lambda)$ more than 
``the posted price auction
with reserve $r_\lambda$'' under ${\bf b}$. 
For the special case where $\delta(\lambda) = 0$ for all $\lambda$,
 we also simply say
that the 
auction is (strictly) revenue dominated by posted price. 
\end{definition}

\begin{theorem}[bIC + psIC + 1-message $\Longrightarrow$ small revenue]
Consider 
any 1-message auction $\Pi_\lambda$ parametrized by $\lambda$,
over the value domain $\inDomSecu \subset [0, 1]$. 
Suppose $\Pi_\lambda$ 
satisfies computational bIC and computational psIC
simultaneously.
Then, 
there exists a negligible function $\negl(\cdot)$, such that 
for every $\lambda$, 
$\Pi_\lambda$
is $\negl(\lambda)$-revenue dominated by posted price. 

Further, if ``computational'' is replaced with ``information theoretic'' above,
then the auction  
must be strictly revenue dominated by posted price. 
\elaine{fix this after fixing the IT defn}
\label{thm:psic-imp}
\end{theorem}
\begin{proof}
Consider a world with only one buyer $\mcal{B}$. 
Henceforth let ${\sf msg}$ be the buyer's single message to the platform, 
and let ${\sf coin}_{\platform}$ be the platform's randomness. 
Let $\mu({\sf msg})$ be the buyer's expected payment 
conditioned on ${\sf msg}$ under honest execution. 
Let $\mu({\sf msg}, {\sf coin}_{\mcal{P}})$ denote the buyer's payment
conditioned on ${\sf msg}$ and the platform's coins being 
${\sf coin}_{\mcal{P}}$.
Let $\msg(b, *)$ denote the distribution of the buyer's message
when it uses the value $b$ as input.

\begin{fact}
Let $\{X_\lambda\}_\lambda$ 
a probability ensemble that takes values between $[0, 1]$. 
If for some polynomially bounded function $\poly(\cdot)$, 
$\Pr\left[|X_\lambda - \E[X_\lambda]| \geq 1/\poly(\lambda)\right] \geq 1/\poly(\lambda)$
for infinitely many $\lambda$'s, 
then 
there is a polynomially bounded function $\poly'(\cdot)$ 
such that 
for infinitely many $\lambda$'s, 
\begin{itemize}[leftmargin=6mm,itemsep=1pt]
\item 
$\Pr\left[X_\lambda - \E[X_\lambda] \geq 1/\poly'(\lambda)\right] \geq 1/\poly'(\lambda)$, and
\item
$\Pr\left[\E[X_\lambda] - X_\lambda \geq 1/\poly'(\lambda)\right] \geq 1/\poly'(\lambda)$.
\end{itemize}
\label{fct:deviate}
\end{fact}
\elaine{does this need a proof?}

\begin{claim}
In a world with only one buyer, 
for sufficiently large $\lambda$, for any 
true value of the buyer drawn from $\inDomSecu$, 
under honest execution, 
except with negligible probability over the choice of ${\sf msg}$,
it holds that 
except with negligible probability over the choice 
${\sf coin}_{\mcal{P}}$, 
$|\mu({\sf msg}, {\sf coin}_{\mcal{P}}) - \mu(\msg)| \leq 
\negl(\lambda)$. 
\label{clm:helperpsic}
\end{claim}
\begin{proof}
For the sake of contradiction, suppose 
that for infinitely many $\lambda$'s, 
there is some value from $v_\lambda \in \inDomSecu$ and a set denoted ${\sf bad}$ 
that contains at least 
$1/\poly(\lambda)$ fraction of ${\sf msg}$ drawn from $\msg(v_\lambda, *)$,  
such that at least $1/\poly(\lambda)$ fraction of 
${\sf coin}_{\mcal{P}}$
would cause the buyer's payment to have at least 
$1 /\poly(\lambda)$ 
difference from 
$\mu({\sf msg})$. 
Below, we omit writing ``for infinitely many $\lambda$'s'' when
the context is clear. 
By \Cref{fct:deviate},  
there is some $\poly_1(\cdot)$ such that  
for any $\msg \in {\sf bad}$, 
at least  $1/\poly_1(\lambda)$ 
fraction of  
${\sf coin}_{\mcal{P}}$
would cause the 
payment to be at least $1/\poly_1(\lambda)$ larger
than  
$\mu({\sf msg})$. 
We can consider the following platform-seller polynomial-time strategy: 
whenever a buyer with true value $v_\lambda$ sends $\msg$, 
the platform 
samples ${\sf coin}_{\mcal{P}}$  
for $\poly_1^2(\lambda)$ times, and chooses
the best 
${\sf coin}_{\mcal{P}}$ that maximizes its revenue. 
Conditioned on ${\sf msg} \in {\sf bad}$, 
except with negligible probability, 
this strategy
would result in at least $1/\poly_1(\lambda)$ 
expected payment 
over the honest strategy. 
Conditioned on ${\sf msg} \notin {\sf bad}$, the 
expected payment is at least as large as the honest case.  
Overall, the expected 
gain over the honest strategy is $1/\poly_2(\lambda)$. 
This violates computational psIC. 
\end{proof}



\begin{claim}
Suppose the world has only one buyer with an arbitrary value $b \in \inDomSecu$. 
Except over negligible probability
over the choice of ${\sf msg} \getr {\sf msg}(b, *)$, 
$|{\sf butil}(\msg) - {\sf butil}(b)| \leq \negl(\lambda)$
where ${\sf butil}(\msg)$ denotes the buyer's 
expected utility conditioned on $\msg$,
and ${\sf butil}(b)$ denotes the buyer's 
expected utility given value $b$. 
\label{clm:psic-msg}
\end{claim}
\begin{proof}
For the same of contradiction, 
suppose 
for infinitely many $\lambda$'s, 
there is some $b_\lambda \in \inDomSecu$, such that 
with at least $1/\poly(\lambda)$ probability over
the choice of $\msg$,  
$|{\sf butil}(\msg) - {\sf butil}(b_\lambda)| \geq 1/\poly(\lambda)$. 
Below, we omit 
writing ``for infinitely many $\lambda$'s'' whenever the context is clear. 
By \Cref{fct:deviate}, 
it must be that with at least $1/\poly_1(\lambda)$ probability
over $\msg$, 
${\sf butil}(\msg) - {\sf butil}(b_\lambda) \geq 1/\poly_1(\lambda)$.
Now a strategic buyer 
with value $b_\lambda$ can simply choose
a $\msg$ 
such that ${\sf butil}(\msg) - {\sf butil}(b_\lambda) \geq 1/\poly_1(\lambda)$, 
and such a message can be provided to the buyer as an advice string. 
This strategy is clearly polynomial time and lets the buyer gain 
$1/\poly_1(\lambda)$ amount over the honest strategy,
which violates computational bIC. 
\end{proof}

Fix some $\lambda$, and $b_\lambda \in \inDomSecu$.
The $\msg$s from $\msg(b_\lambda, *)$ can be classified into two types, 
where $\negl'(\cdot)$ denotes a negligible function 
whose value is sufficiently large 
w.r.t. the negligible functions in \Cref{clm:helperpsic}
for sufficiently large $\lambda$s. 
\begin{itemize}[leftmargin=6mm,itemsep=1pt]
\item 
{\it Type 1 $\msg$:} for at least $\negl'(\lambda)$ fraction of 
${\sf coin}_{\mcal{P}}$, the buyer does not get an item; 
\item 
{\it Type 2 $\msg$:} 
the buyer gets an item with 
at least $1-\negl'(\lambda)$
over the choice of ${\sf coin}_{\mcal{P}}$. 
\end{itemize}

\begin{lemma}
Consider a world with only one buyer. 
It must be that 
for sufficiently large $\lambda$, for any $b_\lambda \in \inDomSecu$, 
except with negligible probability 
over the choice
of $\msg \getr \msg(b_\lambda, *)$, 
if $\msg$ is of type 2, 
then the buyer's payment 
conditioned on $\msg$
is at most 
negligibly apart 
from some constant $\mu_\lambda$ that does not depend on $b$ or $\msg$.
\label{lem:psic-const-mu}
\end{lemma}
\begin{proof}

By \Cref{clm:psic-msg}, 
except with $\negl(\lambda)$ probability
over the choice of $\msg \getr \msg(b_\lambda, *)$, 
if $\msg$ is of type 2, then 
$|\mu(\msg) - \mu(b_\lambda)| \leq \negl(\lambda)$
where $\mu(\msg)$ denotes the buyer's expected
payment conditioned on $\msg$, and $\mu(b_\lambda)$
is a function that depends only on $b_\lambda$. 

Suppose that for infinitely many $\lambda$'s, there
exists some $b'_\lambda \in \inDomSecu$ 
such that type-2 messages happen with
$1/\poly(\lambda)$ 
probability under honest execution. 
We claim that 
for any $b_\lambda \neq b'_\lambda$, 
either 
$\msg \getr \msg(b_\lambda, *)$
gives a type-2 message with negligible probability, 
or 
$\mu(b_\lambda) - \mu(b'_\lambda) \leq \negl(\lambda)$. 
Suppose this is not true, that is, 
for infinitely many $\lambda$'s, 
there is some $b_\lambda \neq b'_\lambda$ such that 
$\msg \getr \msg(b_\lambda, *)$
gives a type-2 message with $1/\poly(\lambda)$ probability, 
and $\mu(b_\lambda) - \mu(b'_\lambda) \geq 1/\poly(\lambda)$. 
In this case, a buyer with value $b_\lambda$  
can adopt the following polynomial-time strategy. 
Choose a type-2 message $\msg^* \in \msg(b'_\lambda, *)$ 
which minimizes its payment --- $\msg^*$ can
be provided as an advice string to the buyer.  
By \Cref{clm:psic-msg}, if the buyer acts honestly, 
its expected utility is at most 
$b_\lambda - \mu(b_\lambda) + \negl(\lambda)$.
Now under the above strategy, its expected utility 
is at least  
$b_\lambda - \mu(b'_\lambda) - \negl(\lambda)$, which
is $1/\poly(\lambda) - 2\negl(\lambda)$  
larger than acting honestly, which violates
computational bIC. 

The lemma now follows 
by letting $\mu_\lambda := \mu(b'_\lambda)$. 
\ignore{
Summarizing the above, for any $b$,
either 
$\msg \getr \msg(b, *)$
gives a type-1 message with all but negligible probability
and the payment is negligibly small,
or 
except with negligible probability 
over 
the choice $\msg \getr \msg(b, *)$, the  
buyer's payment is 
\elaine{FILL}
}
\end{proof}

\paragraph{Finishing the proof of \Cref{thm:psic-imp}.}
By \Cref{clm:helperpsic}, 
and because $\negl'(\cdot)$ is sufficiently large, 
we have 
except with negligible probability
over the choice of $\msg \getr \msg(b_\lambda, *)$, 
if $\msg$ is of type 1, then 
the buyer's expected payment is at most $\negl(\lambda)$.  
Combined with \Cref{lem:psic-const-mu}, 
we conclude that 
when the world has a single buyer, 
if its true value is less than $\mu_\lambda-\negl_1(\lambda)$, 
then its expected payment is at most $\negl_2(\lambda)$; 
else 
its expected payment cannot exceed $\mu_\lambda + \negl_3(\lambda)$. 
Henceforth, define $\mu_{\rm posted}(v)$ to be 
an upper bound on the expected revenue
earned from a buyer with value $v$: 
\[
\mu_{\rm posted}(v) = 
\begin{cases}
\negl_2(\lambda) & \text{ if } v < \mu_\lambda-\negl_1(\lambda) \\
\mu_\lambda + \negl_3(\lambda) & \text{o.w.}\\ 
\end{cases}
\]

Now, suppose 
that for infinitely many $\lambda$'s, 
there is value vector ${\bf b} = (b_1, \ldots, b_n) \in \inDomSecu^n$ 
such that some buyer $i$'s expected payment is more than 
 $\mu_{\rm posted}(b_i) + 1 /\poly(\lambda)$.
We have the following
scenario and a polynomial-time strategy 
that violate computational psIC.
Suppose the world actually has only a single buyer 
with true value $b_i$.
The platform-seller  
pretends that the world is actually ${\bf b}$ where
the victim buyer is the $i$-th coordinate. 
In other words, the strategy will simulate
$n-1$ fake players with the true values $\bfb_{-i}$
and run the protocol with the single victim buyer. 
This strategy has expected gain at least $1/\poly(\lambda) - \negl(\lambda)$.
This violates 
computational psIC. 

For the information-theoretic case, we can force all the negligible functions
to be $0$. In this case, all the $1/\poly(\lambda)$ slacks would vanish to $0$, 
and the statement follows directly. 
\elaine{TODO: double check}
\end{proof}

\subsubsection{The Bayesian Setting}
We now extend the result to the Bayesian setting. 

\begin{corollary}
Consider 
any 1-message auction $\Pi_\lambda$ parametrized by $\lambda$,
over the finite value domain $\inDomSecu \subset [0, 1]$. 
Suppose $\Pi_\lambda$ 
satisfies computational Bayesian bIC and computational Bayesian psIC
w.r.t. a polynomial-time samplable 
\elaine{TODO: some of the other thms may need this condition too}
distribution $\mcal{D}_\lambda$ over $\inDomSecu$. 
Then, 
for any polynomial function $n(\cdot)$, 
there exists a negligible function $\negl(\cdot)$, 
such that for any $\lambda$, there is a constant $\mu_\lambda$, 
for any $n' \leq n(\lambda)$, 
except with $\negl(\lambda)$ probability 
over the choice of $\bfv \getr \mcal{D}_\lambda^{n'}$, 
the auction's expected revenue 
on $\bfv$
is at most $\negl(\lambda)$ 
more than the revenue of a posted price auction with reserve
$\mu_\lambda$ on $\bfv$. 

Further, if computational is replaced with ``information theoretic'' above,
then, for any $\lambda$, 
there exists some constant $\mu_\lambda$ such that 
for any $n$, 
with probability $1$ over the choice of 
$\bfv \getr \mcal{D}_\lambda^n$, the auction's expected revenue
on $\bfv$ is no more than the revenue
of a posted 
price auction with reserve $\mu_\lambda$. 
\end{corollary}
\begin{proof}
With Bayesian psIC, the following variant of \Cref{clm:helperpsic}
holds.  

\begin{claim}
In a world with only one buyer,
for sufficiently large $\lambda$, for 
all but negligible fraction of 
buyer true value drawn from $\mcal{D}_\lambda$,
under honest execution,
except with negligible probability over the choice of ${\sf msg}$,
it holds that
except with negligible probability over the choice
${\sf coin}_{\mcal{P}}$,
$|\mu({\sf msg}, {\sf coin}_{\mcal{P}}) - \mu(\msg)| \leq 
\negl(\lambda)$.
\label{clm:helperpsic-bayesian}
\end{claim}
\begin{proof}
We first 
prove it for the computational setting. 
The proof is almost the same as that of \Cref{clm:helperpsic}
except that ``for any true value'' is replaced
with ``for all but negligible fraction of the values
sampled from $\mcal{D}_\lambda$'', and the new proof relies 
on Bayesian psIC rather than ex post psIC. 
\end{proof}

\Cref{clm:psic-msg} and \Cref{lem:psic-const-mu}  still hold because  
when the world consists of only one buyer,
Bayesian bIC is the same as ex post bIC.

Now, the remainder of the proof works as follows.
Similar to the proof in the ex post setting,  
combining \Cref{clm:helperpsic-bayesian} 
and \Cref{lem:psic-const-mu}, we have that 
when the world consists of a single buyer,  
except with negligible probability over the choice of 
its true value $v \getr \mcal{D}_\lambda$, 
the expected revenue is upper bounded
by $\mu_{\rm posted}(v)$. 

Suppose that for some polynomial $n(\cdot)$, 
for infinitely many $\lambda$'s, 
there exist $1/\poly(\lambda)$ and some $n' \leq n(\lambda)$
such that 
 over $1/\poly(\lambda)$ 
probability over the choice of 
$\bfv = (v_1, \ldots, v_{n'}) \getr \mcal{D}_\lambda^{n'}$, 
the first buyer's expected payment
exceeds $\mu_{\rm posted}(v_1) + 1/\poly(\lambda)$. 
Then, a strategic platform-seller coalition can take the following 
polynomial-time strategy. 
Upon receiving the message $\msg_1$ from the first buyer, 
it samples $n'-1$ buyers whose true values 
drawn at random from 
$\mcal{D}_\lambda^{n'-1}$. It then samples 
all the random coins of the $n'-1$ fake buyers as well as $r_\platform$
and computes the payment of the first buyer.
It repeats the 
above sampling for a sufficiently polynomially many times and chooses
a scenario that maximizes the first buyer's payment
(including the honest case without the fake buyers), and pretends
that the world is that particular 
scenario to the first buyer. 
This strategy allows the strategic platform-seller coalition 
to gain at least $1/\poly'(\lambda)$ amount in expectation.

Finally, for the information theoretic setting, 
the negligible functions and $1/\poly(\lambda)$ terms all vanish to $0$
in the above argument, so the claimed conclusion holds.  
\end{proof}

\subsection{Information Theoretic bIC + psIC $\Longrightarrow$ Small Revenue}

\elaine{NOTE: this proof is implicitly using the IT notion of safe,
all attacks in the proof are safe deviations}

\elaine{TO FIX: the proof still holds even when there is blockchain}

In this section, we prove the following theorem. 

\begin{theorem}[Information theoretic bIC + psIC $\Longrightarrow$ small revenue]
Let $\Pi$ be a possibly multi-round auction that satisfies
information theoretic Bayesian bIC and ex post psIC.  
Then, $\Pi$ is revenue-dominated by posted price. 
\label{thm:psic-multiround}
\end{theorem}

In comparison with the proof in \Cref{sec:psic},
the main challenge here is how to still prove
that with only one buyer, the auction is revenue-dominated
by posted price,  
now that the auction can be multi-round. 
We present our proof below. 

\begin{proof}(of \Cref{thm:psic-multiround}.)
We first analyze a world with only one buyer. 

\paragraph{Single buyer setting.}
We will consider a world with only one buyer with true value $v$. 
Since Bayesian bIC is equal to ex post bIC in a single-buyer setting,
we will simply say bIC for short in this part of the proof.  
We use the following notations:
\begin{itemize}[leftmargin=6mm,itemsep=1pt]
\item 
${\sf bmsg}_{-1}$: the buyer's last message in the protocol, 
\item 
${\sf tr}_{-1}$:  
all messages exchanged between the buyer and the platform-seller
coalition till right before the buyer 
is about to send its last message ${\sf bmsg}_{-1}$, 
\item 
${\sf pcoin}$: additional private coin tosses
made by the platform-seller coalition 
that did not contribute to ${\sf tr}_{-1}$. 
\end{itemize}


\begin{fact}
Let $x$ and $y$  
be random variables sampled from 
an arbitrary joint distribution $\mcal{D}$ 
over $\mathbb{D}_x \times \mathbb{D}_y$. 
Let $f(x, y): \mathbb{D}_x \times 
\mathbb{D}_y \rightarrow \mathbb{R}$ 
be some real-valued function, 
and let $f(x)$ be the expectation of $f(x, y)$ conditioned on $x$.  
If with non-zero probability, $f(x, y) \neq f(x)$, then 
the expectation of $f(x, y)$ in the first process is greater
(or smaller) 
than in the second process below:
\begin{enumerate}[leftmargin=6mm,itemsep=1pt]
\item 
Sample $x$ from the marginal distribution over $x$ induced
by $\mcal{D}$, 
choose $y$ to maximize (or minimize) $f(x, y)$;
\item 
Sample $x$ and $y$ at random from 
$\mcal{D}$. 
\end{enumerate} 
\label{fct:devfrommean}
\end{fact}
\begin{proof}
It suffices to show that with non-zero probability, 
$f(x) < \max_y f(x, y)$. 
Suppose this is not true, that is, 
with probability $1$, 
$f(x) = \max_y f(x, y)$. 
Then, with probability $1$, $f(x, y) \leq f(x)$. 
Because $f(x) = \int_y f(x, y) \cdot {\sf pdf}(y)$
where ${\sf pdf}(y)$ denotes the probability density function over $y$, 
it must be that  
with probability $1$, $f(x) = f(x, y)$ which contradicts our assumption. 
\end{proof}

\elaine{note: trace includes only the msgs exchanged, not including local
random coins}
\begin{claim}
Suppose that the world has only one buyer with an arbitrary value. 
Let $r({\sf tr}_{-1}, {\sf bmsg}_{-1})$
denote the auction's expected revenue 
conditioned on $({\sf tr}_{-1}, {\sf bmsg}_{-1})$. 
Then, 
with probability $1$, 
the auction's revenue is 
$r({\sf tr}_{-1}, {\sf bmsg}_{-1})$, i.e., 
the revenue does not depend on ${\sf pcoin}$. 
\label{clm:lastbuyermsg}
\end{claim}
\begin{proof}
Suppose that the claim is not true, 
then a strategic platform-seller coalition can adopt the following strategy: 
upon observing $({\sf tr}_{-1}, {\sf bmsg}_{-1})$, 
choose a ${\sf pcoin}$ that maximizes its revenue. 
Due to \Cref{fct:devfrommean}, the platform-seller coalition
can gain a positive amount  
with this strategy which violates ex post psIC.  
\ignore{
Suppose for some $v$, there is  
a set $\Gamma = \{({\sf tr}_{-1}, {\sf bmsg}_{-1})\}$ that occurs
with non-zero probability, such that given  
$({\sf tr}_{-1}, {\sf bmsg}_{-1}) \in \Gamma$, 
the auction's revenue additionally depends on 
${\sf pcoin}$, 
then 
a strategic platform-seller coalition
can adopt the following strategy:  
if it ever observes 
some $({\sf tr}_{-1}, {\sf bmsg}_{-1}) \in \Gamma$, it will 
choose ${\sf pcoin}$ to maximize   
its revenue; otherwise, behave honestly. 
This strategy will increase the coalition's
expected utility thus violating 
ex post psIC. 
}
\end{proof}

\Cref{clm:lastbuyermsg} 
implies that 
with probability $1$, 
one of the following must be true: 
\begin{itemize}[leftmargin=7mm,itemsep=1pt]
\item 
{\it Case 1:}
$r({\sf tr}_{-1}, {\sf bmsg}_{-1}) > 0$, the buyer gets
an item and pays $r({\sf tr}_{-1}, {\sf bmsg}_{-1})$; 
\item 
{\it Case 2:}
$r({\sf tr}_{-1}, {\sf bmsg}_{-1}) = 0$, 
and the buyer pays $0$ if it gets an item. 
\end{itemize}

\begin{claim}
Fix an arbitrary $v$ of the lone buyer. 
With probability $1$, 
if $({\sf tr}_{-1}, {\sf bmsg}_{-1})$ belongs to Case 1, 
it must be that 
$r({\sf tr}_{-1}, {\sf bmsg}_{-1}) = r({\sf tr}_{-1})$
where 
$r({\sf tr}_{-1}, {\sf bmsg}_{-1})$ is the expected payment
of the buyer conditioned on $({\sf tr}_{-1}, {\sf bmsg}_{-1})$ 
and $r({\sf tr}_{-1})$ is the expected
payment of the buyer conditioned on ${\sf tr}_{-1}$ 
and sampling some ${\sf bmsg}_{-1}$
such that 
we land in Case 1. 
\label{clm:rindlastmsg}
\end{claim}
\begin{proof}
We may assume that Case 1 happens with non-zero probability, 
since otherwise
the claim trivially holds. 

Suppose the claim is not true, 
then a strategic buyer can 
adopt the following strategy: 
it follows the protocol honestly, until it is about to
send the last message ${\sf bmsg}_{-1}$.  
At this moment, 
the buyer 
first samples ${\sf bmsg}_{-1}$ honestly and 
checks if 
we are in Case 1. 
If so, it replaces the honestly sampled
${\sf bmsg}_{-1}$ with a ${\sf bmsg}^*_{-1}$
such that 
we are in Case 1 
and $r({\sf tr}_{-1}, {\sf bmsg}^*_{-1})$
is minimized. 
Due to \Cref{fct:devfrommean}, this strategy brings positive
gain to the buyer which violates bIC. 
\elaine{TODO: argue this}
\ignore{
1) either 
with non-zero probability over the choice of $({\sf tr}_{-1}, {\sf bmsg}_{-1})$, 
$r({\sf tr}_{-1}, {\sf bmsg}_{-1}) > r({\sf tr}_{-1})$, 
or 2)  
with non-zero probability over the choice of $({\sf tr}_{-1}, {\sf bmsg}_{-1})$, 
$r({\sf tr}_{-1}, {\sf bmsg}_{-1}) < r({\sf tr}_{-1})$. 
}

\end{proof}

\ignore{
\begin{proof}
Suppose that for some $v$, 
there is some set $\Gamma = \{{\sf tr}_{-1}\}$
that happens with non-negligible probability, 
and for any ${\sf tr}_{-1} \in \Gamma$, 
there are ${\sf bmsg}_{-1}$ and ${\sf bmsg}'_{-1}$
such that 
$({\sf tr}_{-1}, {\sf bmsg}_{-1})$ 
and $({\sf tr}_{-1}, {\sf bmsg}'_{-1})$ 
both land in Case 1, and moreover, 
$r({\sf tr}_{-1}, {\sf bmsg}_{-1}) 
\neq r({\sf tr}_{-1}, {\sf bmsg}'_{-1})$. 
A strategic buyer can adopt the following strategy: 
if it ever observes 
any ${\sf tr}_{-1} \in \Gamma$
such that given ${\sf tr}_{-1}$, the conditional expected
revenue is positive, then 
it strategically chooses a ${\sf bmsg}_{-1}$
that minimizes its payment.  
This allows the buyer to gain in expected utility, thus violating
Bayesian bIC. 
\end{proof}
}


\ignore{
\begin{claim}
Conditioned on ${\sf bview}({\sf tr}_{-1}, {\sf bmsg}_{-1})$, 
\elaine{TODO: define}
the buyer knows which case it is, and if it is in Case 1, 
its payment $r$ is fully determined by 
${\sf bview}({\sf tr}_{-1}, {\sf bmsg}_{-1})$. 
\end{claim}
\begin{proof}
\elaine{relies on some notion of safe}
\end{proof}
}

Let $\mu({\sf tr}_{-1}) := r({\sf tr}_{-1}) 
\cdot \Pr[\text{Case 1} | {\sf tr}_{-1}]$ denote
the expected revenue 
conditioned on ${\sf tr}_{-1}$.

\begin{claim}
Fix an arbitrary $v$. Then, with probability $1$ over the choice
of ${\sf tr}_{-1}$, the following holds: 
if $\mu({\sf tr}_{-1}) > 0$, 
then, conditioned on ${\sf tr}_{-1}$, the buyer's expected
utility is 
$v - r({\sf tr}_{-1})$ where $r({\sf tr}_{-1})$ is defined as
in \Cref{clm:rindlastmsg}. 
\label{clm:butilontr}
 \end{claim}
\begin{proof}
\Cref{clm:rindlastmsg} shows that with probability $1$, 
if 
we are in Case 1, 
the buyer's utility
must be $v - r({\sf tr}_{-1})$. 
It suffices to show that 
with probability $1$ over the choice of 
$({\sf tr}_{-1}, {\sf bmsg}_{-1})$, 
if 
$\mu({\sf tr}_{-1}) > 0$ 
and  
we are in Case 2, 
then the buyer's expected utility
conditioned on $({\sf tr}_{-1}, {\sf bmsg}_{-1})$
must still be $v - r({\sf tr}_{-1})$. 
We may assume that  
with non-zero probability, 
$\mu({\sf tr}_{-1}) > 0$ 
and  we are in Case 2, 
since otherwise the claim trivially holds. 
Suppose the above claim is not true, i.e., 
there is some non-zero probability  
such that 
$\mu({\sf tr}_{-1}) > 0$  and 
we are in Case 2, 
but the buyer's conditional expected utility
is not $v - r({\sf tr}_{-1})$. 
Then, a strategic buyer can adopt the following strategy: 
follow the protocol honestly until it is about to send its
last message ${\sf bmsg}_{-1}$, and pick
a ${\sf bmsg}^*_{-1}$ that maximizes its expected utility. 
Similar to the proof of \Cref{fct:devfrommean}, 
it is easy to show that this strategy can increase
the buyer's expected utility which violates bIC. 
\elaine{may need more explanation}
\ignore{
Suppose the 
claim is not true, i.e., there exists a non-zero 
probability set $\Gamma = \{{\sf tr}_{-1}\}$
such that whenever ${\sf tr}_{-1} \in \Gamma$, 

Due to bIC, 
with probability $1$ over the choice of 
${\sf tr}_{-1}$ and ${\sf bmsg}_{-1}$, 
it must be that 
${\sf butil}({\sf tr}_{-1}, {\sf bmsg}_{-1})
= \max_{{\sf bmsg}^*_{-1}} {\sf butil}({\sf tr}_{-1}, {\sf bmsg}^*_{-1})$ where 
${\sf butil}({\sf tr}_{-1}, {\sf bmsg}_{-1})$ denotes
the buyer's expected utility conditioned on $({\sf tr}_{-1}, {\sf bmsg}_{-1})$. 
If the conditional expected revenue given 
${\sf tr}_{-1}$ is positive, 
there must be a choice of ${\sf bmsg}_{-1}$
such that $({\sf tr}_{-1}, {\sf bmsg}_{-1})$ 
belongs to Case 1. 
In this case, the buyer's utility is $v - r({\sf tr}_{-1})$.
Therefore, 
with probability $1$ over the choice of 
${\sf tr}_{-1}$, 
conditioned on ${\sf tr}_{-1}$, the buyer's expected
utility is 
$v - r({\sf tr}_{-1})$.
}
\end{proof}





\begin{claim}
Fix an arbitrary $v$. 
With probability $1$ over the choice of ${\sf tr}_{-1}$, 
the conditional expected revenue $\mu({\sf tr}_{-1})$
depends only on the buyer's coin tosses 
that contributed to 
${\sf tr}_{-1}$.
\end{claim}
\begin{proof}
Consider that the lone 
buyer and the platform-seller coalition engage in a protocol 
to generate ${\sf tr}_{-1}$. 
Suppose right before the platform-seller
coalition sends the first message
${\sf pmsg}_1$
to the buyer, the transcript so far is 
denoted ${\sf tr}_1$. 
Let $\mu({\sf tr}_1)$ 
be the conditional expectation 
of $\mu$ given ${\sf tr}_1$. 
Let 
$\mu({\sf tr}_1, {\sf pmsg}_1)$ 
be the conditional expectation 
of $\mu$ given $({\sf tr}_1, {\sf pmsg}_1)$. 
We claim that 
with probability $1$ over the choice of 
$({\sf tr}_1, {\sf pmsg}_1)$, 
it must be that
$\mu({\sf tr}_1) = 
\mu({\sf tr}_1, {\sf pmsg}_1)$. 
Suppose this is not true.
Then, a strategic platform-seller coalition can adopt the following strategy:
when it is about to send its first message, sample
a ${\sf pmsg}_1$
that maximizes 
$\mu({\sf tr}_1, {\sf pmsg}_1)$. Due to \Cref{fct:devfrommean},
it can gain positively based on this strategy which violates ex post psIC. 

Repeating this argument for every message
sent from the platform-seller coalition to the buyer, we conclude
the following. 
Suppose we sample ${\sf tr}_{-1}$ at random, and let   
${\sf bcoin}$ be the buyer's coins that  
contributed to the sampling of ${\sf tr}_{-1}$,
we can compute $\mu({\sf tr}_{-1})$
from ${\sf bcoin}$ as follows:  
sample the platform-seller coalition's coin tosses at random, 
and combined with ${\sf bcoin}$, we get some ${\sf tr}'_{-1}$. 
With probability $1$ over the choice of ${\sf tr}_{-1}$ and 
${\sf tr}'_{-1}$, it holds that $\mu({\sf tr}'_{-1})
= \mu({\sf tr}_{-1})$. 
\elaine{double check}
\end{proof}

Let ${\sf bcoin}$
be the buyer's coin tosses in generating ${\sf tr}_{-1}$, and suppose  
${\sf bcoin}$ leads to $\mu > 0$. In this case, 
let $\bar{r}({\sf bcoin})$ be the expected payment of the buyer 
conditioned on ${\sf bcoin}$ and landing in Case 1. 
Let $\bar{r}$ denote the 
expectation of 
$\bar{r}({\sf bcoin})$ 
over the choice of ${\sf bcoin}$
subject to $\mu > 0$. 


\begin{claim}
Fix an arbitrary $v$. 
With probability $1$ over the choice of ${\sf bcoin}$, 
if ${\sf bcoin}$ leads to a positive $\mu$, then 
$\bar{r}({\sf bcoin}) = \bar{r}$. 
\end{claim}
\begin{proof}
Henceforth, we may assume that 
given $v$, with non-zero probability over the choice of 
${\sf bcoin}$, it holds that $\mu > 0$. 
Otherwise, the claim holds trivially. 

Suppose the claim does not hold, 
then, a strategic buyer can adopt the following strategy: 
if it happens to sample ${\sf bcoin}$ such that
$\mu  > 0$, then replace ${\sf bcoin}$ with a choice 
${\sf bcoin}^*$ that minimizes  
$\bar{r}({\sf bcoin}^*)$.
Due to 
\Cref{clm:butilontr} and 
\Cref{fct:devfrommean}
this allows the buyer to gain which violates bIC. 
\elaine{state bayesian ex post same upfront}
\end{proof}

Given the above claim, 
given any $v$ such that the expected revenue is positive, 
we can define 
the buyer's pay in Case 1 as
$\bar{r} = \bar{r}(v)$, which is a function only of $v$. 
We next argue that in fact, $\bar{r}$ cannot even depend on $v$. 

\begin{claim}
Given $v_1$ and $v_2$ such that the expected revenue is non-zero
under $v_1$ and $v_2$, 
$\bar{r}(v_1) = \bar{r}(v_2)$. 
\end{claim}
\begin{proof}
Suppose $\bar{r}(v_1) < \bar{r}(v_2)$. 
Then, a buyer with true value $v_2$ can adopt the following strategy.
It first samples ${\sf bcoin}_2$ 
honestly taking on the true value $v_2$. If ${\sf bcoin}_2$
leads to $\mu > 0$, 
it will instead pretend that its true value is $v_1$, 
and sample some corresponding ${\sf bcoin}_1$
which leads to $\mu > 0$, and its pay in Case 1 
would
be $\bar{r}(v_1)$.
It will then participate in the protocol using $v_1$ and ${\sf bcoin}_1$
as input, until the buyer  
is about to send its last message ${\sf bmsg}_{-1}$. 
At this moment, the buyer samples
${\sf bmsg}_{-1}$
such that we land in Case 1.  

Due to \Cref{clm:butilontr}, 
conditioned on ${\sf bcoin}_2$
leading to $\mu > 0$, if the buyer continued
to behave honestly, its expected utility would be 
$v_2 - \bar{r}(v_2)$. 
If it adopted the above strategy, 
its expected utility 
would be 
$v_2 - \bar{r}(v_1)$ which is strictly better
than the honest case. 
This violates bIC. 
\end{proof}

Given the above claim, we can let 
$\bar{r} = \bar{r}(v_1) = \bar{r}(v_2)$
be a universal constant. 
By individual rationality, if the buyer's true value $v < \bar{r}$, then 
$\mu > 0$ 
happens with probability $0$. 
In other words, if $v < \bar{r}$, the 
auction's revenue must be $0$. 
Otherwise, its expected revenue is at most $\bar{r}$. 
Therefore, when the world has a single buyer with value $v$, 
the auction's
 revenue is dominated by 
\[
\mu_{\rm posted}(v) = \begin{cases} 
0 & \text{ if $v < \bar{r}$ } \\
\bar{r} & \text{ o.w.}  
\end{cases}
\]

\paragraph{Completing the proof of \Cref{thm:psic-multiround}.}
Suppose that there exists some value vector $\bfv = (v_1, \ldots, v_n)$ 
and some $i \in [n]$ such that $i$'s expected payment is more
than $\mu_{\rm posted}(v_i)$. 
Then, when the world has a single buyer $v_i$, 
the platform-seller coalition can  
simulate $n-1$ fake buyers with 
values $\bfv_{-i}$. 
This strategy allows the coalition to gain in expectation. 
\end{proof}

\section{Utility-Dominated Emulation}
\label{sec:emulation}

We will eventually 
show that by considering computationally bounded agents,  
we can obtain a constant-round 
auction that satisfies 
computational bIC, pIC, 1-pbIC, Bayesian sIC and Bayesian psIC simultaneously. 
\elaine{double check}
To this end, we introduce a new design paradigm. 
We first define an ideal world in which 
the seller, the platform, and the buyers all interact with a trusted  
ideal functionality $\Fauction$ 
which helps to implement the auction.  
If in this ideal world, we can design
a direct-revelation auction that satisfies 
the desired incentive compatibility properties
in an {\it information theoretic} sense, 
then we can compile it to a real-world protocol that 
replaces the ideal functionality with 
cryptography, and the real-world protocol
satisfies   
the same set of incentive compatibility properties 
{\it against computationally bounded players}. 

\elaine{TODO: even with iO it is not known how to get subquadratic compute}

As mentioned in \Cref{sec:intro}, the most 
straightforward approach is to compile the ideal auction
using generic multi-party computation (MPC). 
However, all known generic approaches (including
threshold FHE) incur at least
quadratic overhead, partly because in an auction, 
every buyer obtains a different outcome. 
Instead, we 
formulate a weaker notion of simulation called
{\it utility-dominated emulation}. 
With this new notion, we show an efficient compiler  
where every buyer and seller has only $\widetilde{O}_\lambda(1)$ 
cost and the platform has $\widetilde{O}_\lambda(n)$  
cost where $n$ is the number of buyers.

\elaine{TODO: do coin toss in real-world, remove platform choose rP
in ideal world}

\subsection{Model for the Ideal Auction}
\label{sec:idealmodel}
Because we will later compile the ideal-world auction
into a real-world cryptographic protocol, 
we will also parametrize
the ideal auction 
with $\lambda$, even though it provides information theoretic guarantees.
In natural designs of the ideal auction, 
the rules of the auction are uniform algorithms that
work for all $\lambda$. 

We will consider a direct revelation mechanism
in the ideal world. Since 
the ideal world has a trusted functionality $\Fauction$, it
is not hard to extend the direct revelation principle
to cover all of the 
incentive compatibility properties we care about.  
Thus, considering a direct revelation mechanism
in the ideal world is without loss of generality. 
Therefore, 
We can specify 
an ideal-world auction for selling $k$ identical items 
by specifying the following possibly randomized algorithms.
Henceforth, we assume that the value domain is 
$\inDomSecu$ given $\lambda$. 
\elaine{TO FIX: we need to make it discrete first before compiling
to real world, so check to see if we want
to do the discrete here or defer to later} 
\begin{itemize}[leftmargin=5mm,itemsep=1pt]
\item 
\textbf{Allocation rule} ${\bf x}(1^\lambda, \bfb)$: 
takes as input $1^\lambda$ and 
a bid vector $\bfb = (b_1,\dots,b_n) \in \inDomSecu^n$,
and outputs a vector $(x_1,\dots,x_n)$ where $x_i \in \{0,1\}$ specifies whether the buyer $i$ is allocated an item.
\item 
\textbf{Payment rule} ${\bf p}(1^\lambda, \bfb)$:
takes as input $1^\lambda$ and 
a bid vector $\bfb = (b_1,\dots,b_n) \in \inDomSecu^n$,
and outputs $(p_1,\dots,p_n)$ where $p_i \geq 0$ 
specifies buyer $i$'s payment.
\item 
\textbf{Seller revenue rule} $\mu_{\mcal{S}}(1^\lambda,\bfb)$: 
takes as input $1^\lambda$ 
and a bid vector $\bfb = (b_1,\dots,b_n) \in \inDomSecu^n$,
and outputs $\mu_{\mcal{S}}\geq 0$ specifying 
the revenue that the seller gets.
\end{itemize}
For compiling the ideal auction to a real-world protocol,
we additionally require that the above rules can be computed in polynomial time.

\paragraph{Ideal auction.}
Given these rules $({\bf x}, {\bf p}, \mu_{\mcal{S}})$, we
define the following ideal auction protocol
denoted $\Pi^\Ideal_\lambda[{\bf x}, {\bf p}, \mu_{\mcal{S}}]$. 
Henceforth let $\Fauction$ denote a trusted ideal functionality.
\begin{enumerate}[leftmargin=7mm,itemsep=1pt]
\item 
 \label{step:idealbid}
Every buyer sends zero to multiple bids to $\Fauction$. 
\item 
$\Fauction$ sends the number of bids received to the platform.
\item 
 \label{step:idealstrategicbid}
The platform and the seller may each send zero to multiple bids  
to $\Fauction$. 
\item 
$\Fauction$ computes the outcome for the seller, the platform, and all buyer identities
based on the rules $({\bf x}, {\bf p}, \mu_{\mcal{S}})$, 
using all bids it has received. If the auction's rules
are randomized, $\Fauction$ tosses the random coins.
$\Fauction$ informs
the platform and seller of their respective outcomes. 
Further, for each bid, $\Fauction$ informs its outcome to the player 
who sent the bid earlier. 
\item 
\label{step:idealok}
The platform sends ${\sf ok}$ or $\bot$ to $\Fauction$.
The auction
is considered successful if the platform sends ${\sf ok}$; else
we treat the auction as a failure and every one's utility is defined to be 0.
\end{enumerate}

\paragraph{Honest behavior.}
The honest behavior is defined as follows.
Every honest buyer sends its true value in 
Step~\ref{step:idealbid}. 
An honest platform and an honest seller do not send any bid in  
Step~\ref{step:idealstrategicbid}. 
An honest platform always sends ${\sf ok}$ in 
Step~\ref{step:idealok}.

\paragraph{Strategic behavior.}
All possible strategic behavior is already made evident
through the description of the ideal auction itself. 


\subsection{Definition: Utility-Dominated Emulation}
\label{sec:defnsim}

Let $\Pi^\Ideal_\lambda$ 
be an ideal auction 
over some a family of finite 
value domains $\inDomSecu \subset \mathbb{R}$, and suppose
$\Fauction$ satisfies information theoretic
incentive compatibility guarantees.  
We want to 
get rid of the ideal functionality 
$\Fauction$
in the ideal auction and instantiate it with cryptography,
resulting in a real-world cryptographic protocol 
denoted $\Pi^\Real_\lambda$. 
We want to argue that the information theoretic incentive compatibility guarantees
in the ideal world 
translate to computational counterparts of the same properties 
in the real world. 
\ignore{
The most straightforward way to do this is to argue
that the real-world
auction $\Pi^\Real_\lambda$
is a cryptographically secure emulation~\cite{Canetti2000,uc} 
of the ideal auction
$\Pi^\Ideal_\lambda$. 
However, this would require a full-fledged MPC protocol
with simulation security, and as mentioned earlier,
without indistinguishability obfuscation, we currently do not know
any such generic MPC protocol with subquadratic communication
in the pairwise channel model.
}
Since we only need to argue game-theoretic properties,
it turns out that we do not require a notion 
as strong as full simulation~\cite{Canetti2000,uc}. 
In this section, define a weaker notion of simulation
called utility-dominated emulation which suffices
for proving the game theoretic
properties we care about. 
With this new relaxed notion, 
we can construct cryptographic protocols  
that are asymptotically faster than known generic MPC
from standard assumptions. 

\ignore{
In the exposition below, the real-world auction
is parametrized with a security parameter $1^\lambda$
since it relies on cryptography, whereas the ideal-world
auction does not take a security parameter because its incentive compatibility
guarantees hold in an information theoretic sense.  
}

\begin{definition}[Utility-dominated emulation]
Let $\Pi^\Real_\lambda$ and $\Pi^\Ideal_\lambda$
be auctions 
over the finite value domain $\inDomSecu \subset \mathbb{R}$.  
We say that 
$\Pi^\Real_\lambda$
is a utility-dominated emulation of 
$\Pi^\Ideal_\lambda$ w.r.t. the strategic individual
or coalition $\mcal{C}$, iff 
the following properties hold:  
\begin{enumerate}[leftmargin=6mm]
\item 
{\bf Utility equivalence under honest execution}:
for any $\lambda \in \N$,
for any ${\bf v} \in \inDomSecu^*$,
and for any player $p$ (buyer, seller or platform), 
${\sf UReal}_p(1^\lambda, {\bf v})$ and ${\sf UIdeal}_p(1^\lambda, {\bf v})$
are identically distributed,
where 
${\sf UReal}_p(1^\lambda, {\bf v})$ 
is the random variable representing 
the utility of player $p$ under an honest execution of 
$\Pi^\Real_\lambda$
parametrized with security parameter $1^\lambda$, and with buyers
taking on the value vector ${\bf v}$, and  
${\sf UIdeal}_p(1^\lambda, {\bf v})$ 
is similarly defined but for $\Pi^\Ideal_\lambda$.
\item 
{\bf Utility dominance for strategic players}:
for any polynomial $n(\cdot)$, 
for any real-world \ppt strategy $S$ adopted by $\mcal{C}$, 
there exist 
a negligible function $\negl(\cdot)$, 
and 
an ideal-world strategy 
$S'(\cdot)$  
which may additionally depend on the coalition's true value
vector, 
such that for any $\lambda$, 
any $n_H \leq n(\lambda)$,  
any true value vector $\bfv_{-\mcal{C}} \in \inDomSecu^{n_H}$, 
any true value vector $\bfv_{\mcal{C}} \in \inDomSecu^{n_C}$ 
where $n_C$ denotes the number of buyers in $\mcal{C}$,  
it holds that 
\ignore{for any strategic coalition $\mcal{C}$, 
for any real-world \ppt strategy $S$ adopted by $\mcal{C}$, 
there exists an ideal-world strategy 
$S'(\cdot)$  
which may additionally depend on the coalition's true value
vector ${\bf v}_{\mcal{C}}$, 
such that for any polynomial $n(\cdot)$, 
\elaine{define notation, we only need it for bounded depth?}
}
\[
\E\left[
{\sf UReal}_{\mcal{C}}^S(1^\lambda, {\bf v}_{\mcal{C}}, {\bf v}_{-\mcal{C}})
\right]
\leq 
\E\left[
{\sf UIdeal}_{\mcal{C}}^{S'({\bf v}_{\mcal{C}})}(1^\lambda, {\bf v}_{\mcal{C}}, {\bf v}_{-\mcal{C}})
\right] + \negl(\lambda)
\]
\ignore{
\begin{align*}
\{{\sf UReal}_{\mcal{C}}^S(1^\lambda, {\bf v}_{\mcal{C}}, {\bf v}_{-\mcal{C}})\}_{\lambda, {\bf v}_{\mcal{C}} \in \inDomSecu^{|\mcal{C}|}, 
{\bf v}_{\mcal{-C}} \in \inDomSecu^{\leq n(\lambda)}} 
\displaystyle\mathop{\preceq}_{\text{depth-$D$}}
\{{\sf UIdeal}_{\mcal{C}}^{S'({\bf v}_{\mcal{C}})}(1^\lambda, {\bf v}_{\mcal{C}}, {\bf v}_{-\mcal{C}})\}_{\lambda, {\bf v}_{\mcal{C}} \in \inDomSecu^{|\mcal{C}|}, 
{\bf v}_{\mcal{-C}} \in \inDomSecu^{\leq n(\lambda)}}
\end{align*}
}
where 
\begin{itemize}[leftmargin=6mm,itemsep=1pt]
\item 
${\sf UReal}_{\mcal{C}}^S(1^\lambda, {\bf v}_{\mcal{C}}, {\bf v}_{-\mcal{C}})$
is the random variable representing the utility of the coalition $\mcal{C}$
when we execute $\Pi^\Real_\lambda$ 
with security parameter $1^\lambda$, honest buyers taking the true values
${\bf v}_{-\mcal{C}}$, and the coalition $\mcal{C}$ taking
the true values ${\bf v}_{\mcal{C}}$ and the strategy $S$; and 
\item 
${\sf UIdeal}_{\mcal{C}}^{S'}(1^\lambda, {\bf v}_{\mcal{C}}, {\bf v}_{-\mcal{C}})$
is similarly defined but for the ideal protocol $\Pi^\Ideal_\lambda$. 
\end{itemize}
\end{enumerate}
\label{defn:sim}
\end{definition}

Note that in general, the utility-dominated emulation notion
is well-defined for any (possibly multi-round) ideal-world auction, even though
in this paper, we shall focus on ideal auctions 
in the model specified in \Cref{sec:idealmodel}.

\begin{theorem}[Ideal-real design paradigm]
Let 
$\Pi^\Ideal_\lambda$ 
and $\Pi^\Real_\lambda$ 
be auctions parametrized by $\lambda$
over a family of finite value domains $\inDomSecu$.  
Suppose that 
$\Pi^\Real_\lambda$ 
is a utility-dominated emulation of 
$\Pi^\Ideal_\lambda$ w.r.t. the strategic individual
or coalition $\mcal{C}$, and 
$\Pi^\Ideal_\lambda$ satisfies 
information-theoretic (Bayesian) incentive compatibility 
w.r.t. $\mcal{C}$. 
\ignore{
the property 
$P$  
in an information theoretic sense, 
\ignore{TODO: need to define 
IT notion for auction parametrized by secparam}
where 
$$P \in \left\{
\begin{array}{l}
\text{bIC, pIC, sIC, psIC, $c$-pbIC}, \\
\text{Bayesian bIC, Bayesian pIC, Bayesian sIC, Bayesian psIC, Bayesian $c$-pbIC}
\end{array}
\right\}$$ 
}
Then, $\Pi^\Real_\lambda$ satisfies 
computational (Bayesian) incentive compatibility 
w.r.t. $\mcal{C}$.
\label{thm:idealreal}
\end{theorem}
\begin{proof}
Suppose for the sake of contradiction
that $\Pi^\Real_\lambda$ does not satisfy 
computational incentive compatibility w.r.t. $\mcal{C}$.
This means 
that there is some polynomial $n(\cdot)$,
some real-world \ppt strategy $S$, some polynomial function $\poly(\cdot)$
such that  
  for infinitely many $\lambda$'s, 
there exist $n_H \leq n(\lambda)$,  
some 
value vector $\bfv_{-\mcal{C}} \in \inDomSecu^{n_H}$ corresponding to honest
buyers, 
some value vector $\bfv_{\mcal{C}} \in \inDomSecu^{n_C}$ corresponding
to the buyers in $\mcal{C}$, 
such that 
\[
\E\left[{\sf util}^S_{\mcal{C}}(1^\lambda, {\bf v}_{-\mcal{C}}, {\bf v}_{\mcal{C}}) \right]
\geq 
\E\left[
{\sf util}^H_{\mcal{C}}(1^\lambda, {\bf v}_{-\mcal{C}}, {\bf v}_{\mcal{C}}) 
\right]
+ 1/\poly(\lambda)
\]

By the definition of utility-dominated emulation,
and the fact that real and ideal executions
are utility equivalent in an all-honest execution,  
we conclude that 
there is some strategy $S'$ 
that outperforms the honest strategy 
in the ideal world by $1/\poly(\lambda) - \negl(\lambda)$. 
This violates the fact that the ideal-world auction 
satisfies incentive compatibility w.r.t. $\mcal{C}$.
The proof for the Bayesian notion is similar, except
that $\bfv_{-\mcal{C}}$
will be sampled at random 
from $\inDomSecu^{n_H}$ instead. 
\end{proof}

\begin{remark}[An alternative definition of computational incentive compatibility]
In fact, \Cref{thm:idealreal} also gives an alternative
way of defining computational incentive compatibility, by  
requiring it to be a utility-dominated
emulation of some ideal auction that is incentive compatible in the information
theoretic sense. 
This alternative definition implies our earlier formulations  
in 
\Cref{def:comp-IC}
and \Cref{def:comp-bayesian-IC}. 
\label{rmk:alternate-comp-ic}
\end{remark}

\section{Compiling an Ideal Auction to a Robust Real-World Protocol}
\label{sec:robust}

\elaine{TODO: add depth to emulation thm. add depth to ideal world
2nd price auction, remove coin for RS}

\elaine{TODO: AoK needs to have adaptive knowledge soundness
VC AoK have crs. performance, change VC.Commit to Digest}  

\elaine{TO FIX: for compiler, x, p, mu must be PPT}

\elaine{TODO: change Knowledge Soundness definition extractor running time
to quasipoly ?? chiesa's thm says quasipoly}

\elaine{NOTES: TODO items:
1.  change the defn of utility-dominated emulation to use computational
stochastic dominance,  
strengthen ideal-world ic guarantee to avoid the non-negative constraint.
2.  use katz or the later paper's defns for the nm timed commitments.
it assumes that the distinguisher is depth bounded.
defined using cca security against depth bounded adv. 
this way, we can remove the random oracle. 
3. change all the computational ic defns and prelim 
to use comp ind against depth-bounded adversaries
4. for the robust version, get rid of the RO.
need to parametrize carefully using the 
depth of the PCP extractor. 
killian's extractor (for extracting the pcp itself) should be parallel
and incurs not too much more than the runtime of the malicious prover. 
}

\elaine{TODO: add some narratives to explain the construction, and imp proofs}

\ignore{
A disadvantage the cryptographic auction 
introduced in \cref{sec:real-world}, is that a single 
strategic buyer can abort the auction by not sending the opening to the platform,
leading to the risk of denial-of-service attacks in practice.
To address this issue, we will construct a protocol that is robust against such strategic buyer deviations.
}

Given an ideal auction where we need a trusted entity $\Fauction$ 
to enforce the auction's rules, we can compile 
it into real-world cryptographic protocol
that does not rely on $\Fauction$. 
Ideally, we want the protocol 
to satisfy an additional robustness property defined below.
Given an auction protocol parameterized
by $\lambda$, we 
say that it is {\it robust}, iff 
the following holds: as long as the platform
is honest and not part of the coalition $\mcal{C}$, then
for any non-uniform \ppt strategy adopted by $\mcal{C}$,  
there exists a negligible function $\negl(\cdot)$ such that 
for every $\lambda \in \mathbb{N}$, 
with at least $1-\negl(\lambda)$ probability,  
all honest players  
accept the auction.


\elaine{NOTE: need the assumption that adv cannot hijack honest ids}

\subsection{Compiler}
\label{sec:robustcompiler}


\ignore{
\paragraph{Assumptions.}
We shall make the following assumptions:
\begin{enumerate}[leftmargin=6mm,itemsep=1pt]
\item 
We assume that the any value in the  value domain
can be encoded with at most $\lambda$ bits. 
This assumption is without loss of generality due to rescaling.
\item 
We assume that 
the rules 
of the ideal auction 
can be computed by logarithmic-depth circuits.
Formally, 
given an arbitrary polynomial $n(\cdot)$, 
and $n(\lambda)$ bids each at most $\lambda$ bits long, 
there is a circuit 
that can 
compute the rules 
${\bf x}$, ${\bf p}$ 
and $\mu_{\mcal{S}}$
in depth at most $O(\log(n\cdot \lambda))$. 
\end{enumerate}
}

\paragraph{Cryptographic building blocks.}
We will make use of the following cryptographic building blocks.
\begin{itemize}[leftmargin=6mm,itemsep=1pt]
\item 
{\it Vector commitment ${\sf VC}$.} 
We 
rely on a vector commitment scheme denoted
${\sf VC} = (\Gen, {\sf Digest}, {\sf Open}, {\sf Vf})$.  
We can instaniate the ${\sf VC}$ scheme
with a Merkle tree~\cite{merkle}. 
\item
{\it Erasure code ${\sf RS}$.}
We rely on an erasure code 
that can correctly reconstruct as long as
at least $2/3$ fraction of the locations 
in the codeword are correct.   
We can employ a Reed-Solomon code denoted
${\sf RS} = ({\sf Encode}, {\sf Recons})$. 
\item 
{\it Publicly verifiable non-malleable timed commitments ${\sf NITC}$.}
We rely on a publicly verifiable non-malleable 
timed commitments
denoted $(\Gen$, ${\sf Com}$, ${\sf ComVf}$, ${\sf DecVf}$, ${\sf FDec}, 
{\sf FDecVf})$. 
We can instantiate the scheme with 
Chvojka and Jager~\cite{CJ-nitc}. 
\item
{\it Succinct argument of knowledge ${\sf AoK}$.}
We additionally make use of a succinct argument 
of knowledge henceforth denoted ${\sf AoK} = (\Gen, P, V)$. 
We will instantiate the ${\sf AoK}$ scheme
using Kilian~\cite{kilian}
using a quasilinear Probabilistic Checkable Proof
(PCP)~\cite{shortpcp}. 
\end{itemize}

The formal definitions of these cryptographic building
blocks are presented in 
\Cref{sec:bldgblock}. 

\paragraph{NP language.}
We specify the following NP language which will be used in our
protocol. Since the terms ${\sf NITC}.\crs$ and ${\sf VC}.\crs$ will
be clear from the context, we do not explicitly carry them in the statement
below.   A statement $(\digest, \digest', n, r_\platform)$ 
is in the language $\mcal{L}_\lambda$, iff there
exists a witness of the form 
$\left({\sf code}, 
\mcal{I}, 
\{c_j, \pi^{\rm com}_j, {\pi}^*_j, v_j, r_j, 
{\sf out}_j \}_{j \in \mcal{I}} 
\right) 
$
such that 
the following claims hold:
\begin{itemize}[leftmargin=6mm,itemsep=1pt]
\item 
all identities in $\mcal{I}$ are distinct, and $|\mcal{I}| = n$;  
\elaine{TO FIX: this needs to appear in proof}
\item 
${\sf code} = {\sf RS}.{\sf Encode}(\bfc)$
where $\bfc := \{j, c_j, \pi^{\rm com}_j\}_{j \in \mcal{I}}$; 
\item 
$\digest$ is the first output of ${\sf VC}.{\sf Digest}({\sf VC}.\crs, \bfc)$;
\item 
$\digest'$ is the first output of ${\sf VC}.{\sf Digest}({\sf VC}.\crs, 
\{i, c_i, \pi_i^{\rm com}, 
{\sf out}_i\}_{i \in \mcal{I}})$;
\item 
$\forall i \in \mcal{I}$, 
${\sf NITC}.{\sf ComVf}({\sf NITC}.\crs, c_i, \pi^{\rm com}_i) = 1$;
\item 
$\forall i \in \mcal{I}$, either
${\sf NITC}.{\sf DecVf}({\sf NITC}.\crs, c_i, i\|v_i\|r_i, {\pi}^*_i) = 1$
or ${\sf NITC}.{\sf FDecVf}({\sf NITC}.\crs, c_i, i\|v_i\|r_i, {\pi}^*_i) = 1$.
\item 
$\{{\sf out}_i\}_{i \in \mcal{I}}$
are the correctly computed outcomes 
by applying the rules $({\bf x}, {\bf p}, \mu_\seller)$
over randomness $ r 
= \left(\displaystyle\mathop{\oplus}_{j \in \mcal{I}} r_j\right) \oplus r_\platform$
 and input $\bfv = \{v_j\}_{j\in \mcal{I} \backslash \{0\}}$.
\end{itemize}

\paragraph{Compiler.}
We assume that every buyer has a distinct
identity $i$.  
The strategic players can choose fake identities and pretend to 
be multiple buyers. 
Given an ideal auction 
$\Pi^\Ideal_\lambda$
over the value domain $\inDomSecu \subset [0, 1]$ 
specified by the rules
$({\bf x}, {\bf p}, \mu_{\mcal{S}})$ which are \ppt algorithms, 
we describe
a new compiler that compiles it to a robust cryptographic protocol. 

\begin{mdframed}
\begin{center}
{\bf Robust cryptographic auction $\Pi_\lambda^{\rm real}({\bf x}, {\bf p}, \mu_{\mcal{S}})$
}
\end{center}
\paragraph{Inputs:}
All players have the security parameter $1^\lambda$. 
Every buyer with identity $i$ has a true value $v_i \in \inDomSecu$.

\paragraph{Setup:}
Let ${\sf NITC}.\crs \getr {\sf NITC}.\Gen(1^\lambda, T)$,
${\sf VC}.\crs \getr {\sf VC}.\Gen(1^\lambda)$, 
and 
${\sf AoK}.\crs \getr {\sf AoK}.\Gen(1^\lambda)$. 
Publish ${\sf NITC}.\crs$, ${\sf VC}.\crs$, and ${\sf AoK}.\crs$.

\paragraph{Main auction protocol:}
\begin{enumerate}[label=(\alph*),leftmargin=8mm]
\item
\label{step:robustcom}
{\bf Time renamed to $0$ (protocol start):}
Each buyer $i$ samples $r_i \getr \{0, 1\}^\lambda$, 
and ${\sf coin}_i \getr \{0, 1\}^\lambda$, 
and computes $(c_i, \pi_i^{\rm com}, \pi_i^{\rm dec}) \leftarrow
{\sf NITC}.{\sf Com}({\sf NITC}.\crs, i \|v_i\|r_i; {\sf coin}_i)$. 
Buyer $i$ sends $(i, c_i, \pi_i^{\rm com})$
to the platform. 

The seller samples $r_0 \getr \{0, 1\}^\lambda$ at random, 
computes $(c_0, \pi_0^{\rm com}, \pi_0^{\rm dec}) \leftarrow
{\sf NITC}.{\sf Com}({\sf NITC}.\crs$, $0 \|\bot\|r_0$; 
${\sf coin}_0)$, and sends 
$(0, c_0, \pi_0^{\rm com})$
to the platform. 
Without loss of generality, we may 
assume that the seller has the identity $0$ that is distinct
from all buyers. 
\item 
{\bf Time $T_1(\lambda)$:}
The platform ignores all tuples $(i, c_i, \pi_i^{\rm com})$
where ${\sf NITC}.{\sf ComVf}({\sf NITC}.\crs, c_i$, $\pi_i^{\rm com})$ $= 0$,
and ignores any tuple of the form $(0, \_, \_)$ that is not sent by
the seller. 
Among the remaining tuples, 
if there are multiple tuples containing the same identity $i$,
the platform randomly keeps one of the tuples and ignores the rest.
For each identity ignored, the platform notifies
the corresponding buyer,  
and the corresponding buyer rejects the auction. 

\item 
Let $\bfc = \{(i, c_i, \pi_i)\}_{i \in \mcal{I}}$
be the remaining tuples where $\mcal{I}$ denotes 
the identities 
they are associated with. 
The platform samples randomness $r_{\platform}$,
and computes 
$(\digest, \_)$ $= {\sf VC}$.${\sf Digest}({\sf VC}.\crs$, 
${\sf code})$ where 
${\sf code} := {\sf RS}.{\sf Encode}(\bfc)$.  
It sends
$n = |\mcal{I}|$, $\digest$ and $r_\platform$ 
to all buyers and the seller.
\label{step:pcom}
\item 
\label{step:porch}
{\bf Time $T_2(\lambda)$:}
If a buyer or seller still has not received a tuple $(n, \digest, r_\platform)$ 
from the platform, reject the auction.
Else, 
each buyer and the seller challenges the platform as follows.  
Sample $\kappa = \omega(\log \lambda)$ random indices
$Q \subset [n']$,
where $n'$ is the length of an ${\sf RS}$ codeword 
for an input of length $n$,
 and send $Q$ to the platform. 
The platform responds with 
the opening of the indices $Q$ along with a membership proof
w.r.t. $\digest$.

\item 
\label{step:porres}
{\bf Time $T_3(\lambda)$:}
If a buyer or the seller  
still has not received from the platform  
a response for its challenge, 
or if the platform's response fails to verify,
reject the auction. 

If a buyer $i \in \mcal{I} \backslash \{0\}$ did not abort in the above, then send
the opening $v_i \|r_i$
along with the opening proof $\pi_i^{\rm dec}$
to the platform. 
If the seller did not abort, then send $\bot \|r_0$
along with the opening proof $\pi_0^{\rm dec}$ to the platform. 

\item 
\label{step:fdec} 
{\bf Time $T_4(\lambda)$:}
Let $(v_i\|r_i, \pi_i^{\rm dec})$ denote the 
opening received by the platform from identity $i$. 
If there exists some $i \in \mcal{I}$, 
such that ${\sf NITC}.{\sf DecVf}({\sf NITC}.\crs, c_i, 
i\|v_i\|r_i, \pi_i^{\rm dec}) = 0$ 
or $i$ did not respond with an opening, 
then the platform 
performs forced decryption 
$(i\|v_i\|r_i, \pi^{\rm fdec}_i)
\leftarrow {\sf NITC}.{\sf FDec}({\sf NITC}.\crs, c_i, 
\pi_i^{\rm com})$.  
If the forced decryption does not produce
a valid decryption proof, the platform posts $\bot$ 
to the blockchain, and all players reject the auction. 
\elaine{TODO: add a property to NITC}
\item 
\label{step:compute}
The platform runs the 
auction's rules $({\bf x}, {\bf p}, \mu_{\mcal{S}})$
using the bid vector $\{v_i\}_{i \in \mcal{I} \backslash \{0\}}$ 
and the randomness 
$r = \left(\displaystyle\mathop{\oplus}_{i \in \mcal{I}} r_i\right) 
\oplus r_\platform $ as input. 
Let ${\sf out}_i$ be $i$'s private outcome
for $i \in \mcal{I}$. 

The platform then computes $(\digest', \_) = 
{\sf VC}.{\sf Digest}({\sf VC}.\crs, 
\{i, c_i, \pi^{\rm com}_i, 
{\sf out}_i \}_{i \in \mcal{I}})$
and sends $\digest'$
to all buyers and the seller. 

\item 
\label{step:robustaok}
Each buyer $i \in \mcal{I} \backslash \{0\}$ and the seller invokes a 
separate $\AoK$ instance with the platform using $\AoK.\crs$, 
where the platform proves
that $(\digest, \digest', n, r_\platform)$ is in the language $\mcal{L}_\lambda$
using the witness 
$$\left({\sf code}, 
\mcal{I}, 
\{c_j, \pi^{\rm com}_j, {\pi}^*_j, v_j, r_j, 
{\sf out}_j \}_{j \in \mcal{I}} 
\right)$$ where 
${\pi}^*_j = \pi^{\rm dec}_j$
if $j$ provided a valid opening, else 
${\pi}^*_j = \pi^{\rm fdec}_j$. 
If for any buyer or the seller, 
the $\AoK$
proof fails to verify, 
it rejects the auction. 
\item 
\label{step:privout}
For each $i \in \mcal{I}$ that responded
with a valid opening earlier, the platform
sends to $i$ its private outcome ${\sf out}_i$, 
as well as a membership proof for  
the tuple 
$(i, c_i, \pi_i^{\rm com}, {\sf out}_i)$ w.r.t. $\digest'$
where $(i, c_i, \pi_i^{\rm com})$ is the message $i$ sent the platform 
in Step \ref{step:robustcom}. 
The player $i$ 
rejects if the verification fails. 

\ignore{
Similarly, the platform opens $(t, \mu_\seller)$ to the seller,
along with the membership proof for 
$(0, c_0, \pi_0^{\rm com}, t, \mu_\seller)$ w.r.t. $\digest'$.
The seller rejects if the verification fails. 
}
\item 
\label{step:finalize}
The platform posts $(\digest, \digest', n, r_\platform)$
to the blockchain. Everyone reads 
the blockchain 
and verifies that they have received the same values from the platform earlier
in the protocol. Reject if the verification fails. 
Otherwise, accept the private outcome they have received. 
\end{enumerate}

\end{mdframed}

\elaine{robustness assumes C cannot hijack honest identities}

\elaine{remove T0}

\paragraph{Intuition.}
Informally, the protocol works as follows --- 
for simplicity, we describe the intuition
for the case when the auction's rules $(\bfx, \bfp, \mu_\seller)$ 
are deterministic.
In Step~\ref{step:robustcom}, 
all buyers commit to their bids
using a non-malleable timed commitment scheme. 
In Step~\ref{step:pcom}, 
the platform 
computes a Reed-Solomon encoding of 
of all buyers' commitments, and commits to the resulting code 
using a hash digest (denoted ${\sf digest}$). 
In Steps~\ref{step:porch}
and \ref{step:porres}, 
all buyers and the seller
challenge the platform to open  
a small number of locations committed under ${\sf digest}$ and verify
that the openings are correct. 
\elaine{can we avoid the seller sending challenge?
for the snark version, if we want seller not to participate, the challenge
needs to be computed with FS}
Then, all buyers open their commitments to the platform. 
In Step~\ref{step:fdec}, 
if any buyer fails to open, the platform force-decrypts
its timed commitment.
In Step~\ref{step:compute}, 
the platform 
computes all buyers' and the seller's private outcomes, and commits
to a hash digest (denoted $\digest'$)
of all original timed commitments as well as every one's private outcome.  
In Step~\ref{step:robustaok}, 
the platform 
invokes a succinct argument-of-knowledge to prove every one that 
1) $\digest$ and $\digest'$ encode the 
same vector of committed bids denoted ${\bf c}$; 
2) the vector 
${\bf c}$ does not contain multiple commitments from the same identity; 
and 3) the outcomes committed under $\digest'$ 
are computed correctly based on valid openings of ${\bf c}$. 
In Step~\ref{step:privout}, 
the platform opens to  each buyer and the seller
its private outcome. To each buyer, it provides a succinct proof  
that its timed commitment 
as well as private outcome  
belong to the set represented by $\digest'$. 
To the seller, it provides a succinct proof 
that its private outcome belongs
to the set represented by $\digest'$. 
Finally, in Step~\ref{step:finalize}, 
the platform posts $\digest$ and $\digest'$ to the blockchain
and all other players only accept if they agree
with the $\digest$ and $\digest'$ they have received from the platform earlier. 

It may be helpful to think of Steps~\ref{step:pcom}, \ref{step:porch}
and \ref{step:porres}
as a proof-of-retrievability (PoR)
protocol~\cite{por,compactpor}. 
If the platform can respond correctly to the PoR challenges,
it is possible to rewind the platform and extract 
an underlying witness, i.e., a vector of timed commitments. 
Jumping ahead, later in our proof, we will 
use the extracted witness 
to construct 
a reduction that breaks the IND-CCA security of the timed commitment 
scheme.
The depth of the extractor will contribute to the depth
of the reduction --- fortunately, the PoR's extractor indeed has small depth. 
Because we have a full AoK later 
in Step~\ref{step:robustaok}, we do not need
the platform to provide a proximity proof 
that the committed code is close to a valid 
Reed-Solomon codeword in the PoR stage. 

\paragraph{Parameter choices.}
We will choose the ${\sf NITC}$ scheme's parameters
such that it satisfies $(W(\cdot), \epsilon)$-IND-CCA security
for any super-polynomial function $W(\cdot)$.
Further, its difficulty 
parameter $T(\lambda)$ and 
the gap parameter $\epsilon$ should 
satisfy the following constraint: 
\[
T^\epsilon(\lambda)
> T_2(\lambda) + T_3(\lambda) 
+ D_{\rm auction}(\lambda^{\alpha(\lambda)}) + 
\poly\log(\lambda)
\]
where 
$\alpha(\lambda)$ is any super-constant function, 
and 
$D_{\rm auction}(\lambda^{\alpha(\lambda)})$
is an upper bound on 
the depth for executing the auction's rules
and computing the coalition's utility  
for any input whose length is an unknown polynomial 
in $\lambda$ --- as long as  
is $D_{\rm auction}(L) = L^{o(1)}$, 
$D_{\rm auction}(\lambda^{\alpha(\lambda)})$
is polynomially bounded for a 
suitably small super-constant 
function $\alpha(\lambda)$. 
\ignore{
\item 
$T_2$ is the maximum duration for the 
platform to produce the tuple $(n, \digest, r_\platform)$, 
and 
\item $T_3$ is the maximum duration for the platform
to compute a response for the vector commitment challenge $Q$,
}

The above specifies the parameter constraints needed
for proving the incentive compatibility guarantees.  
Besides the above, we should appropriately set $T_1$ such that  
buyers have enough time to compute and submit their commitments 
(Step~\ref{step:robustcom}); 
set $T_2 - T_1$ such that the platform  
has enough time to compute and send the tuple $(n, \digest, r_\platform)$
to every one (Step~\ref{step:pcom}); 
set $T_3 - T_2$ such that the platform
has enough time to compute and send 
a response to the PoR opening challenge (Steps~\ref{step:porch} and \ref{step:porres}), 
and set $T_4 - T_3$ such that buyers have enough 
time to send their openings (Step~\ref{step:porres}).  
\ignore{
If $T_1$, $T_2 - T_1$, $T_3 - T_2$ are set to be too small,
it might affect the protocol's liveness but will not     
affect the protocol's the incentive compatibility guarantees. 
If $T_4-T_3$ is too small, it might trigger the platform
to perform more forced decryptions thus increasing its computation overhead. 
}


\paragraph{Asymptotical efficiency.}
It is not hard to see that 
the seller and each buyer's computational 
and bandwidth cost
are bounded by $\widetilde{O}_\lambda(1)$ 
where $\widetilde{O}_\lambda(\cdot)$
hides $\poly(\lambda, \log n)$ 
factors where $n$ is the total number of buyers. 
The platform's bandwidth is upper bounded
by the sum of the bandwidth costs of all buyers and the seller, 
that is, 
$\widetilde{O}_\lambda(n)$.
The platform's computation 
is dominated by  
Step \ref{step:robustaok} where the platform
needs to run an $\AoK$ instance 
to convince each buyer and the seller. 
Notice that the platform is proving the same statement
to all other players.
So if we use Kilian to instantiate 
the $\AoK$, the platform only needs
to compute the PCP proof once as well as the vector
commitment of the PCP proof.   
Then, it can answer each verifier's challenge with the same  
PCP proof and vector commitment. 
Therefore, if we use a quasilinear PCP~\cite{shortpcp}, 
the platform's total computation is bounded by 
$\widetilde{O}_\lambda(n)$.
Last but not the last, the platform posts only $\widetilde{O}_\lambda(1)$
bits to the broadcast channel at the end of the auction.

\paragraph{More practical variants.} 
In practice, we can use 
a Succinct Non-interactive ARgument of Knowledge 
(SNARK)~\cite{qap,groth16,stark,libra,virgo,marlin,plonk,hyperplonk,aurora,spartan}
\elaine{cites}
to replace the AoK. 
This would result in more a practical 
instantiation but at the price of 
worsening the cryptographic assumptions needed, 
since 
known SNARK constructions
typically rely on either the random oracle model
or knowledge-type assumptions. 
Additionally, for the PoR challenge and response steps (Steps~\ref{step:porch}  
and \ref{step:porres}), we can use the Fiat-Shamir paradigm
to remove the interaction. Specifically, after computing $\digest$,
the server can use $H(\digest, i)$ to compute the random challenges 
for player $i$ where $H(\cdot)$ denotes a random oracle.  

Suppose we instantiate the protocol with a SNARK proof in lieu of the $\AoK$,
and using Fiat-Shamir to make the PoR challenge and response steps 
non-interactive, 
then, the only remaining message for the seller to send 
is the commitment of its random coins $\coin_0$ in Step~\ref{step:robustcom}. 
We can remove even this message from the protocol, and get a version that supports
an {\it offline server}. 
Specifically, the offline seller can  
query the platform a-posteriori to get its private outcome
along with the corresponding SNARK and ${\sf VC}$ 
proofs (i.e., Steps~\ref{step:robustaok}
and \ref{step:privout}), 
and it can check that 
the outcome is correct against the messages that have been posted
to the blockchain. 
It is not hard to mechanically go through our proofs and verify that 
this offline-server variant is also a utility-dominated emulation
of the ideal auction, as long as  
one of the following is true: 1)  
the auction's rules $(\bfx, \bfp, \mu_\seller)$ are deterministic,
2) there is at least one honest buyer, or 3) the platform is honest. 
\elaine{TODO: double check this}

Finally, if the SNARK is instantiated
with modern schemes~\cite{plonk,hyperplonk,libra,virgo,aurora,spartan,marlin}
that combine polynomial commitments and interactive oracle proofs (IOPs), 
it is typically desirable to instantiate the vector commitment 
using the polynomial commitment scheme 
native to the SNARK construction.

\ignore{
then, {\it the seller need not send any messages in the resulting protocol}.  
This can allow us to support an {\it offline seller} as follows.  
Basically, the seller can contact the platform offline to discover
its outcome in the auction, and it can verify the result  
against the messages that have been posted to the blockchain. 
For this modification to work, the seller needs to assume
that there is at least one honest buyer (or the platform is honest), 
}


\elaine{TODO: comment on the version where seller need not 
send message}

\begin{theorem}[Efficient cryptographic auction]
Suppose that the 
${\sf NITC}$ scheme is computationally binding  
against quasi-polynomially sized adversaries, 
IND-CCA secure 
against quasi-polynomially sized and $T^\epsilon$-depth adversaries,
and satisfies soundness of forced decryption 
(for polynomially-sized adversaries). 
Additionally, suppose that ${\sf VC}$ satisfies
collision resistance against quasi-polynomially sized
adversaries, 
and $\AoK$ 
satisfies adaptive knowledge soundness for polynomially-sized
adversaries. 
Let $\mcal{C}$ be either a strategic individual buyer, the strategic  
seller alone, the strategic platform alone, 
the strategic platform-seller 
coalition, or a coalition between 
the platform and one or more buyers.
Suppose also that 
the auction's rules $({\bf x}, {\bf p}, \mu_\seller)$
can be computed in $\ell^{o(1)}$ depth
on any input of length $\ell$. 

Then, the real-world auction 
$\pi^{\rm real}_\lambda$ 
satisfies robustness, and 
is a utility dominated emulation of 
the ideal auction  
$\pi^{\rm ideal}_\lambda$ 
w.r.t. $\mcal{C}$. 
\elaine{double check}
\label{thm:robustcompiler}
\end{theorem}
We provide the robustness proof
in \Cref{sec:proof-robust}, and 
the proofs 
for utility-dominated emulation in \Cref{sec:proof-emu1,sec:proof-emu2}.

\begin{remark}[Basing the protocol on standard assumptions when robustness
is not needed] 
If we do not need the robustness property, we can replace
the non-malleable timed commitments
with a computationally hiding 
and perfectly binding commitment,
plus a simulation extractable 
Non-Interactive Zero-Knowledge Proof (NIZK)  
proof that the commitment is computed correctly. 
With this change, 
the proof of \Cref{clm:useindcca} would use the extraction key
of the NIZK scheme to extract the bid under the commitment
 rather than calling the ${\sf Dec}$ oracle of the IND-CCA challenger. Thus, 
\Cref{clm:useindcca}  can be proven using the 
simulation extractability 
of the NIZK and the computational hiding property of the commitment  
against quasi-polynomially sized adversaries. 
\end{remark}

\subsection{Proofs of Robustness}
\label{sec:proof-robust}

\elaine{TODO: make explicit where we use at least one honest buyer assumption}

We prove the following lemma 
showing that our real-world protocol satisfies robustness.

\begin{lemma}
Suppose that ${\sf NITC}$ satisfies soundness of forced decryption
for polynomially-sized adversaries, and that strategic players cannot  
hijack honest players' identities.
Then,  
the real-world protocol $\Pi^{\rm real}_\lambda$
satisfies robustness as defined at the beginning of 
\Cref{sec:robust}. 
\label{lem:robust}
\end{lemma}
In the above, the assumption that 
strategic players cannot
hijack honest players' identities
is mild 
and reflects realistic deployment scenarios. 
For example, 
this assumption is respected if the players' identities  
include their email addresses, 
or government-issued identities. Alternatively, 
if honest players
choose their cryptographic identities at random by sampling
public and private key pairs, 
then the probability that honest keys are hijacked
by strategic players is also negligibly small. 

\begin{proof}(of \Cref{lem:robust})
Suppose that the platform is honest. 
Since the strategic
coalition cannot hijack honest players' identities, 
they cannot cause honest players to reject due to their identities
being suppressed. 
Another potential way 
for the strategic coalition to cause
honest players to reject is 
for some strategic buyer or seller denoted $i$ 
to submit an ${\sf NITC}$ commitment, fail
to provide the opening, and somehow the honest platform
fails to get a valid $\pi_i^{\rm fdec}$ 
during the forced decryption. 
Due to the soundness of forced decryption
of the underlying ${\sf NITC}$, this happens with only negligible probability.
\end{proof}

\subsection{Proofs of Utility-Dominated Emulation: 
When the Platform is Strategic}
\label{sec:proof-emu1}
Throughout the proof, we assume that 
$\inDomSecu \subset [0,1]$ is a discretized subset
of the normalized value domain $[0,1]$ --- see
also \Cref{sec:IC-computational}.

\paragraph{Experiment $\Hyb_1$.}
$\Hyb_1$ is almost the same as the real-world execution except
with the following modifications.
\ignore{
First, 
to some honest player, 
the platform sent 
$r_\platform$, 
$n$, $\digest$ and $\digest'$ values 
that differ 
from what it posted to the blockchain, 
 then  
we abort and treat the coalition's utility as $0$.
Second, 
we make the following modification. 
}
There must be one honest verifier $V$ (either a buyer or seller) 
in the Step \ref{step:robustaok}
of the protocol. Therefore, we can run the \AoK's extractor 
along the side and 
extract a witness 
of the form 
$\left(\widetilde{\sf code}, 
\widetilde{\mcal{I}}, 
\{\widetilde{c}_j, \widetilde{\pi}^{\rm com}_j, \widetilde{\pi}_j, 
\widetilde{v}_j, \widetilde{r}_j, 
\widetilde{{\sf out}}_j\}_{j \in \widetilde{\mcal{I}}}\right)$
for the \AoK instance where $V$ is the verifier. 
Abort and treat the coalition's utility as $0$ if any of the
following bad events happens:
\begin{enumerate}[leftmargin=7mm,itemsep=1pt]
\item the witness is not valid; 
\item 
some honest $i \in \mcal{I}$ is not included in 
the extracted $\widetilde{\mcal{I}}$;  
\item 
for some $i \in \widetilde{\mcal{I}}$, 
${\sf NITC}.{\sf FDec}({\sf NITC}.\crs, \widetilde{c}_i, 
\widetilde{\pi}_i^{\rm com}) \neq i \| \widetilde{v}_i\|\widetilde{r}_i$;
\item 
there exists some honest $i \in \mcal{I}$ 
that accepts a private
outcome ${\sf out}_i$ that is different from the 
extracted $\widetilde{{\sf out}}_i$. 
\item 
for some $i \in \mcal{I}$ that is honest, 
the message $i$ sent to the platform
in Step \ref{step:robustcom}
is not contained in $\widetilde{\bfc}$ where 
$\widetilde{\bfc} = \{j, \widetilde{c}_j, 
\widetilde{\pi}^{\rm com}_j\}_{j \in \widetilde{\mcal{I}}}$ comes from 
the extracted values;
\item 
for some $i \in \mcal{I}$ that is honest, 
the actual $v_i$ and $r_i$ values it used   
in the protocol 
differ from $\widetilde{v}_i$ and $\widetilde{r}_i$ from
the extracted witness --- here we assume that $v_0 = \bot$ if the seller
is honest.
\end{enumerate}

\elaine{check if negative util will matter}
\begin{claim}
Suppose that the \AoK scheme satisfies knowledge soundness,
the ${\sf NITC}$
scheme is computationally binding against 
quasi-polynomially sized adversaries, 
\elaine{param}
and ${\sf VC}$ 
is collision resistant against quasi-polynomially sized adversaries, 
\elaine{param}
then the coalition's utility in the real world
has negligible statistical distance 
from its utility in $\Hyb_1$. 
\label{clm:robust-checks}
\end{claim}
\begin{proof}
We will go over each bad event one by one, and show that 
except with negligible probability, 
if all honest players
accept in the end, 
the bad event  should not happen. 
\ignore{
We may assume that 
if all honest players accept in the end,  
then the platform must send the same $n$, $r_\platform$, $\digest$
and $\digest'$ terms to all honest players, and must post the same  
terms to the blockchain too.
}
\begin{enumerate}[leftmargin=7mm,itemsep=1pt]
\item 
Using a straightforward reduction to the knowledge
soundness of \AoK, we can show that the extracted witness is invalid
with only negligible probability.
Henceforth, we may assume that the extracted witness is valid.

\item 
If the extracted witness is valid, 
it must be 
$\digest' = {\sf VC}.{\sf Commit}({\sf VC}.\crs, 
\{i, \widetilde{c}_i, \widetilde{\pi}_i^{\rm com}, 
\widetilde{{\sf out}}_i\}_{i \in \widetilde{\mcal{I}}})$. 
If all honest players accepted at the end of the protocol,
it must be that 
the platform showed a valid membership
proof 
for some tuple $(i, \_, \_, \_, \_)$
w.r.t. $\digest'$ to every honest $i \in \mcal{I}$.  
Thus, 
if not all honest $i \in \mcal{I}$
is contained in $\widetilde{\mcal{I}}$, we can construct
a quasi-polynomial time reduction that outputs a collision
for ${\sf VC}$. 

\item 
If the extracted witness is valid, 
then, 
for all $i \in \widetilde{\mcal{I}}$, 
either 
${\sf NITC}.{\sf DecVf}({\sf NITC}.\crs, \widetilde{c}_i, i\|\widetilde{v}_i\|\widetilde{r}_i, \widetilde{\pi}_i) = 1$
or ${\sf NITC}.{\sf FDecVf}({\sf NITC}.\crs, \widetilde{c}_i, 
i\|\widetilde{v}_i\|\widetilde{r}_i, \widetilde{\pi}_i) 
= 1$ which is checked by the language $\mcal{L}_\lambda$.
Suppose for some $i \in \mcal{I}$, 
the message output by  
${\sf NITC}.{\sf FDec}({\sf NITC}.\crs, \widetilde{c}_i, \widetilde{\pi}_i^{\rm com})$ is not equal to 
$i \| \widetilde{v}_i\|\widetilde{r}_i$,
we can construct a quasi-polynomial time adversary 
that breaks the computationally binding property of 
${\sf NITC}$. 

\item 
If the extracted witness is valid, 
it must be 
$\digest' = {\sf VC}.{\sf Commit}({\sf VC}.\crs, 
\{i, \widetilde{c}_i, \widetilde{\pi}_i^{\rm com}, 
\widetilde{{\sf out}}_i\}_{i \in \widetilde{\mcal{I}}})$. 
Thus, if for some honest $i$, the platform can convince
it 
to accept some different ${\sf out}_i \neq \widetilde{{\sf out}}_i$
by producing a valid membership proof for some tuple 
$(i, \_, \_, {\sf out}_i)$,  
then we can build a quasi-polynomial time reduction
that breaks the collision resistance of the ${\sf VC}$ scheme.  

\item 
The proof is similar to (4) above. \elaine{hard-coded ref!!} 
\item 
If the extracted witness is valid, 
then, 
either 
${\sf NITC}.{\sf DecVf}({\sf NITC}.\crs, \widetilde{c}_i, i\|\widetilde{v}_i\|\widetilde{r}_i, \widetilde{\pi}_i) = 1$
or ${\sf NITC}.{\sf FDecVf}({\sf NITC}.\crs, \widetilde{c}_i, 
i\|\widetilde{v}_i\|\widetilde{r}_i, \widetilde{\pi}_i) 
= 1$. 
Given (5) \elaine{hard coded ref}, we know
that for an honest $i$, 
$\widetilde{c}_i = c_i$, 
$\widetilde{\pi}_i^{\rm com} = \pi_i^{\rm com}$
where $c_i$ and $\pi_i^{\rm com}$ are what $i$ sent
in the protocol. 
By the correctness of ${\sf NITC}$, 
\elaine{Is there negl correctness error?}
if all honest players accept in the end, $i$ must have also
sent $\pi_i^{\rm dec}$
such that 
${\sf NITC}.{\sf DecVf}({\sf NITC}.\crs, {c}_i, i\|{v}_i\|{r}_i, {\pi}_i^{\rm dec}) = 1$. 
If $\widetilde{v}_i \| \widetilde{r}_i \neq 
v_i \| r_i
$, we can construct a quasi-polynomial time reduction
that breaks the computationally binding property of ${\sf NITC}$. 
\end{enumerate}
\end{proof}

\elaine{given one honest verifier, extract 
witness from the AoK proof. 
if witness invalid or 
anyone's outcome not consistent with witness, abort,
coalition has arbitrary util. relies on AoK knowledge soundness, MT CR}

\paragraph{Experiment $\Hyb_2$.}
$\Hyb_2$ is almost the same as $\Hyb_1$ except
the following modifications.
We know that there is at least one honest verifier $V$ 
(either a buyer or seller)
for the challenge-response protocol in Steps \ref{step:porch}
and \ref{step:porres}, 
and we additionally run the following extractor 
along the side. 

\begin{itemize}[itemsep=1pt]
\item 
Let ${\sf code}'$
be an array of length $n'$ where all elements are 
initialized to $\bot$. 
\item
Repeat 
$\lambda^{\alpha(\lambda)}$ 
times in parallel 
where $\alpha(\cdot)$ is an arbitrary super-constant function:
\elaine{param}
\begin{itemize}
\item 
rewind the coalition's algorithm to right before it receives  
the challenge;
\item 
sample $\kappa$  fresh random indices denoted $Q$ 
on behalf of $V$ and send $Q$ to 
the coalition; 
\item 
if the coalition 
responds with openings $\{z_q\}_{q \in Q}$ 
with valid proofs within in 
$T_3(\lambda)$
time,  \elaine{param}
 then populate the entry 
${\sf code}'[q] := z_q$ for $q \in Q$.

\end{itemize}
\end{itemize}

We define the following bad event in $\Hyb_2$: 
all honest players accept in the end, but one of the following
happens:
\begin{enumerate}[leftmargin=8mm,itemsep=1pt]
\item 
fewer than 
$2/3$ fraction of the entries of ${\sf code}'$ have been populated; or
\item 
${\sf RS}.{\sf Recons}({\sf code}') \neq \widetilde{\bfc}$
where
$\widetilde{\bf c} = \{(j, \widetilde{c}_j, \widetilde{\pi}_j^{\rm com})\}_{j \in 
\widetilde{\mcal{I}}}$ is part 
of the witness extracted
by the \AoK extractor.
\end{enumerate}

\begin{claim}
Suppose that ${\sf VC}$ satisfies collision resistance against 
quasi-polynomial time adversaries. 
\elaine{param}
Then, the coalition's utility 
in $\Hyb_2$ is identically distributed as 
in $\Hyb_1$, and moreover, 
the probability that the bad event happens 
is negligibly small. 
\label{clm:prbadevent}
\end{claim}
\begin{proof}
$\Hyb_2$ does not modify the way the outcomes are computed,
it only runs an additional extractor
for the challenge-response protocol  
corresponding to Steps 
\ref{step:porch} and \ref{step:porres}
on the side. 
We now prove that the probability of encountering a bad event is 
negligibly small.

Since the witness output by the \AoK extractor is valid,
if all honest players accept at the end, they must 
all have the correct $n$, that is, 
$\digest$ is a vector commitment of a length $n'$-vector
where $n'$ is the length of the ${\sf RS}$ encoding
of a length-$n$ vector.

Applying Lemma 1 of Chiesa et al.~\cite{chiesa-kilian}, we can conclude
that except with negligible probability, 
the above extractor 
extracts at least $2/3$ locations, along with valid
membership proofs w.r.t. $\digest$. 
We know that 
${\sf digest} = 
{\sf VC}.{\sf Commit}({\sf VC}.\crs, \widetilde{\sf code})$
and $\widetilde{\sf code} = {\sf RS}.{\sf Encode}(\widetilde{\bfc})$
since the witness output by the \AoK extractor is valid. 
If there exists some location $i \in [n']$ such that
${\sf code}'[i] \neq \bot$ 
and 
${\sf code}'[i] \neq \widetilde{\sf code}[i]$, 
then we can construct a quasi-polynomial time reduction
that produces a collision for ${\sf VC}$.
If the above bad event does not happen, by the property
of ${\sf RS}$, 
it must be that
${\sf RS}.{\sf Recons}({\sf code}') = \widetilde{\bfc}$. 

\end{proof}

\paragraph{Experiment $\Hyb_3$.}
In $\Hyb_3$, 
if all honest players accept the auction,
we instead 
compute the 
coalition's utility as follows: 
\elaine{TODO: change to widetilde notation}
\begin{itemize}[leftmargin=6mm,itemsep=1pt]
\item 
reconstruct $\bfc = \{(i, c_i, \pi_i^{\rm com})\}_{i \in \mcal{I}}$ 
from the ${\sf code}'$ extracted from the challenge-response
protocol in Steps \ref{step:porch} and \ref{step:porres}; 
\item for $i\in \mcal{I} \backslash \mcal{H}$ where $\mcal{H}$ denotes
the honest players, 
let 
$(i\|v_i\|r_i, \_)\leftarrow {\sf NITC}.{\sf FDec}({\sf NITC}.\crs, 
c_i, \pi_i^{\rm com})$; 
\item 
compute the coalition's utility by running the rules
$({\bf x}, {\bf p}, \mu_\seller)$
on 
the following inputs and randomness:
\begin{itemize}[itemsep=1pt]
\item 
the honest buyer's true values 
$\{v_i\}_{i \in \mcal{H} \backslash \{0\}}$, 
and the strategic bids  
$\{v_i\}_{i \in \mcal{I} \backslash (\mcal{H} \cup \{0\})}$
extracted above;
\item 
randomness 
$r = \left(\displaystyle\mathop{\oplus}_{i \in \mcal{I}} r_i\right) 
\oplus r_\platform $ 
where 
$\{r_i\}_{i \in \mcal{H}}$ 
are the random coins chosen by honest players, 
$\{r_i\}_{i \in \mcal{I} \backslash \mcal{H}}$ are the 
strategic coins extracted above, and $r_\platform$
denotes the 
$r_\platform$ value the platform has sent to an arbitrary honest player.
\elaine{note they may not all be the same}
\end{itemize}
\end{itemize}
If the above computation fails (e.g., due to failure
of ${\sf RS}$ reconstruction), 
simply treat the coalition's  utility as $0$.

\begin{claim}
Under the same assumptions as \Cref{clm:prbadevent}, 
the coalition's utility in $\Hyb_3$ 
has negligible statistical difference from $\Hyb_2$.  
\end{claim}
\begin{proof}
If the bad events defined in $\Hyb_2$ do not happen, then 
the two ways 
of computing the auction's outcomes are equivalent by observing
that 
1) the witness extracted by the \AoK extractor is valid, i.e., the conditions checked
by  
the language $\mcal{L}_\lambda$ 
all hold, 
and 2)
the additional checks introduced by $\Hyb_1$ all hold. 
\end{proof}

\ignore{compute C utility from extracted codeword from PoR, decode using RS, and 
open using FDec.
relies on soundness of NITC
}

\paragraph{Experiment $\Hyb_4$.}
$\Hyb_4$ is almost the same as $\Hyb_3$, 
except that we stop running the experiment after 
Step 
\ref{step:porres}. 
We compute the coalition's utility like in $\Hyb_3$.
If the utility is less than $0$, 
then simply treat the coalition's utility as $0$. 

\begin{claim}
The coalition's utility in $\Hyb_4$ stochastically dominates
$\Hyb_3$. 
\end{claim}
\begin{proof}
Let ${\it tr}$ be an execution trace of $\Hyb_4$, 
let ${\it tr}'$ be a continuation of this trace for $\Hyb_3$. 
If not all honest players accept in ${\it tr}'$, then
the coalition's utility is $0$ in $\Hyb_3$, and its utility has to be at least
$0$ in $\Hyb_4$.
If all honest players accept in ${\it tr}'$, then 
either 
the coalition's utility 
in $\Hyb_3$ given ${\it tr}'$ is the same
as its utility 
in $\Hyb_4$ given ${\it tr}$, or 
the coalition's utility is negative in 
$\Hyb_3$ but rounded up to $0$ 
in $\Hyb_4$. 
\end{proof}

\paragraph{Experiment $\Hyb_5$.}
$\Hyb_5$ is almost the same as $\Hyb_4$, 
except that when the experiment needs to compute 
${\sf NITC}$ commitments on behalf of 
honest players, we replace the commitments with commitments of $0$. 

\begin{claim}
Suppose that the ${\sf NITC}$ scheme satisfies
IND-CCA security for 
quasi-polynomially sized \elaine{param}
and depth-$T^\epsilon$ adversaries, 
then,  the coalition's expected utility in $\Hyb_5$ 
is only negligibly apart from its utility in $\Hyb_4$.
\label{clm:useindcca}
\end{claim}
\begin{proof}
We may assume that the number of buyers is exactly $n-1$ 
since we an apply the same proof for every choice of $n$. 
Without loss of generality, we will index the seller as $0$, and index all  
honest buyers using $1, 2, \ldots, {n-1}$ below.  
The $i$-th buyer has identity ${\it id}_i$. Let ${\it id}_0 = 0$,
and 
let $v_0 = \bot$ if the seller is honest.

We will prove this claim through a sequence of inner hybrid experiments
denoted ${\sf H}_0, {\sf H}_1, \ldots, 
{\sf H}_n$. 
${\sf H}_0$ is the same as 
$\Hyb_4$, and 
for $i \in [n]$, 
${\sf H}_i$ is the same as ${\sf H}_{i-1}$ 
except for changing the $(i-1)$-th honest player's ${\sf NITC}$ commitment
to a commitment of $0$.
Clearly, ${\sf H}_n = \Hyb_5$. 

Due to 
\Cref{fct:indtoexp}, 
\elaine{TODO: this requires some lemma}
it suffices to show that 
no  polynomially sized,  
depth-$C\log(\cdot)$ distinguisher
can effectively distinguish 
the coalition's utility 
in each pair of adjacent hybrids, where $C$
is a suitably large constant. 

Suppose that for some $i$, 
there is some polynomially sized, 
depth-$C \log(\cdot)$ distinguisher $\algB$
that can distinguish 
the coalition's utilities in ${\sf H}_{i-1}$
and ${\sf H}_i$ with non-negligible probability.  
We will reach a contradiction by constructing the following
reduction $\algR$ that wants to break the 
IND-CCA security of ${\sf NITC}$. 
\begin{enumerate}[leftmargin=6mm,itemsep=1pt]
\item 
$\algR$ first prepares  
the commitments for all honest players except the $(i-1)$-th player.
Specifically, for any honest buyer $j < i - 1$, 
it will compute a commitment of $0$.
For any honest buyer ranked $j \geq i$,  
it will compute an honest commitment of its true value $v_j$
and some random $r_j$. 
Let the resulting commitments be 
$(c_0, \pi^{\rm com}_0), \ldots, (c_{i-2}, \pi_{i-2}^{\rm com}),
 (c_{i}, \pi^{\rm com}_{i}), \ldots, (c_{n-1}, \pi_{n-1}^{\rm com})$.
This preprocessing 
only needs to be quasi-polynomial
in total work and we need not worry about its depth.
\item 
Next, $\algR$ sends 
a challenge to the IND-CCA challenger on messages
$(0,  v_{i-1}||r_{i-1})$ where $v_{i-1}$ is player $(i-1)$'s true value
and $r_{i-1}$ is some random string.
It obtains some commitment $(c_{i-1}, \pi^{\rm com}_{i-1})$ from the IND-CCA challenger.
\item 
$\algR$ invokes 
the coalition's algorithm, and sends 
$\{{\it id}_i, c_i, \pi^{\rm com}_i \}_{i \in \{0, 1, \ldots, n-1\}}$
to the coalition 
$\mcal{C}$. 
$\mcal{C}$ responds with $(n, {\sf digest}, r_\platform)$.
\elaine{TODO: need to assert n is correct earlier}
$\algR$ now runs the extractor
defined in $\Hyb_2$ with $\mcal{C}$, and extracts  
some ${\sf code}'$ which reconstructs to $\bfc$
using the reconstruction algorithm of the ${\sf RS}$ code. 
$\algR$
now computes the coalition's utility using $\bfc$ 
just like in $\Hyb_4$, but where the calls to ${\sf FDec}$
are replaced with calls
to the ${\sf Dec}$ oracle provided by the IND-CCA challenger.  
\item 
$\algR$ now 
forwards the resulting utility 
to $\algB$ and outputs whatever it outputs. 
\end{enumerate}

Clearly, the work of $\algR$ is quasi-polynomially bounded. 
We now analyze the depth of $\algR$ in Steps 3 and 4.
\elaine{hard coded refs}
The depth of Step 3 
is upper bounded by 
\[
T_2(\lambda) + T_3(\lambda) 
+ D_{\rm rs}(\lambda^{\alpha(\lambda)}) + 
D_{\rm auction}(\lambda^{\alpha(\lambda)}) 
+ 
\poly\log \lambda
\]
where 
\begin{itemize}[leftmargin=6mm,itemsep=1pt]
\item 
$T_2(\lambda)$ is the maximum time for the coalition to produce 
the tuple $(n, \digest, r_\platform)$; 
\item 
$T_3(\lambda)$ is the maximum time for the coalition to respond
to an opening challenge for ${\sf VC}$;  
\item 
$D_{\rm rs}(\lambda^{\alpha(\lambda)})$
is an upper bound on the 
depth of the reconstruction algorithm of ${\sf RS}$
for an input of unbounded polynomial length; 
\item 
$D_{\rm auction}(\lambda^{\alpha(\lambda)})$
is an upper bound on the depth of the algorithm that
runs 
the auction's rules $({\bf x}, {\bf p}, \mu_\seller)$
and then computes the  
coalition's utility from the coalition's true value vector
and the outcomes of the auction
for any unbounded polynomial length input;  
\item 
and $\poly'\log(\lambda)$  
captures the depth of all other operations, e.g., for a machine 
that does not support concurrent writes, 
concurrent writes to the same array ${\sf code}'$ may take 
$O(\log(\lambda^{\alpha(\lambda)}))$ depth. 
\end{itemize}
The depth of Step 4 
\elaine{hard coded refs}
is  upper bounded by 
the maximum depth of the distinguisher $\algB$
whose depth is logarithmic in the input length.
Because the input length can be an unbounded polynomial,
we can upper bound the depth of $\algB$ by $\log^2 \lambda$
for sufficiently large $\lambda$.   
Because we do not need any error correction,
the reconstruction algorithm of ${\sf RS}$
simply computes a linear combination, and its depth
is logarithmic in the input length. 
Therefore, for a suitably small $\alpha(\cdot)$,
and sufficiently large $\lambda$, 
$D_{\rm rs}(\lambda^{\alpha(\lambda)})$
is upper bounded by $\log^2\lambda$. 
Summarizing the above, 
the total depth of $\mcal{R}$ in Steps 3 and 4 
\elaine{hard coded refs}
is upper bounded by 
\[
T_2(\lambda) + T_3(\lambda) 
+ D_{\rm auction}(\lambda^{\alpha(\lambda)}) 
+ 
\poly\log \lambda
\]


Recall that we choose the parameters $T$ and $\epsilon$ 
of the ${\sf NITC}$ to satisfy 
\[
T^\epsilon(\lambda)
> T_2(\lambda) + T_3(\lambda) 
+ D_{\rm auction}(\lambda^{\alpha(\lambda)}) + 
\poly\log(\lambda)
\]
Therefore, if $\algB$ can effectively distinguish
the coalition's utility in ${\sf H}_{i-1}$ and ${\sf H}_i$, 
then $\algR$ would  
break the IND-CCA security of the underlying 
${\sf NITC}$.

\elaine{TO FIX: Drs and Dauction must be an UB for all polynomial n}

\end{proof}

\paragraph{Ideal auction and ideal strategy.}
$\Hyb_5$ is almost equivalent to the following ideal auction
in terms of the coalition's utility, except
that when ${\sf NITC}.{\sf FDec}$ fails,
the coalition may have positive utility in the ideal auction below below,  
whereas in $\Hyb_5$ it has 0 utility: 
\begin{enumerate}[leftmargin=7mm,itemsep=1pt]
\item 
Honest buyers send their true values to 
$\Fauction$;
\item  
$\Fauction$ sends the number of honest bids $n_H$ to $\algC$;
\item  
$\algC$ 
computes $n_H$ ${\sf NITC}$ commitments
of $0$ and sends them to the real-world coalition's algorithm; 
when the real-world coalition's algorithm responds
with some $(n, \digest, r_\platform)$, it runs the extractor
algorithm defined in $\Hyb_2$
with the real-world coalition's algorithm, 
to extract and reconstruct $\bfc = \{i, c_i, 
\pi^{\rm com}_i\}_{i \in \mcal{I}}$. 
Then, for every $i \in \mcal{I}$ that is not an honest
identity,  
$\algC$ calls ${\sf NITC}.{\sf FDec}({\sf NITC}.\crs, c_i, \pi_i^{\rm com})$ 
to decrypt $v_i\|r_i$. 
If ${\sf FDec}$ failed, simply send $0$ bids on behalf of all strategic buyers
to $\Fauction$; otherwise, 
send the extracted bids corresponding to strategic buyers to $\Fauction$. 
\item 
$\Fauction$ 
chooses the coins $r$ at random, and 
computes the outcome 
of the auction using 
the honest players' values, the bids
submitted by $\algC$,  
the randomness $r$. 
It sends the outcome vector to $\algC$.
If the coalition's utility is less than $0$, 
the platform sends $\bot$ to $\Fauction$; else it sends {\sf  ok} 
to $\Fauction$. 
\end{enumerate}

Recall that at least one buyer or the seller is honest (henceforth 
denoted $i$). 
Given that the honest player $i$ chooses its random coins
$r_i$ at random, 
the joint coin toss result
$r = \left(\displaystyle\mathop{\oplus}_{j \in 
\mcal{I}} r_j\right) \oplus r_\platform$
is also a uniform 
random string. 
Therefore, it is equivalent for $\Fauction$ to just choose
$r$ at random in the above. 

\ignore{
switch honest buyers' NITC commitments to commitments of 0. 
Drs : depth of PoR extractor bounded by the timeout;
Drs': depth of RS reconstr;  
Du: depth of the function that computes utility;  
Drs + Drs' + Du 
less than the depth of the CCA adv of NITC
rely on IND-CCA against depth bounded
}

Summarizing the above, we have that 
the coalition's expected utility in the real world
cannot be negligibly more than its expected utility in the ideal world.

\subsection{Proofs of Utility Dominated Emulation: When the Platform is Honest}
\label{sec:proof-emu2}

\elaine{TODO: fix at least one buyer honest assumption here}

This is the easier case. 
The coalition $\mcal{C}$ consists of either one or more buyers,
or just the seller.
Consider the following sequence of hybrid experiments.
Throughout the proof below, we assume
that the coalition's algorithm is \ppt, and 
all the cryptographic building blocks only need
to be secure against \ppt adversaries. 
\elaine{TODO: double check}

\paragraph{Experiment $\Hyb_1$.}
$\Hyb_1$ is the same as the real-world protocol
except that when $\mcal{C}$ sends
${\sf NITC}$ commitments 
of the strategic bids 
$\{(j, c_j, \pi_j^{\rm com})\}_{j \in \mcal{J'}}$,
we ignore all the tuples where $\pi_j^{\rm com}$
does not verify or tuples with identity $0$ but not sent from the seller, 
perform duplicate suppression on the identities
and notify the suppressed identities as before,
but additionally, 
we call the ${\sf NITC}.{\sf FDec}$ function 
to force-open 
$\{(v_j, r_j)\}_{j \in \mcal{J}}$ where $\mcal{J}$ is the set of strategic
identities after duplicate suppression.

If all honest players accept at the end, 
we compute the coalition's utility
using 1) the honest buyers' true values 
$\{v_i\}_{i \in \mcal{H} \backslash \{0\}}$ 
where $\mcal{H}$ denotes the honest players (not including the platform), 
as well as 
honest randomness $\{v_i\}_{i \in \mcal{H}}$ and $r_\platform$;
and 2) the  
force-opened bids 
$\{v_j\}_{j \in \mcal{J} \backslash \{0\}}$ 
and forced-opened randomness
$\{r_j\}_{j \in \mcal{J}}$ 
corresponding
to strategic players. 

\begin{claim}
Suppose that the ${\sf NITC}$ scheme is computationally 
binding. Then, 
the coalition's utility in $\Hyb_1$ has negligible
statistical distance from the real world. 
\end{claim}
\begin{proof}
The only reasons that would cause the coalition's utility
to differ is if   
some strategic $j \in \mcal{J}$ successfully
opened its ${\sf NITC}$ commitment 
to a different tuple than the forced opened  
tuple. 
This happens with only negligible probability
as long as ${\sf NITC}$ is computationally sound.
\end{proof}

\elaine{TO FIX: 
for adv produced commitment, FDec must be successful.
which property is this? this is also needed in the
strategic platform case?}

\paragraph{Experiment $\Hyb_2$.}
In $\Hyb_2$, 
we stop executing the protocol after the coalition
submits their ${\sf NITC}$ commitments.   
If no honest identity has been suppressed, 
we use the approach of $\Hyb_1$ to compute 
the coalition's identity.
Otherwise, we treat the coalition's identity as $0$.

\begin{claim}
Suppose that ${\sf NITC}$ satisfies 
soundness of forced decryption. 
The coalition's utility in $\Hyb_1$ 
cannot be negligibly more than its utility in $\Hyb_2$.
\end{claim}
\begin{proof}
The only way for the coalition to gain
more utility 
in $\Hyb_1$ than $\Hyb_2$ is if 
the utility computed
through the forced-opened strategic inputs 
is negative, but in $\Hyb_1$, the  
coalition manages to make the execution abort
by not submitting an opening, and ${\sf FDec}$ fails
to produce a valid proof.  
Note that we implicitly assume that ${\sf FDec}$ always
outputs a message even if it does not output a valid proof.
By the 
soundness of forced decryption 
of ${\sf NITC}$, this happens
with negligible probability.
\end{proof}

\paragraph{Ideal auction and strategy.}
Note that $\Hyb_2$ can be equivalently viewed
as an ideal auction with the following coalition strategy. 
The coalition 
simply invokes the real-world
$\mcal{C}$'s algorithm, and wait for it to 
send the ${\sf NITC}$
commitments. It then verify the commitment proofs,
suppresses duplicate identities as well as any buyer that claims
the identity of $0$, and performs
force decryption as mentioned above. 
It sends the forced-opened strategic bids
to $\Fauction$. 
$\Fauction$ then tosses random coins, and computes
the outcome of the auction honestly. We then
determine the coalition's utility based on the outcome
and the true value of the buyers in $\mcal{C}$.

Summarizing the above, the coalition's expected utility in the real
world cannot be 
negligibly more than its utility in the ideal world.

\section{Ideal Auction: 2nd-Price with Reserve and Fixed Platform Fees}
\label{sec:ideal2nd}

\subsection{Ideal Auction}
We assume that the value domain is discrete like
in \Cref{sec:ascending}. 
However, in this section, the auction needs
to be parametrized by the security parameter $\lambda$
since we want to compile it to 
a real-world cryptographic protocol
using the compiler in \Cref{sec:robust}.
More formally ,
suppose the value domain $\inDomSecu$ is finite 
and $0 \in \inDomSecu$ for
every $\lambda$. 
Suppose $\inDomSecu$ 
consists of the values
$0 = \theta^1_\lambda < \theta^2_\lambda 
< \ldots 
< 
\theta^{T(\lambda)}_\lambda$.
Just like in \Cref{sec:ascending}, 
given some distribution $\mcal{D}_\lambda$ 
over $\inDomSecu$
with the cumulative distribution function $F_\lambda$ and the probability
density function $f_\lambda$, 
the virtual value 
$\phi(\theta^i_\lambda) := \theta^i_\lambda - \frac{1 - 
F_\lambda(\theta^i_\lambda)}{f_\lambda(\theta^i_\lambda)}\left(
\theta^{i+1}_\lambda - \theta^i_\lambda\right)$, and 
$\phi(\theta^T_\lambda) := \theta^T_\lambda$. \elaine{double check}
We say that the distribution $\mcal{D}_\lambda$
is {\it regular}, iff
$\phi(\cdot)$
is a strictly increasing function. 
Let the reserve price $r_\lambda$ be the smallest
$\theta^i_\lambda$ such that  $\phi(\theta^i_\lambda) \geq 0$.

We describe the ideal auction below --- recall
that to specify an ideal auction, all we need is to specify
the allocation, payment, and revenue rules: 
\begin{mdframed}
\begin{center}
{\bf Ideal auction: 2nd price with reserve and fixed platform fees}
\end{center}

\begin{itemize}[leftmargin=6mm,itemsep=1pt]
\item 
Discard any bid that is less than the reserve $r_\lambda$. 
Rank the remaining bids in order from large to small, and 
break ties randomly. 
Let $b_1 \geq b_2 \geq \ldots \geq b_n$ be the resulting vector of bids,
and rename corresponding buyers as $1, 2, \ldots, n$. 
\item 
Buyers $1, 2, \ldots, k'$ are allocated an item, where 
$k' = \min(k, n)$ --- henceforth these are called 
the confirmed buyers and confirmed bids. 
\item 
Define $b_{k' + 1} := r_\lambda$ if $n \leq k$.
Define the payment as follows 
where $A$ be the total number of bids that are equal to $b_{k' + 1}$, 
and $\alpha$ is the number of confirmed bids 
that are equal to $b_{k' + 1}$: 
\begin{itemize}[leftmargin=6mm,itemsep=1pt]
\item Any confirmed buyer 
whose bid is equal to $b_{k'+1}$ pays $b_{k' + 1}$;
\item 
Any confirmed buyer whose 
bid is strictly greater than $b_{k'+1}$ pays
$b_{k' + 1} \cdot q+ 
{\sf NextTick}(b_{k' + 1}) \cdot (1 - q)$
where $q = \frac{\alpha + 1}{A + 1}$, 
and ${\sf NextTick}(v)$ means
the smallest  
value larger than $v$ in $\inDomSecu$. 
\end{itemize}
\item 
The platform gets nothing (or a fixed fee independent of the auction),
and the seller gets all the payment. 
\end{itemize}
\end{mdframed}

\begin{theorem}[Ideal auction: second price with reserve and fixed platform fees]
Suppose that the above second-price auction with reserve
and fixed platform fees is executed
under the ideal model described in \Cref{sec:idealmodel}.
Then, the resulting ideal auction 
has the following properties: 
\begin{itemize}[leftmargin=6mm,itemsep=1pt]
\item It satisfies information-theoretic
bIC, pIC, 1-pbIC regardless of $\mcal{D}_\lambda$;
\item 
Suppose that the distribution $\mcal{D}_\lambda$
is regular for every $\lambda$, 
then the auction additionally satisfies;
Bayesian psIC, and Bayesian sIC. 
\item 
Suppose that the distribution $\mcal{D}_\lambda$
is regular for every $\lambda$, 
then the auction 
is the revenue maximizing auction among all bIC auctions;
\item 
The auction's rules can be computed 
depth $O(\log^2(L))$ where $L$ is the total input length.
\end{itemize}
\label{thm:ideal2ndprice}
\end{theorem}

The proof of \Cref{thm:ideal2ndprice}
is presented in \Cref{sec:proof-2nd}.

\subsection{Applying the Ideal-to-Real Compiler}

The above \Cref{thm:ideal2ndprice}, combined with 
the ideal-to-real compiler
described 
in \Cref{thm:robustcompiler} of \Cref{sec:robustcompiler}, 
as well as  
\Cref{thm:idealreal}, 
immediately leads to the following corollary:

\begin{corollary}[Achieving bIC, pIC, pbIC, Bayesian psIC and Bayesian sIC in $O(1)$ rounds]
Suppose that the strong repeated squaring assumption
and the hardness of the DDH and DCR problems in suitable groups
hold against quasi-polynomially sized adversaries.  
Then, there exists an auction that satisfies
computational bIC, pIC, pbIC regardless of the value distribution $\mcal{D}$; 
further, assuming that the value distribution $\mcal{D}$ is regular, then 
it additionally satisfies
computational
Bayesian psIC, and Bayesian sIC. 
Moreover, 
the auction satisfies the following properties:
\begin{itemize}[leftmargin=6mm,itemsep=1pt]
\item 
It completes in constant number of rounds
with $\widetilde{O}_\lambda(n)$ total communication and computation,  
and each buyer and the seller's computation and communication
is bounded by $\widetilde{O}_\lambda(1)$. 
\item 
When all buyers' values are sampled independently from a regular
distribution $\mcal{D}$, 
the auction maximizes the expected
revenue among all bIC auctions. 
\end{itemize} 
\end{corollary}

\subsection{Proof of \Cref{thm:ideal2ndprice}}
\label{sec:proof-2nd}

\paragraph{Additional preliminaries.}
We will rely on the following technical
lemma proven by Elkind~\cite{discrete-auction}.  

\ignore{
\begin{claim}[Restatement of Theorem 1 in \cite{discrete-auction}]
\label{lem:payment-bound-virtual-value-2nd}
Suppose each buyer's true value is i.i.d.~sampled from a discrete distribution $\mcal{D}$.
Then, for any auction $[\bfx, \bfp, \mu_{\seller}]$ that satisfies information-theoretic Bayesian bIC,
and for any buyer $i$,
the expected payment of buyer $i$ must satisfy
\begin{equation}
\label{eq:payment-upper-bound-2nd}
\mathop{\E}_{\bfb \getr \mcal{D}^{n}} [p_i(\bfb)]
\leq
\mathop{\E}_{\bfb \getr \mcal{D}^{n}} \left[
x_i(\bfb) \cdot \phi(b_i)
\right].
\end{equation}
The statement holds no matter $\mcal{D}$ is regular or not.
\end{claim}
}

\begin{lemma}[Technical lemma from \cite{discrete-auction}]
\label{lem:revopt}
Suppose that the value domain $\inDom$
is a finite set consisting of $0 = \theta^0 < \theta^1 < \ldots < \theta^T$.   
Suppose each buyer's true value is sampled independently 
from some distribution $\mcal{D}$ over $\inDom$.
Given any $\bfb_{-i} \in \inDom^*$, 
any non-decreasing allocation rule $\bfx$,
suppose that the payment rule is defined as follows: 
\begin{equation}
    \label{eqn:optpay}
    p_i(\theta^\tau, \bfb_{-i}) = \theta^\tau \cdot x_i(\theta^\tau, \bfb_{-i}) - 
\sum_{j = 1}^\tau \left(\theta^j - \theta^{j-1}\right) \cdot x_i\left(\theta^{j-1}, \bfb_{-i}\right),
\end{equation}
Then, the resulting auction is incentive compatible for any individual
buyer that is restricted to input replacement strategies. Moreover, 
when all buyers' bids are sampled independently from $\mcal{D}$, 
it is the revenue-maximizing auction 
subject to bIC 
and the allocation rule ${\bf x}$. 
Specifically, given an arbitrary $\bfb_{-i}$, the expected
payment from any  
buyer $i$ is the following:  
\ignore{
\[
\mathop{\E}_{\bfb \getr \mcal{D}^{n}} [p_i(\bfb)]
= \mathop{\E}_{\bfb \getr \mcal{D}^{n}} \left[
x_i(\bfb) \cdot \phi(b_i)
\right]
\]
}
\[
    \mathop{\E}_{b_i \getr \mcal{D}} [p_i(b_i,\bfb_{-i})]
    =
    \mathop{\E}_{b_i \getr \mcal{D}} \left[
    x_i(b_i,\bfb_{-i}) \cdot \phi(b_i)\right].
\]
\end{lemma}

We now prove the claimed properties one by one. 

\paragraph{bIC.}
It is not hard to check that the allocation rule is monotone, and the
auction's payment rule agrees with \Cref{eqn:optpay} specified
in \Cref{lem:revopt}. Thus, 
by \Cref{lem:revopt}, any input replacement strategy cannot benefit
an individual buyer.  
It remains to prove that 
injecting fake bids does not benefit the buyer.  
Let $b^* = r_\lambda$ if $n \leq k$ or $b^* = b_{k+1}$ otherwise. 
Any fake bid that is less than $b^*$ does not affect
the outcome of the auction. 
Any fake bid that is at least $b^*$ can never decrease
the payment if the buyer gets an item.  
Therefore, injecting fake bids does not help the buyer. 

\paragraph{pIC.}
Since the platform's revenue is always zero,  
pIC is trivially satisfied.

\paragraph{1-pbIC.}
Directly implied by bIC and the fact that the platform's revenue is always $0$. 

\paragraph{Bayesian psIC.}
The only possible strategies for a platform-seller coalition
in the ideal world is 1) to abort the auction causing every one's utility
to be 0, and 2) to inject fake bids after learning the number of honest bids. 
Clearly, aborting will not benefit the coalition. 
Therefore, it suffices to show that injecting fake
bids does not increase the coalition's expected utility. 

Injecting any bid that is less than $r_\lambda$ 
makes no difference to the auction. 
Therefore, let $\bfb^*$ 
be the fake bids that are at least $r_\lambda$ 
injected by the coalition. 
Let $n$ be the number of honest buyers. 
By \Cref{lem:revopt}, 
the coalition's expected revenue from a fixed honest buyer $i$ is  
\[\mathop{\E}_{\bfb \getr \mcal{D}^n} [p_i(\bfb, \bfb^*)]
    =
    \mathop{\E}_{\bfb \getr \mcal{D}^n} \left[
    x_i(\bfb, \bfb^*) \cdot \phi(b_i)\right]
\leq 
    \mathop{\E}_{\bfb \getr \mcal{D}^n} \left[x_i(\bfb) \cdot \phi(b_i)\right]
\]
where the
expression on the 
right-hand side 
   $\mathop{\E}_{\bfb \getr \mcal{D}^n} \left[x_i(\bfb) \cdot \phi(b_i)\right]$
corresponds to the expected revenue from buyer $i$
when the coalition behaves honestly. 
Therefore, injecting fake bids does not increase
the coalition's expected revenue.

\paragraph{Bayesian sIC.}
A strategic seller's only possible strategy is to inject
fake bids. 
Recall that the platform gets nothing and all revenue
goes to the seller. 
Therefore, due to the same argument as the above proof of Bayesian psIC, injecting
fake bids does not increase the seller's revenue.

\paragraph{Revenue optimality.}
Let $n$ be the number of buyers. 
By \Cref{lem:revopt}, for any bIC auction with a monotonic
allocation rules ${\bf x}$,  
the maximum possible expected revenue 
is 
$\sum_{i \in [n]}
\left(\mathop{\E}_{\bfb \getr \mcal{D}^n} \left[x_i(\bfb) \cdot \phi(b_i)\right]\right)
= 
\mathop{\E}_{\bfb \getr \mcal{D}^n} 
\left[
\sum_{i \in [n]} x_i(\bfb) \cdot \phi(b_i)
\right]
$. 
It suffices to argue that our choice
of the allocation rule 
$\bfx$ maximizes
$\mathop{\E}_{\bfb \getr \mcal{D}^n} 
\left[
\sum_{i \in [n]} x_i(\bfb) \cdot \phi(b_i)
\right]$.  
For any $\bfb$, we argue
that 
$\sum_{i \in [n]} x_i(\bfb) \cdot \phi(b_i)$ is maximized
under our choice of $\bfx$.  
This is because our $\bfx$ allocates 
one item to those with the largest virtual value $\phi(b_i)$, 
subject to allocating at most $k$ items;  
moreover, it only allocates an item to those with non-negative virtual values.

\paragraph{Small depth.}
The auction's rules can be computed by a sorting network
consisting of comparators. Each comparator can be implemented
in logarithmic depth in its input  
length, and the sorting network
needs only logarithmic layers~\cite{aks}. 
Therefore, it follows that the auction's rules
can be computed in  
$O(\log^2 (L))$ depth.




\section*{Acknowledgments}
Elaine Shi would like to thank Dakshita Khurana, 
Binyi Chen, and 
Sacha Servan-Schreiber for helpful technical discussions. 
This work is in part supported
by NSF awards under the grant numbers  
2212746 and 2044679, a Packard Fellowship,
a generous gift from the late  
Nikolai Mushegian, an an ACE center grant from the Algorand Foundation.

\newpage
\bibliographystyle{alpha}
\bibliography{refs,crypto,gametheory}
\clearpage

\appendix

\section{Preliminaries: Cryptographic Building Blocks}
\label{sec:bldgblock}
In this section, we formally define the cryptographic primitives used in \ref{sec:robustcompiler}.
Recall that we use $\secpar$ to denote the security parameter.
$\{X_{\secu}\}_{\secu\in\N}\equiv_C \{Y_{\secu}\}_{\secu\in\N}$ implies that the two distribution ensembles $\{X_{\secu}\}_{\secu\in\N}$ and $\{Y_{\secu}\}_{\secu\in\N}$ are computationally indistinguishable.
We use PPT as an abbreviation for probabilistic polynomial time.

\ignore{
\subsection{Collision-Resistant Hash Function}
We say that a family of hash functions $\mcal{H} = \{h_s:\{0,1\}^*\rightarrow\{0,1\}^{\secu}\}$ is a coalition-resistant hash function family iff
\begin{itemize}
    \item There is a sampling algorithm $s\leftarrow \Gen(1^{\secu})$ that samples a random function $h_s$ from the family $\mcal{H}$;
    \item There exists a polynomial $q$ such that for any $h_s\in\mcal{H}$ and for any $x\in\{0,1\}^*$, the running time of $h_s(x)$ is bounded by $q(|x|, \secu)$;
    \item {\bf Collision-resistance:} For any non-uniform \ppt $\A$, there exists a negligible function $\negl(\cdot)$ such that for any $\secu\in\N$, the following holds:
    \[\Pr[s\leftarrow \Gen(1^{\secu}), (x,y)\leftarrow\A(1^{\secu}, s): x\neq y, h_s(x) = h_s(y) ]\leq \negl(\secu).\]
\end{itemize}
}

\subsection{Vector Commitment}


A vector commitment scheme is a tuple of algorithms
$({\sf Gen}, {\sf Digest}, {\sf Open}, {\sf Vf})$:
\begin{itemize}[leftmargin=6mm,itemsep=1pt]
\item 
$\crs \leftarrow {\sf Gen}(1^\lambda)$: 
on input the security parameter $1^\lambda$, 
output
a common reference string $\crs$; 
\item 
$({\sf digest}, {\sf aux}) \leftarrow {\sf Digest}(\crs, m)$:
given $\crs$ and a message $m$, output
a digest ${\sf digest}$ 
and some auxiliary information ${\sf aux}$ --- we may assume
that ${\sf aux}$ contains the message length $\ell := |m|$;
\item 
${\sf Open}(\crs, {\sf aux}, Q)$: on input $\crs$, 
auxiliary information 
${\sf aux}$ (assumed to contain the message length $\ell$), 
and  
a query set $Q \subseteq [\ell]$,
output an opening proof $\pi$ that $m[Q]$ is a restriction
of $m$ to the indices $Q$;
\item 
$(0, 1) \leftarrow {\sf Vf}(\crs, 
\ell, 
{\sf digest}, Q, {\sf ans}, \pi)$:
on input $\crs$, 
message length $\ell$, 
$\digest$, a query set $Q \subseteq [\ell]$,  
a purported answer ${\sf ans}$, and a proof $\pi$, 
outputs either $0$ or $1$ indicating  
reject or accept. 
\end{itemize}

\elaine{note that the verification needs the length ell}

\paragraph{Correctness.}
Correctness requires that for any $\lambda \in \mathbb{N}$, 
any $\ell$, 
any message $m\in \{0, 1\}^\ell$, 
any $Q \subseteq [\ell]$, 
the following holds with probability 1:
let 
$\crs \leftarrow {\sf Gen}(1^\lambda)$, 
$(\digest, {\sf aux}) \leftarrow {\sf Digest}(\crs, m)$, 
$\pi \leftarrow {\sf Open}(\crs, {\sf aux}, Q)$, 
then it holds
that 
${\sf Vf}(\crs, \ell, \digest, Q, m[Q], \pi) = 1$.

\paragraph{Collision resistance.}
We say that a vector commitment scheme satisfies
collision resistance against size-$W(\cdot)$ adversaries,
iff for any non-uniform probabilistic
machine $\algA(1^\lambda, *)$ whose total work 
is bounded by $W(\lambda)$, 
there exists a negligible function $\negl(\cdot)$ 
such that for every $\lambda \in \mathbb{N}$, 
the probability that the following experiment
outputs $1$ is at most $\negl(\lambda)$: 
\begin{itemize}[leftmargin=6mm,itemsep=1pt]
\item 
$\crs \leftarrow \Gen(1^\lambda)$; 
\item 
$(\ell, \digest, {\sf ans}, {\sf ans}', Q, Q', \pi, \pi') \leftarrow 
\algA(1^\lambda, \crs)$ where $Q, Q' \subseteq [\ell]$; 
\item 
Output $1$ if 
${\sf Vf}(\crs, \ell, \digest, Q, {\sf ans}, \pi)
= {\sf Vf}(\crs, \ell, \digest, Q', {\sf ans}', \pi')
$; however, 
there is some $i \in Q \cap Q'$ such that 
${\sf ans}$ and ${\sf ans}'$ contain
different answers for the index $i$.  
\end{itemize}

Merkle~\cite{merkle} showed how to build such a vector
commitment scheme secure against polynomially sized adversaries
(or quasi-polynomially 
sized adversaries resp.) 
 assuming the existence
of a collision resistant hash family secure
against polynomially sized adversaries (or quasi-polynomially sized
adversaries).

\elaine{I removed NIZK and comm, i will rewrite non-robust sec}

\subsection{Argument of Knowledge}
\label{sec:argument-of-knowledge}

\elaine{TODO: change to quasipoly extractor!!}
Let $\mcal{L}_\lambda$ denote an NP language paramtrized by $\lambda$. 
An argument of knowledge system 
consists of the following algorithms:
\begin{itemize}[leftmargin=6mm,itemsep=1pt]
\item 
$\crs \leftarrow 
\Gen(1^\lambda, \mcal{L}_\lambda)$: 
a randomized algorithm that samples and 
outputs a common reference string $\crs$.
\item 
$b \leftarrow \langle P(\crs, \stmt, w),  V(\crs, \stmt) \rangle$: 
a pair of 
randomized  interactive 
algorithms where 
where $P$ denotes the prover and $V$ denotes the verifier. 
Both $P$ and $V$ receive
the common reference string $\crs$ and the statement $\stmt$,
and $P$ additionally receives 
a witness $w$ as input. The two then  
engage in an interactive protocol at the end of which
the verifier outputs $b \in \{0, 1\}$ indicating
reject or accept. 
\end{itemize}
\ignore{
For any $\secu \in \N$, let $\P^*_{\secu, \aux, r_{\P}}$ denote the prover $\P_\secu^*$ with auxiliary input $\aux$ and randomness $r_{\P}$,
and let $\V_{r_{\V}}$ denote the verifier $\V$ using randomness $r_{\V}$.
Given any statement $\stmt$, auxiliary information $\aux$, and randomness $r_{\P}, r_{\V}$,
we use $\langle\P^*_{\secu, \aux, r_{\P}}(1^{\secu}, \stmt), \V_{r_{\V}}^*(1^{\secu}, \stmt)\rangle \in \{0,1\}$ to denote whether $\V^*$ rejects or accepts the execution between a prover $\P^*_{\secu, \aux, r_{\P}}(1^{\secu}, \stmt)$ and a verifier $\V_{r_{\V}}^*(1^{\secu}, \stmt)$.
Let $\ell(\secu)$ be the upper bound on the number of random bits.
In this work, we consider an argument of knowledge scheme that satisfies the following properties.
}

\paragraph{Perfect completeness.}
We say an argument of knowledge scheme satisfies \emph{perfect completeness}
if for any $\secu \in \N$, 
for any NP language $\mcal{L}_\lambda$ whose corresponding NP 
relation is denoted $\mcal{R}_\lambda$, 
for any 
statement $(\stmt, w) \in \mcal{R}_\lambda$, 
the following holds with probability $1$:
let $\crs \leftarrow 
\Gen(1^\lambda, \mcal{L}_\lambda)$, 
then $\langle P(\crs, \stmt, w),  V(\crs, \stmt) \rangle = 1$. 


\paragraph{Adaptive knowledge soundness.}
We say that an argument of knowledge 
scheme satisfies adaptive knowledge soundness,
iff for any  NP language $\mcal{L}_\lambda$ with the corresponding NP
relation $\mcal{R}_\lambda$, 
 there exists a probabilistic quasi-polynomial-time extractor 
$\mcal{E}$, a quasi-polynomially 
bounded function $q(\cdot)$, 
and some negligible function $\negl(\cdot)$, 
such that for every $\lambda \in \mathbb{N}$, 
any auxiliary distribution $\mcal{D}$, 
any deterministic $P^*$ that runs in time at most $t$, 
\[
\Pr\left[
(\stmt, w) \notin \mcal{R}_\lambda
\wedge b = 1 
\left| 
\begin{array}{l}
\crs \leftarrow \Gen(1^\lambda, \mcal{L}_\lambda) \\
\eta \leftarrow \mcal{D}\\
(\stmt, {\sf aux} ) \leftarrow P^*(\crs, \eta)\\ 
b \overset{{\it tr}}{\leftarrow} 
\langle P^*({\sf aux}), V(\crs, \stmt) \rangle\\
w \leftarrow \mcal{E}^{P^*({\sf aux})}(\crs, \stmt, {\it tr})
\end{array}
\right.
\right]
\leq \negl(\lambda) 
\]
and moreover, $\mcal{E}$'s running time
is upper bounded by $q(\lambda, t)$.  
In the above, 
$b \overset{{\it tr}}{\leftarrow} 
\langle P^*({\sf aux}), V(\crs, \stmt) \rangle$ means
that ${\it tr}$ is the transcript
of the execution 
$\langle P^*({\sf aux}), V(\crs, \stmt) \rangle$
which includes the $\crs$ and messages exchanged   
between $P^*$ and $V$, 
and the notation $\mcal{E}^{P^*}$
means that $\mcal{E}$ has blackbox access to each next-message
function of $P^*$.  

\paragraph{Succinctness.}
An argument of knowledge system is said
to be succinct iff the verifier's runtime is upper bounded
by $\poly(\lambda, \log |\mcal{R}_\lambda|)$
where $|\mcal{R}_\lambda|$ denotes the size of the circuit
that checks the NP relation. 

Kilian~\cite{kilian} showed  
how to construct a succinct argument of knowledge scheme 
that satisfies the aforementioned 
adaptive knowledge soundness notion as well as succinctness requirement, 
assuming
the existence of a collision resistant 
hash family secure against quasi-polynomially sized
adversaries.

\ignore{
\paragraph{Argument of Knowledge for NP.}
There exists a probabilistic oracle machine $\mcal{E}$ and a polynomial $q$ such that 
for every non-uniform probabilistic polynomial-time prover $\P^* = \{\P^*_\secu\}_\secu$,
there exists a negligible function $\negl(\cdot)$ such that for any $\secu\in\N$,
for any statement $\stmt$,
for any auxiliary input $\aux$, and for any randomness $r_{\P} \in \{0,1\}^{\ell(\secu)}$,
the following hold.
\begin{itemize}
\item 
The expected running time of $\mcal{E}^{\P^*_{\secu, \aux, r_{\P}}, \V_{r_{\V}}}(1^\secu, \stmt)$ is bounded by $q(\secu)$, 
where the expectation is over $r_{\V} \in \{0,1\}^{\ell(\secu)}$ and the random coins of $\mcal{E}$.
\hao{We can simply saying $\mcal{E}$ is expected PPT algorithm if we don't need to specify the probability space explicitly.}
\item 
It holds that
\[
\Pr\left[
\begin{array}{c}
   r_{\V} \leftarrow \{0,1\}^{\ell(\secu)},   \\
    w \leftarrow \mcal{E}^{\P^*_{\secu, \aux, r_{\P}}, \V_{r_{\V}}}(1^\secu, \stmt)
\end{array}
:
\begin{array}{c}
    \langle\P^*_{\secu, \aux, r_{\P}}(1^\secu, \stmt), \V_{r_{\V}}(1^\secu, \stmt)\rangle = 1 \\
    \wedge \ (\stmt, w)\notin\mcal{R} \\
\end{array}
\right] \leq \negl(\secu).
\]
\end{itemize}
}

\subsection{Publicly Verifiable Non-Malleable Timed Commitments}
\elaine{TODO: need security against quasipoly}

We use the definitions from Chvojka and Jager~\cite{CJ-nitc}.
A publicly verifiable, 
non-interactive timed commitment (${\sf NITC}$) with message space
$\mcal{M}$ is a tuple
of algorithms  
$(\Gen$, ${\sf Com}$, ${\sf ComVf}$, ${\sf DecVf}$, ${\sf FDec}, {\sf FDecVf})$
with the following syntax:
\begin{itemize}[leftmargin=5mm,itemsep=1pt]
\item 
$\crs \leftarrow \Gen(1^\lambda, T)$: a probabilistic
algorithm 
that takes as input the security parameter $1^\lambda$, 
a difficulty parameter $T$, and output a common reference string $\crs$. 
\elaine{do we need a sec key?}
\item 
$({\sf cm}, \pi_{\rm com}, \pi_{\rm dec}) \leftarrow {\sf Com}(\crs, m)$:
 a probabilistic algorithm that takes
as input a common reference string $\crs$ and a message
$m \in \mcal{M}$, and outputs
 a commitment ${\sf cm}$ along with proofs 
$\pi_{\rm com}$
and 
$\pi_{\rm dec}$.
\item 
$0/1  \leftarrow {\sf ComVf}(\crs, {\sf cm}, \pi_{\rm com})$:
a deterministic algorithm that takes as input
a common reference string $\crs$, a commitment ${\sf cm}$
and a proof $\pi_{\rm com}$, and outputs
a bit indicating whether to reject or accept.
\item 
$0/1  \leftarrow {\sf DecVf}(\crs, {\sf cm}, m, \pi_{\rm dec})$:
a deterministic algorithm that takes as input
a common reference string $\crs$, a commitment ${\sf cm}$, 
a message $m$, and a proof $\pi_{\rm dec}$, and outputs
a bit indicating whether to reject or accept.
\item 
$(m, \pi_{\rm fdec}) \leftarrow {\sf FDec}(\crs, {\sf cm}, \pi_{\rm com})$: 
a deterministic algorithm that takes as input
a common reference string $\crs$, a commitment ${\sf cm}$, 
and outputs $m \in \mcal{M} \cup \{\bot\}$ in time
at most $T \cdot \poly(\lambda)$. 
\item 
$0/1  \leftarrow {\sf FDecVf}(\crs, {\sf cm}, m, \pi_{\rm fdec})$:
a deterministic algorithm that takes as input
a common reference string $\crs$, a commitment ${\sf cm}$, 
a message $m$, and a proof $\pi_{\rm fdec}$, and outputs
a bit indicating whether to reject or accept.
\end{itemize} 

Note that when the ${\sf NITC}$ is used in an actual protocol,
if the committer is adversarial, it may not provide
$\pi_{\rm dec}$ in the opening phase.
This is why we need ${\sf FDec}$
to introduce a forced opening proof 
$\pi_{\rm fdec}$
and have a separate  
${\sf FDecVf}$
algorithm for the public verifiability of forced openings. 

\paragraph{Soundness of forced decryption.}
We say that an ${\sf NITC}$ scheme satisfies 
soundness of forced decryption, iff 
for any non-uniform \ppt adversary $\algA$, 
for any polynomial function $T(\cdot)$,  
\elaine{do we need quasipoly here}  
there exists a negligible function $\negl(\cdot)$ such that 
the probability that the following experiment
outputs $1$ is at most $\negl(\lambda)$:
\begin{itemize}[leftmargin=6mm,itemsep=1pt]
\item 
$\crs \leftarrow \Gen(1^\lambda, T)$;
\item 
$({\sf cm}, \pi_{\rm com}) \leftarrow \algA(1^\lambda, \crs)$; 
\item 
$(m, \pi_{\rm fdec}) 
\leftarrow {\sf FDec}(\crs, {\sf cm}, \pi_{\rm com})$; 
\item 
${\sf ComVf}(\crs, {\sf cm}, \pi_{\rm com} = 1$
but 
${\sf FDecVf}(\crs, {\sf cm}, m, \pi_{\rm fdec}) \neq 1$. 
\end{itemize} 
In other words, except with negligible probability, 
if an adversarially produced commitment has a valid commitment
proof, then ${\sf FDec}$ should be able to extract a message
along with a valid forced-decryption proof.

\paragraph{Correctness.}
We say that 
an ${\sf NITC}$ scheme is correct if for all $\lambda, T \in \mathbb{N}$, 
and all $m \in \mcal{M}$, the following holds
with probability $1$: 
let 
$\crs \leftarrow \Gen(1^\lambda, T)$, 
$({\sf cm}, \pi_{\rm com}, \pi_{\rm dec}) \leftarrow {\sf Com}(\crs, m)$, 
$(m', \pi_{\rm fdec}) \leftarrow {\sf FDec}(\crs, {\sf cm})$, 
then it holds that 
${\sf ComVf}(\crs, {\sf cm}, \pi_{\rm com}) = 1$, 
${\sf DecVf}(\crs, {\sf cm}, m, \pi_{\rm dec}) = 1$, 
$m' = m$, 
and 
${\sf FDecVf}(\crs, {\sf cm}, m, \pi_{\rm dec}) = 1$.

\begin{definition}[IND-CCA against depth-bounded adversaries]
Given some function $W(\cdot)$ and some $\eps \in (0, 1)$, 
we say that an ${\sf NITC}$ scheme satisfies $(W(\cdot), \eps)$-IND-CCA (short
for indistinguishability under chosen-ciphertext-attack) security 
iff there exists
a polynomial $T_\emptyset(\cdot)$, such that for all
polynomials $T(\cdot) \geq T_\emptyset(\cdot)$, 
and every non-uniform adversary 
\elaine{TODO: change to quasipolynomial time adv --- look at their paper
for params}
$\algA = (\algA_1, \algA_2)$ 
where $\algA(1^\lambda, \cdot)$ is bounded by $W(\lambda)$ 
in total work
and  
$\algA_2(1^\lambda, \cdot)$
is bounded $T^\epsilon(\lambda)$ in depth, 
\elaine{explain how to measure depth for an oracle-adv}
there exists a negligible function $\negl(\cdot)$ 
such that for all $\lambda \in \mathbb{N}$, it holds that
$|
\Pr[{\sf Expt}^{\algA}_0(1^\lambda) = 1]
- \Pr[{\sf Expt}^{\algA}_1(1^\lambda) = 1]
\leq \negl(\lambda)
$
where ${\sf Expt}^{\algA}_b(1^\lambda)$
is defined as follows:
\begin{itemize}[leftmargin=8mm,itemsep=1pt]
\item
$\crs \leftarrow \Gen(1^\lambda, T(\lambda))$;
\item 
$(m_0, m_1, {\sf st}) 
\leftarrow \algA_{1, \lambda}^{{\sf Dec}(\cdot, \cdot)}(\crs)$;
\item 
$({\sf cm}^*, \pi^*_{\rm com}, \pi^*_{\rm dec})
 \leftarrow {\sf Com}(\crs, m_b)$;
\item 
$b' \leftarrow \algA_2^{{\sf Dec}(\cdot, \cdot)}({\sf cm}^*, \pi^*_{\rm com}, {\sf st})$;
\item 
output $b'$.
\end{itemize}
where 
the oracle ${\sf Dec}({\sf cm}, \pi_{\rm com})$
returns the result of ${\sf FDec}(\crs, {\sf cm})$
if ${\sf ComVf}(\crs, {\sf cm}, \pi_{\rm com}) = 1$;
otherwise it returns $\bot$;
and further, 
$\algA$ must respect the following restrictions:  
$|m_0| = |m_1|$, and  
moreover, it is not allowed to query the oracle ${\sf Dec}(\cdot, \cdot)$
on $({\sf cm}^*, \pi^*_{\rm com})$. 
\label{defn:cca}
\end{definition}

\begin{definition}[Computationally binding]
We say that an ${\sf NITC}$ scheme is computationally binding 
for size-$W(\cdot)$ 
adversaries, iff 
for any non-uniform \ppt adversary $\algA(1^\lambda, \cdot)$
whose total work is bounded  by $W(\lambda)$, there exists a negligible
function $\negl(\cdot)$ such that for every $\lambda$, 
the probability that  
the following experiment outputs $1$ is at most $\negl(\lambda)$:
\begin{itemize}[leftmargin=8mm,itemsep=1pt]
\item 
$\crs \leftarrow \Gen(1^\lambda, T(\lambda))$;
\item 
$(m, {\sf cm}, \pi_{\rm com}, \pi) 
\leftarrow \algA(\crs)$; 
\item 
output $1$ iff 
all of the following hold:
\begin{itemize}[leftmargin=6mm,itemsep=1pt]
\item 
${\sf ComVf}(\crs, {\sf cm}, \pi_{\rm com}) = 1$, 
\item 
${\sf DecVf}(\crs, {\sf cm}, m, \pi) =1$, 
or ${\sf FDecVf}(\crs, {\sf cm}, m, \pi) =1$, 
\item 
but the message output by ${\sf FDec}(\crs, {\sf cm})$
is not equal to $m$.
\end{itemize}
\ignore{
but one of the following happens:
\begin{itemize}[leftmargin=5mm,itemsep=1pt]
\item 
$m \neq m'$ and $
{\sf DecVf}(\crs, {\sf cm}, m, \pi_{\rm dec}) 
= {\sf DecVf}(\crs, {\sf cm}, m', \pi'_{\rm dec}) 
= 1$; or 
\item
${\sf DecVf}(\crs, {\sf cm}, m, \pi_{\rm dec}) =1$
and ${\sf FDec}(\crs, {\sf cm}) \neq m$. 
\end{itemize}
}
\end{itemize}
\end{definition}
\elaine{note: i simplified the defn of the original paper}

Chvojka and Jager~\cite{CJ-nitc}
constructed an ${\sf NITC}$ scheme 
satisfying the aforementioned requirements, assuming 
that the following assumptions hold for 
quasi-polynomially sized adversaries: 
the strong sequential squaring assumption, the 
Decisional Diffie-Hellman (DDH) and 
the Decisional Composite Residuosity (DCR) assumptions
in suitable groups,  
and the existence of a simulation sound extractable
Non-Interactive Zero Knowledge (NIZK) system (which is implied
by 
the hardness of DDH against quasi-polynomially sized adversaries). 
\elaine{TODO: double check}

\section{Deferred Proofs for the Ascending Auction}
\label{sec:ascend-proof}

We now prove \Cref{thm:ascend}. 

\begin{proof}
We prove each property individually. 
	
\paragraph{bIC.}
Fix an arbitrary buyer $i$ with true value $v_i$,
and let $\tau$ be the number that the honest platform posts on the blockchain in the honest execution.
It is easy to see that when the buyer $i$ underbids (sending $\bot$ when $v_i > \theta_{\tau}$)
or overbids (keep sending $\ok$ even when $v_i < \theta_{\tau}$),
it cannot increase its utility compared to behaving honestly.
Moreover, injecting fake bids either does not affect or increases the payment a winning buyer needs to pay.
Thus, behaving honestly is the best response for a buyer.
	
\paragraph{pIC.}
Since $\tau$ is posted on the blockchain,
the platform only gets $\theta_{\tau}$ from each winner
and has to pay $\theta_{\tau}$ to the seller for each item sold in any safe execution trace.
Because $\widetilde{k}$ must be at least the number of buyers 
who win items in a safe execution trace,
the platform's utility is at most $0$.

\paragraph{$1$-pbIC.}
As we argued for pIC,
the coalition cannot extract any profit from buyers outside the coalition,
so the only way the coalition can increase their utility is 
to lower the price that the colluding buyer $i$ has to pay,
or to bias the random tie-breaking.
However, if the coalition lowers the price (i.e.~stop ascending the price earlier), 
there are more than $k$ buyers who should win items.
In this case, either buyer $i$ has to give up winning an item which leads to zero utility, 
or the platform has to sell items to more than $k$ buyers where the execution is not safe.
On the other hand, if the colluding buyer is chosen from $\mcal{I}$ at the final round $\tau$,
it must be that buyer $i$'s utility is zero no matter it wins an item or not.
Thus, the auction satisfies $1$-pbIC.

\paragraph{Bayesian psIC.}
Let $\tau \in \{\tau_0,\dots, T\}$ denote the number that the platform would post on the blockchain in the honest execution.
Let $\mcal{R}_v$ be the set of buyers whose true values are at least $v$.
If the platform stops the auction early and posts some $\tau^*<\tau$ on the blockchain, then by definition, it must be that $|\mcal{R}_{\theta_{\tau^*}}| > k$. 
The resulting execution traces cannot be safe. 
Therefore, we focus on the case where $\tau^* \geq \tau$.

	
	
If the platform-seller coalition follows the protocol, the revenue they gain from each winner is $\theta_\tau$.
If they keep ascending the price to $\theta_{\tau^*} > \theta_\tau$, 
the coalition can only earn revenue from a buyer if the buyer's true value is still larger than or equal to $\theta_{\tau^*}$.
Thus, the expected revenue they gain from each buyer $i$ in $\mcal{R}_{\theta_{\tau^*}}$ is $\theta_{\tau^*} \cdot \underset{v_i \getr \mcal{D}}{\Pr}[v_i \geq \theta_{\tau^*} \mid v_i \geq \theta_\tau]$.
To make this deviation profitable, it must be
\begin{equation}
	\label{eq:ascending}
	\theta_{\tau^*} \cdot \underset{v_i \getr \mcal{D}}{\Pr}[v_i \geq \theta_{\tau^*}\mid v_i \geq \theta_\tau] > \theta_\tau.
\end{equation}
Re-arranging \cref{eq:ascending}, 
we obtain $\theta_{\tau^*} (1 - F(\theta_{\tau^*})) > \theta_\tau (1 - F(\theta_\tau))$.
Now, consider any $i$, we have 
\begin{align*}
	\theta_{i+1}(1 - F(\theta_{i+1})) - \theta_i(1 - F(\theta_i))
	&= (\theta_{i+1} - \theta_i) (1 - F(\theta_i)) - \theta_{i+1}f(\theta_i) \\
	&\leq (\theta_{i+1} - \theta_i) (1 - F(\theta_i)) - \theta_i f(\theta_i) \\
	&= - \phi(\theta_i) \cdot f(\theta_i).
\end{align*}
For any $i \geq \tau_0$, we have $\phi(\theta_i) \geq 0$.
Thus, $\theta(1 - F(\theta))$ is a non-increasing function for any $\theta \geq \theta_{\tau_0}$.
Because $\theta_{\tau^*} > \theta_\tau \geq \theta_{\tau_0}$,
we have $\theta_{\tau^*} (1 - F(\theta_{\tau^*})) \leq \theta_\tau (1 - F(\theta_\tau))$,
which implies keeping ascending the price above $\theta_\tau$ is not profitable in expectation.
Thus, the auction satisfies Bayesian psIC.

\paragraph{Bayesian sIC.}
If there exists a scenario where the strategic seller can profit with some strategy $S$,
the seller colludes with the platform, and the seller adopts strategy $S$ while the platform behaves honestly.
Because the platform revenue is always zero,
the joint utility of the coalition increases since the seller itself can increase its utility,
and it violates Bayesian psIC.
Thus, the auction must satisfy Bayesian sIC.

\paragraph{Approximate revenue maximization.}
In \Cref{thm:ideal2ndprice}, we showed that 
the ideal-world second-price  
satisfies revenue optimality. 
For any fixed bid vector $\bfb$, it is easy to see that 
the total number of items allocated
is the same in the ascending auction and the second price
auction of \Cref{sec:ideal2nd}. 
Further, each confirmed buyer 
pays at most ${\sf tick}(\inDom)$ less 
in the ascending auction than 
the maximum pay of any confirmed buyer in the second price auction. 
Therefore, $(k \cdot {\sf tick}(\inDom))$-approximate revenue optimality
directly follows.

\end{proof}

\newpage

\end{document}